\documentclass[a4paper,11pt]{article}

\usepackage[a4paper, hmargin=1.25in,vmargin=1.25in]{geometry}
\usepackage{amsmath,amssymb,amsthm,mathrsfs,amsfonts,mathtools,dsfont,bm} 
\usepackage[dvipsnames,svgnames, x11names]{xcolor}
\usepackage{natbib}
\usepackage{bibunits}
\usepackage{marginnote} 
\usepackage{enumitem} 
\newlist{problems}{enumerate}{1}
\setlist[problems]{label={\arabic*.}, ref={\arabic{part}.\thechapter.\arabic*}}
\usepackage{pdfpages}
\usepackage[hyphens]{url}
\usepackage[labelsep=period]{caption} 
\usepackage{subcaption}
\usepackage{environ,trimspaces}
\usepackage[breakable,most]{tcolorbox}
\tcbuselibrary{breakable}
\usepackage[titletoc]{appendix}
\usepackage{fancyhdr,setspace, titlesec,titling} 
\usepackage{booktabs,threeparttable,multirow,makecell} 
\usepackage{pdflscape,selectp,float,rotating,tabularx}
\usepackage{tikz,pgfplots} 
\pgfplotsset{compat=1.11}
\usetikzlibrary{arrows,intersections}
\usetikzlibrary{backgrounds,positioning,fit,petri}
\usetikzlibrary{mindmap,trees,calendar,external,shapes}
\usetikzlibrary{decorations.fractals,decorations.pathmorphing,shadows,shadings,patterns}
\usetikzlibrary{spy,calc} 
\usepackage[bookmarksnumbered,bookmarksopen,psdextra,unicode]{hyperref}
\hypersetup{
colorlinks=true, 
linkcolor=blue, 
urlcolor=black, 
citecolor=blue 
}
\usepackage{color}
\usepackage{tikz}
\usepackage{authblk}

\makeatletter
\renewcommand\paragraph{\@startsection{paragraph}{4}{\z@}
{-2.5ex\@plus -1ex \@minus -.25ex}
{1.25ex \@plus .25ex}
{\normalfont\normalsize\bfseries}}
\makeatother
\def\beqr{\begin{eqnarray}}
\def\eeqr{\end{eqnarray}}
\def\beqrs{\begin{eqnarray*}}
\def\eeqrs{\end{eqnarray*}}
\usepackage[sort]{cleveref}
\newtheorem{thm}{Theorem}

\newtheorem{assum}{Assumption}
\newtheorem{lemma}{Lemma}
\newtheorem{prop}{Proposition}

\newtheorem{example}{Example}
\crefname{thm}{theorem}{theorems}
\crefname{cond}{condition}{conditions}
\crefname{assum}{assumption}{assumptions}
\crefname{lemma}{lemma}{lemmas}
\crefname{prop}{proposition}{propositions}
\crefname{coro}{corollary}{corollaries}
\crefname{figure}{figure}{figures}
\Crefname{cond}{Condition}{Conditions}
\Crefname{thm}{Theorem}{Theorems}
\Crefname{assum}{Assumption}{Assumptions}
\Crefname{lemma}{Lemma}{Lemmas}
\Crefname{prop}{Proposition}{Propositions}
\Crefname{coro}{Corollary}{Corollaries}
\Crefname{figure}{Figure}{Figures}
\crefrangelabelformat{assum}{#3#1#4--#5#2#6}
\crefrangelabelformat{section}{#3#1#4--#5#2#6}
\newtheoremstyle{normalfont}{0.5em}{0.5em}{\normalfont}{}{\bfseries}{.}{0.5em}{} 
\theoremstyle{normalfont} 
\newtheorem{remark}{Remark}[section]

\def\mR{\mathbb{R}}
\def\calG{\mathcal{G}}
\def\calH{\mathcal{H}}
\def\calL{\mathcal{L}}
\def\calM{\mathcal{M}}
\def\calN{\mathcal{N}}
\def\calO{\mathcal{O}}
\def\calR{\mathcal{R}}
\def\calV{\mathcal{V}}
\def\calX{\mathcal{X}}
\def\calY{\mathcal{Y}}
\def\calZ{\mathcal{Z}}

\def\1{{\bm 1}}

\def\P{{\mathbf{P}}}

\def\0{{\bm 0}}
\def\1{{\bm 1}}

\def\indep{\bot\!\!\!\bot}
\def\nindep{\not\!\!\!\bot\!\!\!\bot}
\def\wh{\widehat}
\def\wt{\widetilde}

\newcommand{\convd}{\xrightarrow{d}}
\newcommand{\amin}{\operatornamewithlimits{arg\,min}}
\newcommand{\ainf}{\operatornamewithlimits{arg\,inf}}

\newcommand{\bigbc}[1]{\left\{#1\right\}} 
\newcommand{\pprime}{^\prime}
\newcommand{\trans}{^{\mathrm{\scriptscriptstyle T}}} 
\newcommand{\Abs}[1]{\left\vert#1\right\vert}
\allowdisplaybreaks[1]
\addtocontents{toc}{\protect\setcounter{tocdepth}{0}}
\newcommand{\boverbar}[1]{\bm{\mkern 2mu\overline{\mkern-2mu#1\mkern-2mu}\mkern 2mu}}
\newcommand{\bbar}[1]{\bm{\bar{#1}}}

\newcommand{\bunderline}[1]{\bm{\underline{#1}}}
\newcommand{\Rlogo}{\protect\includegraphics[height=1.8ex,keepaspectratio]{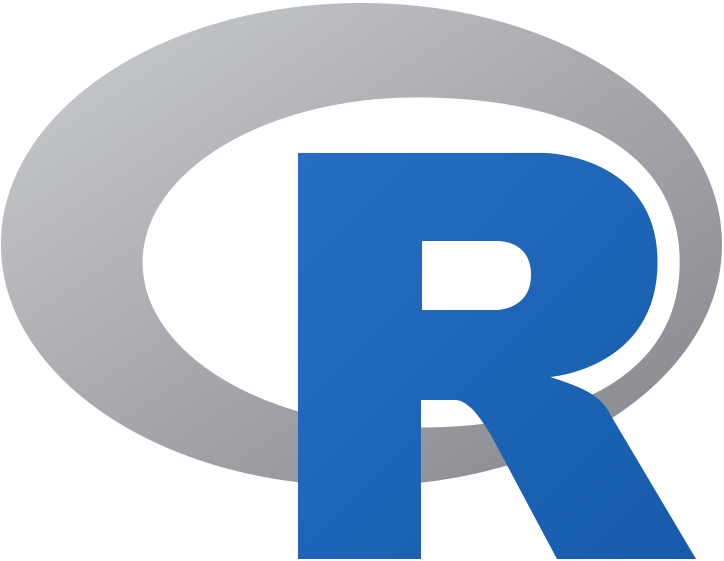}}

\defaultbibliographystyle{apalike}
\defaultbibliography{reference}

\def\papertitle{Nonparametric Estimation of Path-specific Effects in Presence of Nonignorable Missing Covariates}
\title{\papertitle}
\author[1]{Jiawei Shan\thanks{Both authors contributed equally to this work.}}
\author[2]{Ting Wang$^*$}
\author[3]{Wei Li}
\author[4]{Chunrong Ai}
\affil[1]{ Institute of Statistics and Big Data, Renmin University of China}
\affil[2,3]{ Center for Applied Statistics and School of Statistics, Renmin University of China}
\affil[4]{ School of Management and Economics, The Chinese University of Hong Kong, Shenzhen}
\affil[$~$]{\small \texttt{$^{1,2,3}$\{jwshan,ting.wang,weilistat\}@ruc.edu.cn, $^{4}$chunrongai@cuhk.edu.cn}}
\date{\today}
\usepackage{microtype}

\begin{document}

\begin{bibunit}
\pdfbookmark[1]{Title}{title}
\maketitle
\begin{abstract}
The path-specific effect (PSE) is of primary interest in mediation analysis when multiple intermediate variables between treatment and outcome are observed, as it can isolate the specific effect through each mediator, thus mitigating potential bias arising from other intermediate variables serving as mediator-outcome confounders. However, estimation and inference of PSE become challenging in the presence of nonignorable missing covariates, a situation particularly common in epidemiological research involving sensitive patient information. In this paper, we propose a fully nonparametric methodology to address this challenge. We establish identification for PSE by expressing it as a functional of observed data and demonstrate that the associated nuisance functions can be uniquely determined through sequential optimization problems by leveraging a shadow variable. Then we propose a sieve-based regression imputation approach for estimation. We establish the large-sample theory for the proposed estimator, and introduce a robust and efficient approach to make inference for PSE. The proposed method is applied to the NHANES dataset to investigate the mediation roles of dyslipidemia and obesity in the pathway from Type 2 diabetes mellitus to cardiovascular disease.
\end{abstract}

\begin{center}
\end{center}

\setcounter{page}{0}
\thispagestyle{empty}
\clearpage

\section{Introduction}
\label{sec:intro}
Mediation analysis is a powerful analytic tool to explore the mechanisms of treatment-outcome relationships in both experimental and observational studies. 
While it has traditionally been studied based on structural equation models \citep{BaronKenny1986}, the model-free concepts of natural direct and indirect effects were later formalized by \cite{RobinsGreenland1992} and \cite{Pearl2001} within the potential outcomes framework. 
This formalization separates identification assumptions from parametric constraints, enabling nonparametric identification and semiparametric inference for these causal estimands \citep{ImaiKeeleTingley2010,TchetgenTchetgenShpitser2012}.

In epidemiological studies, mediation analysis has been widely employed to explore the mediated risk factors in the development of chronic diseases, for example, mechanisms by which type 2 diabetes mellitus (T2DM) causes cardiovascular disease \citep{ sharif2019mediation, wang2022role}. 
As indicated by \cite{zhang2023triglyceride}, several intermediate variables, such as obesity and dyslipidemia, play a crucial role in the relationship between T2DM and cardiovascular events.
The conventional mediation analysis approach typically regards intermediate variables collectively as a single vector-valued mediator variable, which makes it difficult to distinguish the specific mediation role of each intermediate variable.
\cite{sharif2019mediation} examined the natural indirect effects of each risk factor separately. However, despite controlling for a set of pre-treatment covariates, their results may still be subject to potential bias if one intermediate variable induced by the treatment T2DM becomes a confounder between another intermediate variable and the cardiovascular outcome. For example, dyslipidemia is a common symptom for T2DM patients, which subsequently influences both obesity and the cardiovascular outcome \citep{vekic2019obesity}. 
This is known as a mediator-outcome confounding problem in mediation analysis \citep{AvinShpitserPearl2005}, and a remedy is to consider the path-specific effect (PSE) instead of the natural indirect effect \citep{Shpitser2013,daniel2015causal,miles2017quantifying}. 
Essentially, the PSE mitigates mediator-outcome confounding by controlling for the counterfactuals of intermediate variables preceding the mediator of interest at the same treatment status when switching the potential mediator values. 
\cite{miles2020semiparametric} proposed a semiparametric approach for estimation of PSE in presence of mediator-outcome confounding. 
\cite{zhou2022semiparametric} extended the framework to accommodate multiple mediators and developed a general methodology for identifying and estimating the PSEs.

A prevalent challenge impeding the estimation and inference of the PSE in epidemiological studies is the issue of missing data.
Sensitive information related to privacy, such as lifestyle, mental health status, socioeconomic status, and other self-reported information, are often withheld by respondents. These involved variables are sometimes crucial covariates in identifying causal effects, particularly as pre-treatment covariates to eliminate the confounding bias.
For example, in empirical research exploring the causal mechanism between T2DM and cardiovascular disease using data from the National Health and Nutrition Examination Survey (NHANES), baseline covariates such as age, gender, race, hypertension, and drinking are considered, with approximately $36\%$ of the drinking data missing \citep{zhang2023triglyceride}.
Compounding the issue, this missingness is likely nonignorable, also known as missing not at random \citep[MNAR,][]{Rubin1976}, as individuals aware of the health risks associated with heavy drinking may be more reluctant to answer related questions \citep{bartlett2014improving}.
A widely adopted approach for addressing nonignorable missingness is the use of shadow variables \citep{dHaultfoeuille2010,WangShaoKim2014,ZhaoMa2022,MiaoLiuLiTchetgenTchetgenGeng2024}. In our example, alcoholic hepatitis is considered a suitable shadow variable for recovering missing information about alcohol consumption because it is associated with long-term alcohol use; besides, it is diagnosed through physiological indicators obtained from laboratory data, which is independent of the missing process related to self-reported drinking behavior. These factors imply that alcoholic hepatitis meets the necessary criteria for a shadow variable.
Several studies have utilized auxiliary or other observed variables as shadow variables to address nonignorable missing data in causal inference \citep{YangWangDing2019} or mediation analysis \citep{LiZhou2017, ZuoGhoshDingYang2024, ShanLiAi2024}. 
However, current methods for mediation analysis with potential mediator-outcome confounding and multiple mediators subject to nonignorable missing covariates are still lacking.

In this paper, we address these gaps in the current literature by providing a comprehensive methodology for identifying, estimating and making inference about the PSE when multiple intermediate variables are observed and pre-treatment covariates are subject to nonignorable missingness. 
We first establish identification of PSE by expressing it as a functional of observed data using a weighting method. The weight is identified by an integral equation with an available shadow variable, while the other nuisance functions are demonstrated to be unique solutions to specific weighted least squares optimization problems.
Building on these identification results, we propose a sieve-based regression imputation estimator and establish its asymptotic normality
despite the well-known challenges of ill-posedness in solving the integral equation \citep{Kress_LinearIntegralEquations}.
Our estimator exhibits two main advantages over classical parametric regression imputation methods: it is more flexible in accommodating nonlinear structures as it nonparametrically models and estimates nuisance functions using versatile sieves like power series and splines; and it is shown to attain the semiparametric efficiency bound. 
We also propose a novel approach for estimating the asymptotic variance and making efficient inference for the PSE. Unlike previous methods that estimate the density appearing in the influence function separately, we identify density ratios as a whole through optimization problems and employ a one-step estimation process.

The rest of the paper is organized as follows. 
\Cref{sec:setup} defines the counterfactuals, the PSE of interest and specifies main assumptions.
\Cref{sec:estimation} derives the identifiable expression of the parameter of interest and introduces our estimation approach.
\Cref{sec:inference} derives the large-sample properties and our inference methods. 
\Cref{sec:simulation} provides a simulation study, and we apply our method to NHANES in \Cref{sec:application}.
Technical proofs and additional results are provided in the supplementary materials.

\section{Setup} \label{sec:setup}

Let $A$ denote a binary treatment, and $Y$ the outcome of interest. Let $X$ be a set of pre-treatment covariates subject to nonignorable missingness, and let the binary variable $R$ be the missingness indicator with $R=1$ if $X$ is fully observed and 0 otherwise.
Suppose we observe $K$ causally ordered clusters of intermediate variables between $A$ and $Y$, denoted by $M_1, \dots, M_K$, such that no component of $M_{k^\prime}$ causally affects any component of $M_k$ for $k<k^\prime$.
Classical mediation analysis, which views all intermediate variables as a vector-valued mediator \citep{Pearl2001}, is a special case here with $K=1$.
There are $2^K$ possible causal paths from treatment to outcome through different combinations of the $K$ mediators. For example, when $K=2$, the four paths are: (a) $A\to Y$; (b) $A\to M_1\to Y$; (c) $A\to M_2\to Y$; (d) $A\to M_1\to M_2\to Y$.
Throughout, for any generic sequence $\{v_1,v_2,\ldots\}$, we write $\bbar{v}_k=(v_1,\ldots,v_k)$, and $\bbar{v}_0=\emptyset$.
We employ the potential outcome framework to define path-specific effects.
Let $Y(a, \boverbar{m}_K)$ denote the potential outcome under treatment status $A = a$ and mediator values $\boverbar{M}_K = \boverbar{m}_K$, $M_k(a, \boverbar{m}_{k-1})$ the potential value of the mediator $M_k$ under treatment status $a$ and mediator values $\boverbar{M}_{k-1} = \boverbar{m}_{k-1}$ for $k\in\{1,\ldots,K\}$.
It allows us to define the \emph{counterfactual functionals}:
$$
\psi(\bbar{a}_{K+1}) \triangleq E\left\{Y\left(a_{K+1}, \boverbar{M}_K\left(\bbar{a}_K\right)\right)\right\},
$$
where $\boverbar{M}_k\left(\bbar{a}_k\right) =\big(\boverbar{M}_{k-1}(\bbar{a}_{k-1}), M_k(a_k, \boverbar{M}_{k-1}(\bbar{a}_{k-1}))\big)$ for $k\in\{1,\ldots,K\}$, and $a_1, \ldots, a_{K+1}$ are treatment statuses with each taking values 0 or 1. For example, when $K=2$, $\psi(\bbar{a}_3)=E\{Y(a_3,M_1(a_1),M_2(a_2,M_1(a_1))\}$.
Essentially, $\psi(\bbar{a}_{K+1})$ excludes a proportion of nested counterfactuals, for example, $E\{Y(a_3,M_1(a_1),M_2(a_2,M_1(a_1^\prime))\}$ with $a_1\neq a_1^\prime$ when $K=2$, which is typically unidentifiable according to the recanting witness criterion \citep{AvinShpitserPearl2005}. 
A set of PSEs of interest are introduced by $\psi(\bbar{a}_{K+1})$, as illustrated in the following examples.

\begin{example}[Mediation Analysis with treatment-induced mediator-outcome confounding]
\label{exam:m-o-confounding}
Suppose $K=2$ and $M_2$ is the mediator of interest, as shown in \Cref{fig:mediation analisis}. 
The natural direct and indirect effects are typically not identifiable because the mediator-outcome confounder $M_1$ induced by the treatment exists. Specifically, the natural indirect effect through $M_2$ is defined following \cite{Pearl2001} as $E\{Y(1,M_2(1))\} - E\{Y(1,M_2(0))\}$, which is equal to
\begin{align}\label{eqn:nie}
E\{Y(1,M_1(1),M_2(1,M_1(1)))\} - E\{Y(1,M_1(1),M_2(0,M_1(0)))\},
\end{align}
by the composition assumption where
$M_1$ is set to what it should be in the nested counterfactual. 
The first term in \eqref{eqn:nie} corresponds to $\psi(1,1,1)$, while the second term is typically unidentifiable because it accounts for counterfactuals under conflicting treatment values \citep{RobinsRichardson2011}. 
Referring to \cite{Shpitser2013} and \cite{miles2020semiparametric}, an alternative estimand quantifying the path-specific effect along the path $A \to M_2 \to Y$, which is of primary interest, is defined via $\psi(\bbar{a}_{K+1})$ as follows:
\begin{align*}
\text{PSE}_{A \to M_2 \to Y} = \psi(1,1,1) - \psi(1,0,1).
\end{align*}

\begin{figure}[htbp]
\centering
\centering
\includegraphics[width=0.4\textwidth]{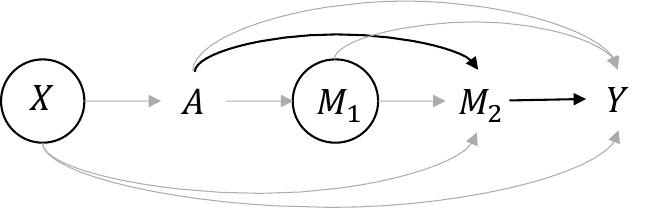}
\caption{Causal DAGs with the thick arrows indicating the PSE of interest.}
\label{fig:mediation analisis}
\end{figure}

\end{example}

\begin{example}[Mediation analysis with multiple ordered mediators]
As an extension of \Cref{exam:m-o-confounding}, 
we aim to describe the mediation effects through each intermediate variable among $M_1,\ldots,M_K$.
Referring to \cite{daniel2015causal} and \cite{zhou2022semiparametric}, the 
natural direct effect of $A$ on $Y$, and the portion of path-specific natural indirect effects that operates through the $k$th mediator $M_k$, are respectively defined as
\begin{align*}
\text{NDE}_{A \to Y} =&~ \psi(a_{K+1} = 1, \bbar{a}_K=\0) - \psi(\bbar{a}_{K+1}=\0),~~\\
\text{NIE}_{A \to M_k \rightsquigarrow Y} =&~ \psi( \bbar{a}_{k-1}=\0,\bunderline{a}_k=\1) - \psi( \bbar{a}_k=\0,\bunderline{a}_{k+1}=\1),~k\in\{1,\ldots,K\},
\end{align*}
where $\bunderline{a}_k=(a_k,\ldots,a_{K+1})$.
\Cref{fig:causal path} is a causal DAG for illustration. For instance, $\text{NIE}_{A \to M_k \rightsquigarrow Y}$ is a composite effect of the causal paths $A\to M_k\to Y$, $A\to M_k\to M_{k+1}\to Y$, $\ldots$, and $A\to M_k\to \ldots \to M_K\to Y$.
Then the total effect (TE) of $A$ on $Y$ can be decomposed as
\begin{align*}
\text{TE} \triangleq \psi(\bbar{a}_{K+1}=\1) - \psi(\bbar{a}_{K+1}=\0)
= \text{NDE}_{A \to Y} + \sum_{k=1}^{K} \text{NIE}_{A \to M_k \rightsquigarrow Y}.
\end{align*} 

\begin{figure}[htbp]
\centering
\begin{subfigure}{.3\linewidth}
\centering
\includegraphics[width=\textwidth]{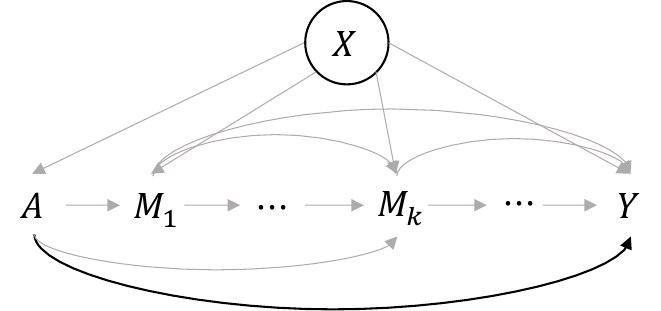}
\caption{ $\text{NDE}_{A \to Y}$. }
\end{subfigure}
\begin{subfigure}{.3\linewidth}
\centering
\includegraphics[width=\textwidth]{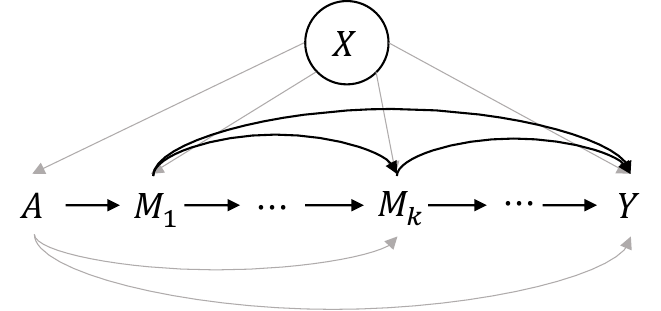}
\caption{$\text{NIE}_{A \to M_1 \rightsquigarrow Y}$.}
\end{subfigure}
\begin{subfigure}{.3\linewidth}
\centering
\includegraphics[width=\textwidth]{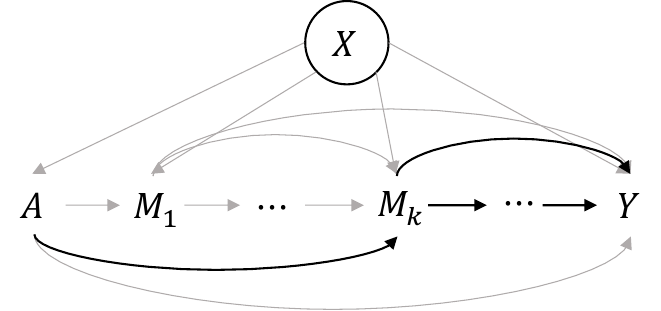}
\caption{ $\text{NIE}_{A \to M_k \rightsquigarrow Y}$.}
\end{subfigure}
\caption{Causal DAGs of some PSEs. }
\label{fig:causal path}
\end{figure}
\end{example}

In the following, we will concentrate on the estimation and inference of $\psi(\bbar{a}_{K+1})$.
The following assumptions are imposed by \cite{zhou2022semiparametric} for identifying $\psi(\bbar{a}_{K+1})$, which can be viewed as an extension of those proposed by \cite{ImaiKeeleYamamoto2010} in mediation analysis.
\begin{assum}[Consistency]\label{ass:consistency}
$M_k=M_k\left(a_k, \boverbar{m}_{k-1}\right)$ if $A=a_k$ and $\boverbar{M}_{k-1}=\boverbar{m}_{k-1}$ for $k \in\{1,\ldots,K\}$, and $Y=Y\left(a_{K+1}, \boverbar{m}_K\right)$ if $A=a_{K+1}$ and $\boverbar{M}_K=\boverbar{m}_K$.
\end{assum}

\begin{assum}[Sequential ignorability]\label{ass:ignorability}
$\left\{M_1(a_1),M_2(a_2, \boverbar{m}_1), \ldots, Y(a_{K+1}, \boverbar{m}_K)\right\} \indep A \mid X$,
and $\left\{M_{k+1}(a_{k+1}, \boverbar{m}_k), \ldots, M_K(a_K, \boverbar{m}_{K-1}), Y(a_{K+1}, \boverbar{m}_K)\right\} \indep M_k(a_k, \boverbar{m}_{k-1}^*) \mid X, A, \boverbar{M}_{k-1}$ for any $a_1,\ldots,a_{K+1}$, $m_1,\ldots,m_K$, $m_1^*,\ldots,m_{K-1}^*$, and $k \in\{1,\ldots,K\}$.
\end{assum}

\begin{assum}[Positivity]\label{ass:positivity}
$f_{A \mid X}(a \mid x)>\varepsilon>0$ whenever $f_X(x)>0 ; f_{A \mid X, \boverbar{M}_k}\left(a \mid x, \boverbar{m}_k\right)>\varepsilon>0$ whenever $f_{X, \boverbar{M}_k}\left(x, \boverbar{m}_k\right)>0$ for $k \in\{1,\ldots,K\}$.
\end{assum}

\Cref{ass:consistency,ass:positivity} are standard in causal inference literature. 
\Cref{ass:ignorability} stipulates that baseline covariates control for all potential confounding between the treatment and mediators, as well as between the treatment and the outcome. Additionally, the covariates, treatment, and preceding mediators jointly control for any confounding between subsequent mediators and between the mediators and the outcome.
Under \Crefrange{ass:consistency}{ass:positivity}, \cite{zhou2022semiparametric} showed that $\psi(\bbar{a}_{K+1})$ can be expressed as 
\begin{equation}\label{eq:iden_zhou}
\psi(\bbar{a}_{K+1})=\int_x \int_{\boverbar{m}_K} E(Y \mid x, a_{K+1}, \boverbar{m}_K)\Big\{\prod_{k=1}^K d F\left(m_k \mid x, a_k, \boverbar{m}_{k-1}\right)\Big\} d F(x).
\end{equation}

Unfortunately, the identification formula \eqref{eq:iden_zhou} is infeasible if the covariates $X$ have missing values. This becomes particularly challenging when the missingness is nonignorable, namely $R\nindep X\mid (A,\boverbar{M}_K,Y)$, because the underlying joint distribution of $(X,A,\boverbar{M}_K,Y)$ in this case is typically not identified by the observed data without additional information, even when strict parametric models are employed \citep{MiaoDingGeng2016}.
Inspired by \cite{dHaultfoeuille2010} and \cite{MiaoTchetgenTchetgen2018}, we consider the shadow variable framework for identification. Let $Z$ be a shadow variable satisfying the following assumptions. 

\begin{assum}[Shadow variable]\label{ass:shadow}
(i) $Z\nindep X\mid A,\boverbar{M}_K,Y$, ~ (ii) $Z\indep R\mid X,A,\boverbar{M}_K,Y$.
\end{assum}

\begin{assum}[Completeness]\label{ass:completeness}
For any squared-integrable function $g$ and for any $a,\boverbar{m}_K,y$, $E\{g(X)\mid Z,A\!=\!a,\boverbar{M}_K\!=\!\boverbar{m}_K,Y\!=\!y\}=0$ almost surely if and only if $g(X)=0$ almost surely.
\end{assum}

\Cref{ass:shadow} emphasizes that the shadow variable should be associated with the missing covariates conditional on the other observed data, but independent of the missingness process conditional on both observed and missing data. 
In our motivating example, alcoholic hepatitis serves as a shadow variable for the partially observed drinking behavior. 
It is expected to satisfy the conditions because one of the main pathogenesis of alcoholic hepatitis is long-term alcohol use. Furthermore, alcoholic hepatitis is diagnosed based on laboratory data, which is fully observed and expected not to directly influence the propensity of respondents reporting their alcohol consumption. 
In fact, the shadow variable is also referred to as the “nonresponse instrument” in some literature \citep{dHaultfoeuille2010,WangShaoKim2014}, and has been widely used for establishing identifiability \citep{MiaoTchetgenTchetgen2016,MiaoTchetgenTchetgen2018}.
\Cref{ass:completeness} is a technical condition relevant to identification process, and holds for a wide range of parametric or semiparametric models such as exponential families \citep{NeweyPowell2003,DHaultfoeuille2011}. Intuitively, it requires the shadow variable to be of as much variability as the missing variable. We also assume the response probability $f(R=1\mid X,A,\boverbar{M}_K,Y)>\underline{c}>0$ for some constant $\underline{c}$ to rule out the degenerate case.
Under these assumptions, the identification of the joint distribution $f(X,A,\boverbar{M}_K,Y)$ can be established similarly to the studies of, for example, \cite{DingGeng2014,YangWangDing2019}. In this paper, we always assume the underlying data distribution satisfies these conditions and focus on robust estimation and efficient inference of $\psi(\bbar{a}_{K+1})$ using observed data.

\section{Estimation} \label{sec:estimation}
To motivate our estimation approach, we first establish the identifiable expression of $\psi(\bbar{a}_{K+1})$. Specifically, we express $\psi(\bbar{a}_{K+1})$ as a functional of fully observed data and demonstrate that the associated nuisance functions can be uniquely determined through sequential optimization problems.
Then we propose a robust and efficient nonparametric approach for estimation.
For notational simplicity, we consider a given treatment status $\bbar{a}_{K+1}$ and denote $\psi\equiv\psi(\bbar{a}_{K+1})$ from now on.
Let $\mu_{K+1}(X,\boverbar{M}_K)\triangleq E(Y \mid X, A = a_{K+1}, \boverbar{M}_K)$ and $\mu_k(X,\boverbar{M}_{k-1})\triangleq E\{\mu_{k+1}(X,\boverbar{M}_k)\mid X, A = a_k, \boverbar{M}_{k-1}\}$ for $k\in\{1,\ldots,K\}$ be a series of iterated conditional expectations. Then \eqref{eq:iden_zhou} implies that $\psi=E\{\mu_{1}(X)\}$.
Let
\beqrs
\gamma(X,A,\boverbar{M}_K,Y)=\frac{f(R=0\mid X,A,\boverbar{M}_K,Y)}{f(R=1\mid X,A,\boverbar{M}_K,Y)},
\eeqrs
be the odds function.
Then it is easy to check that $\psi$ can be expressed as
\begin{align*}
\psi = E\left\{R\mu_{1}(X)+(1-R)\mu_{1}(X)\right\}= E\left\{R\mu_{1}(X)+R\gamma(X,A,\boverbar{M}_K,Y)\mu_{1}(X)\right\}.
\end{align*}
Essentially, the odds function $\gamma(X,A,\boverbar{M}_K,Y)$ summarizes the distribution shift between units with and without missingness.
Once we obtain estimators of $\gamma(X,A,\boverbar{M}_K,Y)$ and $\mu_1(X)$, an estimator of $\psi$ is given by 
\beqr \label{eq:est_theta}
\wh\psi = \frac{1}{n}\sum_{i=1}^{n}\left\{R_i\wh\mu_{1}(X_i) + R_i\wh\gamma(X_i,A_i,\boverbar{M}_{K,i},Y_i)\wh\mu_{1}(X_i)\right\}.
\eeqr
In the following, we address the identification and estimation of $\gamma$ and $\mu_1$.

We first consider $\gamma$. With an available shadow variable $Z$ satisfying \cref{ass:shadow,ass:completeness}, similar to \cite{LiMiaoTchetgenTchetgen2023} and \cite{MiaoLiuLiTchetgenTchetgenGeng2024}, we show in \Cref{sec:app_identification} that $\gamma$ is identified by the integral equation:
\beqr
\label{eq:iden_gamma}
E\left\{\gamma(X,A,\boverbar{M}_{K},Y) \mid R=1,Z,A,\boverbar{M}_{K},Y\right\} = \beta(Z,A,\boverbar{M}_{K},Y), 
\eeqr
where $\beta(Z,A,\boverbar{M}_{K},Y)\triangleq {f(R=0\mid Z,A,\boverbar{M}_K,Y)}\big/{f(R=1\mid Z,A,\boverbar{M}_K,Y)}$.
Here, \Cref{ass:shadow} ensures that $\gamma$ satisfies the integral equation \eqref{eq:iden_gamma}, and \Cref{ass:completeness} guarantees the unique solution of the equation.
Note that the integral equation \eqref{eq:iden_gamma} is equivalent to
\begin{align}
\label{eq:iden_gamma_transf}
E\left\{R\gamma(X,A,\boverbar{M}_{K},Y) \mid Z,A,\boverbar{M}_{K},Y\right\} = E(1-R\mid Z,A,\boverbar{M}_{K},Y).
\end{align}
We define the criterion function 
\begin{align*}
Q(\wt\gamma)= E \left[
\left\{E(R \wt\gamma(X,A,\boverbar{M}_{K},Y) -1+R \mid Z,A,\boverbar{M}_{K},Y) \right\}^2
\right],~~\wt\gamma\in\Gamma,
\end{align*}
where $\Gamma$ is the parameter space of $\gamma$.
Then we have $\gamma=\ainf_{\wt\gamma\in\Gamma}Q(\wt\gamma)$. Therefore, we can estimate $\gamma$ by finding a minimizer of the sample analogue $Q_n(\wt\gamma)$ of $Q(\wt\gamma)$.
For this purpose, we first estimate the unknown conditional expectation $E(\cdot\mid Z,A,\boverbar{M}_{K},Y)$ using a series estimator \citep{Newey1997}. Specifically, let $\{p_j(\cdot)\}_{j=1}^\infty$ denote a sequence of known basis functions (such as power series, splines, Fourier series, etc.), with the property that its linear combination can approximate any squared integrable real-valued function of $(z,a,\boverbar{m}_{K},y)$ well. Let $\bbar{p}_{l_n}(z,a,\boverbar{m}_{K},y)= \big(p_1(z,a,\boverbar{m}_{K},y),\ldots,p_{l_n}(z,a,\boverbar{m}_{K},y)\big)\trans$,
and $\P= \big(\bbar{p}_{l_n}(Z_1,A_1,\boverbar{M}_{K,1},Y_1),\allowbreak\ldots,\bbar{p}_{l_n}(Z_n,A_n,\boverbar{M}_{K,n},Y_n)\big)\trans$.
Then for a generic random variable $V$ with realizations $\{V_i\}_{i=1}^n$, the estimator of $E(V\mid Z,A,\boverbar{M}_{K},Y)$ is given by
\begin{align*}
&\wh{E}(V\mid Z,A,\boverbar{M}_{K},Y)=
\sum_{i=1}^n V_i\bbar{p}_{l_n}(Z_i,A_i,\boverbar{M}_{K,i},Y_i)\trans (\P\trans \P)^{-1}\bbar{p}_{l_n}(Z,A,\boverbar{M}_{K},Y).
\end{align*}
Then the sample analogue $Q_n(\wt\gamma)$ of $Q(\wt\gamma)$ is given by
\begin{align*}
Q_n(\wt\gamma) = \frac{1}{n}\sum_{i=1}^n 
\Big[\wh{E}\{R \wt\gamma(X,A,\boverbar{M}_{K},Y) -1+R \mid Z,A,\boverbar{M}_{K},Y\} \Big]^2.
\end{align*}
Without imposing parametric assumptions on the functional form of $\gamma$, its parameter space $\Gamma$ is typically an infinite-dimensional set, thus it is impractical to find a minimizer of $Q_n(\wt\gamma)$ over $\Gamma$ with finite samples.
As a remedy, we employ a sieve approximation \citep{Grenander1981} for $\Gamma$. 
Suppose $\Gamma_n$ is a sieve space that is a computable and often finite-dimensional compact parameter space that becomes dense in $\Gamma$ as $n$ increases. 
Several commonly used sieve spaces and their corresponding approximation properties are well discussed in \cite{Chen2007Handbook}.
For instance, let $\bbar{q}_{s_n}(x,a,\boverbar{m}_{K},y)$ denote a $s_n$-vector of basis functions similar to $\bbar{p}_{l_n}(z,a,\boverbar{m}_{K},y)$. Considering $\gamma$ is a probability ratio and has non-negative values, we suggest a sieve space $\Gamma_n = \{\wt\gamma:~\wt\gamma(x,a,\boverbar{m}_{K},y) = \exp[\bbar{q}_{s_n}(x,a,\boverbar{m}_{K},y)\trans\bm\pi],\bm\pi\in\mR^{s_n}\}$. Then an estimator of $\gamma$ is given by 
\beqrs\label{eq:SMD_gamma}
\wh\gamma=\amin_{\wt\gamma\in\Gamma_n} Q_n(\wt\gamma).
\eeqrs

Next we consider estimation of $\mu_1(X)$. Since $\{\mu_k(\cdot)\}_{k=1}^{K+1}$ are defined sequentially as iterated conditional expectations, we should estimate from the innermost conditional expectation $\mu_K$ and iterate outward until we reach the outermost conditional expectation $\mu_1$.
When $X$ is fully observed, \cite{zhou2022semiparametric} proposed a regression-imputation approach by fitting a parametric model for the conditional mean of $\mu_{k+1}(X,\boverbar{M}_k)$ on $(X,A,\boverbar{M}_{k-1})$ and then setting $A=a_{k+1}$ for all units. 
However, it encounters two challenges in our setting. 
First, accurately specifying parametric models for all iterated regressions is difficult, especially in presence of missing values, unless the model of $Y$ on $(A,\boverbar{M}_K)$ and the models of $\boverbar{M}_k$ on $(A,\boverbar{M}_{k-1})$ are all linear, which is somewhat restrictive. 
Second, when $X$ has missing values, it is not possible to directly apply the ordinary least squares regression. 
We remedy the first challenge by directly approximating the function space of $\mu_k$, denoted by $\Lambda_k\subset\calL_2(X,\boverbar{M}_{k-1})$, using an increasing sieve space.
To address the second issue, 
we show in \Cref{sec:app_identification} that $\mu_{K+1}(X,\boverbar{M}_K)$ is the unique solution to the following criterion function of the weighted $\calL_2$-projection problem, which involves only observed data:
\beqrs 
\inf_{\mu\in \Lambda_{K+1}} E\left[I(A=a_{K+1})R\{1+\gamma(X,A,\boverbar{M}_K,Y)\}\{Y-\mu(X,\boverbar{M}_K)\}^2\right],
\eeqrs
and $\mu_k(X,\boverbar{M}_{k-1})$ is the unique solution to 
\beqrs 
\inf_{\mu\in \Lambda_k} E\left[I(A=a_k)R\{1+\gamma(X,A,\boverbar{M}_K,Y)\}\{\mu_{k+1}(X,\boverbar{M}_k)-\mu(X,\boverbar{M}_{k-1})\}^2\right],
\eeqrs
for $k\in\{1,\ldots,K\}$. 
Therefore, we construct an estimator $\wh\mu_{K+1}(X,\boverbar{M}_K)$ by solving 
\beqrs 
\min_{\mu\in \Lambda_{K+1,n}} \frac{1}{n}\sum_{i=1}^{n}I(A_i=a_{K+1})R_i\{1+\wh\gamma(X_i,A_i,\boverbar{M}_{K,i},Y_i)\}\{Y_i-\mu(X_i,\boverbar{M}_{K,i})\}^2 ,
\eeqrs
and then sequentially from $K$ to 1, obtain estimators $\wh\mu_k(X,\boverbar{M}_{k-1})$ by solving 
\beqrs 
\min_{\mu\in \Lambda_{kn}} \frac{1}{n}\sum_{i=1}^{n}I(A_i=a_k)R_i\{1+\wh\gamma(X_i,A_i,\boverbar{M}_{K,i},Y_i)\}\{\wh\mu_{k+1}(X_i,\boverbar{M}_{k,i})-\mu(X_i,\boverbar{M}_{k-1,i})\}^2.
\eeqrs 
Here, 
$\Lambda_{kn}= \{g:
g(x,\boverbar{m}_{k-1})=\bbar{u}_k(x,\boverbar{m}_{k-1})\trans\bm\pi,\bm\pi\in\mR^{t_{kn}}
\}$ 
for $k\in\{1,\ldots,K+1\}$
is a finite-dimensional sieve space that becomes dense in $\Lambda_k$ as $n\to\infty$ for some given basis functions $\bbar{u}_k(x,\boverbar{m}_{k-1})$ as discussed in \citet[p.5569]{Chen2007Handbook}.
The estimation procedure can be easily implemented using basic computing software by the following outlined algorithm:
\begin{description}
\item[Step 1.]\label{enu:setp1} Minimize $Q_n(\wt\gamma)$ over $\Gamma_n$ and obtain $\wh\gamma$;
\item[Step 2.]\label{enu:setp2} Regress $Y_i$ on $\bbar{u}_{K+1}(X_i,\boverbar{M}_{K,i})$ with weights $I(A_i=a_{K+1})R_i\{1+\wh\gamma(X_i,A_i,\boverbar{M}_{K,i},Y_i)\}$ and obtain $\wh\mu_{K+1}(X_i,\boverbar{M}_{K,i})$;
\item[Step 3.]\label{enu:setp3} Regress $\wh\mu_{k+1}(X_i,\boverbar{M}_{k,i})$ on $\bbar{u}_k(X_i,\boverbar{M}_{k-1,i})$ for $k$ from $K$ to $1$ iteratively with weights $I(A_i=a_k)R_i\{1+\wh\gamma(X_i,A_i,\boverbar{M}_{K,i},Y_i)\}$ and obtain $\wh\mu_{k}(X_i,\boverbar{M}_{k-1,i})$;
\item[Step 4.]\label{enu:setp4} Obtain $\wh\psi$ from equation \eqref{eq:est_theta}.
\end{description} Notably, step 2--3 can be implemented using, for example, the ``lm'' function in \Rlogo. 
We remark here that in the special case without data missingness, by simply setting $\wh\gamma=0$ in step 1, the algorithm remains valid as a nonparametric regression-imputation approach for estimation of $\psi$.
Moreover, if $\mu_k(x,\boverbar{m}_{k-1})$ is assumed to follow a specific parametric form, one only needs to find the estimator $\wh\mu_k$ by solving the above optimization problem over this parametric space. Particularly, if $\wh\gamma\equiv 0$ and $\Lambda_{kn}$'s are fixed-dimensional spaces spanned by linear functions of $(x,\boverbar{m}_{k-1})$, the proposed estimator reduces to the regression imputation estimator introduced by \cite{zhou2022semiparametric}.

\section{Large-sample properties and inference}\label{sec:inference}
\subsection{Asymptotic normality} 

In this section, we derive the asymptotic distribution of $\wh\psi$.
We first introduce some notations. For a given $\bbar{a}_{K+1}$, we define
\begin{subequations}
\begin{align}
\omega_1(x) \triangleq&~ f_{A\mid X}(a_{1}\mid x)^{-1}, \label{eq:def_w1}\\ 
\omega_k(x,\boverbar{m}_{k-1}) \triangleq &~ \frac{f_{A\mid X,\boverbar{M}_{k-1}}(a_{k-1}\mid x,\boverbar{m}_{k-1})}{f_{A\mid X,\boverbar{M}_{k-1}}(a_k\mid x,\boverbar{m}_{k-1})},~~k=2,\ldots,K+1, \label{eq:def_wk}\\
\phi(x,a,\boverbar{m}_K,y) \triangleq &~ \mu_{1}(x) + \sum_{k=1}^{K}I(a=a_k) \Big\{\prod_{j=1}^k \omega_j(x,\boverbar{m}_{j-1})\Big\}
\Big\{\mu_{k+1}(x,\boverbar{m}_k)-\mu_k(x,\boverbar{m}_{k-1})\Big\} \nonumber\\ 
&\quad + I(a=a_{K+1}) \Big\{\prod_{j=1}^{K+1} \omega_j(x,\boverbar{m}_{j-1})\Big\}
\Big\{y-\mu_{K+1}(x,\boverbar{m}_K)\Big\}. \label{eq:def_phi}
\end{align}
\end{subequations}
Note that $\omega_k$ and $\phi$ vary with different choice of $\bbar{a}_{K+1}$, here we omit the notations $\bbar{a}_{K+1}$ for simplicity. \cite{zhou2022semiparametric} showed that when $X$ is fully observed, the efficient influence function of $\psi$ is $\phi(X,A,\boverbar{M}_K,Y)-\psi$. 

The first challenge in establishing asymptotic normality arises from ill-posed nature of the inverse problem in estimating $\gamma$, resulting in a slow convergence rate of $\wh\gamma$. 
We therefore introduce a pseudo metric that is weaker than conventional $\calL_2$ and $\calL_\infty$ metrics. It helps expedite the convergence of $\widehat{\gamma}$ to the true value, which is necessary in establishing asymptotic normality. We define an inner product
\begin{align*}
\langle \gamma_1,\gamma_2 \rangle_w = 
E\big[E\{R\gamma_1(X,A,\boverbar{M}_K,Y)\mid Z,A,\boverbar{M}_K,Y\}E\{R\gamma_2(X,A,\boverbar{M}_K,Y)\mid Z,A,\boverbar{M}_K,Y\}\big],
\end{align*}
and then define $\|\gamma_1\|_w=\langle \gamma_1,\gamma_1 \rangle_w^{1/2}$ for any $\gamma_1,\gamma_2\in\Gamma$. It is easy to see $\|\gamma_1\|_w\le\|\gamma_1\|_{\calL_2}$ for any $\gamma_1$ by Jensen's inequality. The convergence results of $\widehat{\gamma}$ and $\widehat{\mu}_k$ for $k \in \{1, \ldots, K+1\}$ are formally established in \Cref{sec:app_convergence_rate}. Importantly, we demonstrate that the convergence rates of $\widehat{\gamma}$ under the pseudo-metric $\|\cdot\|_w$, as well as $\widehat{\mu}_k$ under the $\calL_2$ and $\calL_\infty$ metrics, are all $o_p(n^{-1/4})$ under standard regularity conditions for nonparametric regression. 

Another challenge in obtaining the asymptotically linear representation of $\wh\psi$ is that $\wh\gamma$ typically lacks a closed form due to the complexity of $\Gamma_n$ with regularization.
To address this, we introduce a representer to capture the cumulative influence of $\wh\gamma$ on $\wh\psi$ and derive its linear representation.
Let $\overline\calH$ be the closure of the linear span of $\Gamma$ under the metric $\|\cdot\|_w$, then $\overline\calH$ is a Hilbert space with inner product $\langle \cdot,\cdot \rangle_w$ defined above. Note that the linear functional $\wt\gamma\mapsto E\{R\phi(X,A,\boverbar{M}_K,Y)\wt\gamma(X,A,\boverbar{M}_K,Y)\}$ is continuous under $\|\cdot\|_w$, by Riesz representation theorem, there is $\varrho\in\overline\calH$ such that 
\beqrs 
\langle \varrho,\wt\gamma \rangle_w
= E\big\{ R \phi(X,A,\boverbar{M}_K,Y)\wt\gamma(X,A,\boverbar{M}_K,Y) \big\},
\eeqrs
holds for all $\wt\gamma\in\overline\calH$. We impose the following assumption.

\begin{assum}\label{ass:representer}
Assume that $\varrho\in\Gamma$, or there is $\varrho_n\in\Gamma_n$ such that $\|\varrho_n-\varrho\|_w=o(n^{-1/4})$.
\end{assum}
\Cref{ass:representer} requires the Riesz representer $\varrho$ to fall in $\Gamma$, or be well approximated by the sieve space $\Gamma_n$. It is a technical condition that suffices for $\sqrt{n}$-estimability of $\psi$. It is commonly imposed and thoroughly discussed in the literature regarding the estimation of functionals involving ill-posed inverse problems \citep{AiChen2003,Santos2011,SeveriniTripathi2012}.
The asymptotic normality of $\wh\psi$ is formally established in the following theorem and is proved in \Cref{sec:app_pf_thm_asy_normal}.

\begin{thm}\label{thm:asy_normal}
Suppose \Crefrange{ass:consistency}{ass:representer} and \ref{ass:supp_dgp}--\ref{ass:supp_technique_mu} in \Cref{sec:app_assum} hold. We obtain 
\beqrs
\sqrt{n}(\wh\psi-\psi) = \frac{1}{\sqrt{n}}\sum_{i=1}^{n} IF(R_i,Z_i,X_i,A_i,\boverbar{M}_{K,i},Y_i) + o_p(1),
\eeqrs
with the influence function 
\begin{align*}
IF =&~ R\big\{1+\gamma(X,A,\boverbar{M}_K,Y)\big\}\phi(X,A,\boverbar{M}_K,Y) -\psi \\ 
&\quad- E\{R\varrho(X,A,\boverbar{M}_K,Y)\mid Z,A,\boverbar{M}_K,Y\}\{R\gamma(X,A,\boverbar{M}_K,Y)-1+R\}.
\end{align*}
It follows that $\sqrt{n}(\wh\psi-\psi)\convd N(0,\sigma^2)$ with $\sigma^2=E(IF^2)$.
\end{thm}
\Cref{thm:asy_normal} demonstrates that $\wh\psi$ exhibits $\sqrt{n}$-convergence and asymptotic normality, unaffected by the slow convergence rate of $\wh\gamma$. This result is the foundation for making inferences about the PSE, which will be discussed in the next subsection.
Below we briefly discuss the efficiency of the proposed estimator in the following two remarks.

\begin{remark}
If all covariates $X$ are observed, that is $R\equiv 1$ and $\gamma\equiv 0$, the influence function of the proposed estimator $\wh\psi$ turns out to be $IF=\phi(X,A,\boverbar{M}_K,Y) -\psi$, which is exactly the efficient influence functions of $\psi$ in the full data law derived in \cite{zhou2022semiparametric}. 
Therefore, our approach also provides an efficient estimation of $\psi$ when the data are completely observed by simply setting $\wh\gamma=0$.
\end{remark}

\begin{remark}\label{remark:eif}
We show in \Cref{ssec:pf_EIF} that the proposed estimator $\wh\psi$ attains the semiparametric efficiency bound in the semiparametric model $\calM_{sp}$ which places no restriction on the observed data distribution other than the model restriction \eqref{eq:iden_gamma} at key submodels. This reflects that our approach effectively leverages the information of observed data. 
\end{remark}

\subsection{Variance estimation and inference} \label{ssec:variance}
We consider estimating the asymptotic variance $\sigma^2$ and constructing confidence intervals for $\psi$. To achieve this, we need to estimate the unknown functions that appear in the influence function of $\psi$ as described in \Cref{thm:asy_normal}.
For estimation of $\phi(\cdot)$, we need to estimate the unknown functions $\omega_k(x,\boverbar{m}_{k-1})$ for $k\in\{1,\ldots,K+1\}$ defined in \eqref{eq:def_w1}--\eqref{eq:def_wk}. The following lemma provides their identifiable formulas and is proved in \Cref{ssec:app_pf_lemma1}.
\begin{lemma}\label{lemma:iden_coef}
We have that $\omega_1(x)$ satisfies the following conditional moment restriction
\begin{align*}
E\big[\big\{1+\gamma(X,A,\boverbar{M}_K,Y)\big\}\big\{I(A=a_1)\omega_1(X)-1\big\}\mid R=1,X\big] =0, 
\end{align*}
and for $k\in\{2,\ldots,K+1\}$, $\omega_k(x,\boverbar{m}_{k-1})$ satisfies the following conditional moment restriction
\begin{align*}
E\big[\big\{1+\gamma(X,A,\boverbar{M}_K,Y)\big\}\big\{I(A=a_k)\omega_k(X,\boverbar{M}_{k-1})-I(A=a_{k-1})\big\}\mid R=1,X,\boverbar{M}_{k-1}\big] =0. 
\end{align*}
\end{lemma}
Similar to the approach used for estimating $\gamma$ from \eqref{eq:iden_gamma_transf}, we can solve the integral equations in \Cref{lemma:iden_coef} by minimizing the corresponding optimization problems over finite-dimensional sieve spaces. Specifically, an estimator $\wh\omega_1(x)$ can be obtained by solving 
\begin{align*}
\min_{\wt\omega_1\in\Lambda_{1n}}\sum_{i=1}^n R_i\wh{E}\big[\big\{1+\wh\gamma(X,A,\boverbar{M}_K,Y)\big\}\big\{I(A=a_1)\wt\omega_1(X)-1\big\}\mid R=1,X_i\big]^2,
\end{align*}
and for $k\in\{2,\ldots,K+1\}$, estimators $\wh\omega_k(x,\boverbar{m}_{k-1})$ are obtained by solving 
\begin{align*}
\min_{\wt\omega_k\in\Lambda_{kn}} \sum_{i=1}^n R_i \wh{E}\Big[\big\{1+\wh\gamma(X,A,\boverbar{M}_K,Y)\big\} & \big\{I(A=a_k)\wt\omega_k(X,\boverbar{M}_{k-1}) \\ 
&\quad -I(A=a_{k-1})\big\}\mid R=1,X_i,\boverbar{M}_{k-1,i}\Big]^2 ,
\end{align*}
respectively. Then we obtain an estimator of $\phi(\cdot)$ as follows:
\begin{align}\label{eq:est_phi}
\wh\phi(x,a,\boverbar{m}_K,y) = &~ \wh\mu_{1}(x) + \sum_{k=1}^{K}I(a=a_k) \Big\{\prod_{j=1}^k \wh\omega_j(x,\boverbar{m}_{j-1})\Big\}
\Big\{\wh\mu_{k+1}(x,\boverbar{m}_k)-\wh\mu_k(x,\boverbar{m}_{k-1})\Big\} \nonumber\\ 
&\quad + I(a=a_{K+1}) \Big\{\prod_{j=1}^{K+1} \wh\omega_j(x,\boverbar{m}_{j-1})\Big\}
\Big\{y-\wh\mu_{K+1}(x,\boverbar{m}_K)\Big\}.
\end{align}
It is worth noting that \eqref{eq:est_phi} represents a general formula applicable to any $\bbar{a}_{K+1}$. However, in specific cases where $a_k=a_{k-1}$, \eqref{eq:def_wk} implies that $\omega_k(x,\boverbar{m}_{k-1})\equiv 1$, allowing for a simplification of \eqref{eq:est_phi} without the need of $\wh\omega_k$ and $\wh\mu_k$.
Moreover, unlike the method proposed by \cite{zhou2022semiparametric}, which estimates the densities appearing in the numerator and denominator of $\omega_k$ separately, we propose a one-step approach that directly estimates $\omega_k$, resulting in greater stability.
Next, we consider the estimation of the representer $\varrho$ in \Cref{ass:representer}. It can be observed that $\varrho$ is a solution to the optimization problem
\beqrs\label{eq:criterion_varrho}
\inf_{\wt\varrho\in\Gamma} \frac{1}{2} E\left[\left\{E\big(R\wt\varrho(X,A,\boverbar{M}_K,Y)\mid Z,A,\boverbar{M}_K,Y\big)\right\}^2\right] - E\left\{ R\phi(X,A,\boverbar{M}_K,Y)\wt\varrho(X,A,\boverbar{M}_K,Y) \right\}.
\eeqrs
It suggests an estimator $\wh\varrho$ that solves
\begin{align*}
\min_{\wt\varrho\in\Gamma_n} \frac{1}{2n}\sum_{i=1}^{n} & \left[\wh{E}\{R\wt\varrho(X,A,\boverbar{M}_K,Y)\mid Z_i,A_i,\boverbar{M}_{K,i},Y_i\}\right]^2 \\ 
& - \frac{1}{n}\sum_{i=1}^{n} R_i\wh\phi(X_i,A_i,\boverbar{M}_{K,i},Y_i)\wt\varrho(X_i,A_i,\boverbar{M}_{K,i},Y_i).
\end{align*}
By \Cref{thm:asy_normal}, one can then construct an estimator of the asymptotic variance $\sigma^2$ by 
\begin{align*}
\wh\sigma^2 &= 
\frac{1}{n}\sum_{i=1}^{n} \wh{IF}_i^2 \triangleq 
\frac{1}{n}\sum_{i=1}^{n} \Big[
R_i\big\{1+\wh\gamma(X_i,A_i,\boverbar{M}_{K,i},Y_i)\big\}\wh\phi(X_i,A_i,\boverbar{M}_{K,i},Y_i) - \wh\psi \\ 
&\qquad- \wh{E}\{R\wh\varrho(X,A,\boverbar{M}_K,Y)\mid Z_i,A_i,\boverbar{M}_{K,i},Y_i\}\{R_i\wh\gamma(X_i,A_i,\boverbar{M}_{K,i},Y_i)-1+R_i\} \Big]^2 .
\end{align*}
Then the 95\% symmetric confidence interval for $\psi$ is given by $[\wh\psi-1.96\wh\sigma/\sqrt{n},\wh\psi+1.96\wh\sigma/\sqrt{n}].$

The above procedure can be slightly modified to make inference for PSE. As an illustration, an estimator of $\text{NDE}_{A \to Y}$ is given by $\wh{\text{NDE}}_{A \to Y} = \wh\psi(a_{K+1} = 1, \bbar{a}_K=\0) - \wh\psi(\bbar{a}_{K+1}=\0)$. Then \Cref{thm:asy_normal} implies that 
\begin{align*}
\sqrt{n}\big(\wh{\text{NDE}}_{A \to Y}-\text{NDE}_{A \to Y}\big) \convd N(0,\sigma^2_{\text{NDE}}),
\end{align*}
where $\sigma^2_{\text{NDE}}=E[\{IF(a_{K+1}=1,\bbar{a}_K=\0)-IF(\bbar{a}_{K+1}=\0)\}^2]$, and $IF(\bbar{a}_{K+1})$ is the influence function of $\wh\psi(\bbar{a}_{K+1})$ obtained in \Cref{thm:asy_normal} for a given treatment status $\bbar{a}_{K+1}$. The asymptotic variance $\sigma^2_{\text{NDE}}$ is estimated by 
\begin{align*}
\wh\sigma_{\text{NDE}}^2 = \frac{1}{n}\sum_{i=1}^{n} \left\{\wh{IF}_i(a_{K+1} = 1, \bbar{a}_K=\0)-\wh{IF}_i(\bbar{a}_{K+1}=\0)\right\}^2, 
\end{align*}
and the confidence interval can be constructed following standard procedures. 

\section{Simulation} \label{sec:simulation}
In this section, we conduct simulation studies to demonstrate the finite-sample performance of the proposed estimators. We generate independent and identically distributed samples through the following data generating mechanism. 
Let $(\varepsilon_1, \ldots, \varepsilon_5)\sim N(0,\bm{I}_5)$ be the random noise.
Let $\Phi$ be the cumulative distribution function of the standard normal distribution. We generate the component of covariates subject to nonignorable missingness by $X_1 \sim \Phi(\alpha\varepsilon_{1}+\sqrt{1-\alpha^2}\varepsilon_{2})$ with $\alpha=0.6$, and the shadow variable by $Z \sim \Phi(\varepsilon_1)$. 
Notably, the completeness assumption \ref{ass:completeness} holds following Theorem 2.1 of \cite{DHaultfoeuille2011}.
We also generate fully observed covariates $X_2 \sim \text{Uniform}(0,1)$ and $X_3 \sim \text{Bernoulli}(0.5)$. The treatment $A$ is generated from a Bernoulli distribution such that $A \sim \text{Bernoulli}\{\text{expit}(-0.1+X_1-X_2+0.2X_3)\}$, where $\text{expit}(u)=\exp(u)/\{1+\exp(u)\}$. We generate two mediators $(M_1,M_2)$ as follows: 
\begin{align*}
M_1 & = -1 + 0.5A -2\sin(X_1) + 3X_{1}^2 - 2X_2 + X_3 + \varepsilon_{3},\\
M_2 & = 1 - 0.5A + X_1 + X_{2}^2 - X_3 - 0.5AM_1 + \varepsilon_{4}.
\end{align*}
The outcome is generated by $Y = -1 + 0.5A - 1.5M_1 + 1.5M_2 + 3X_1 + 3X_{1}^2 - 3\sin(X_{2}) + X_{2}^2 - X_3 + 0.5AM_1 + 0.5AM_2 +\varepsilon_{5}$.
Here, the models of the mediators and outcome are nonlinear in covariates and include interaction terms to account for heterogeneity, making them more representative of complex real-world scenarios.
Finally, the missingness indicator of $X_1$ is generated by
$$
R \sim \text{Bernoulli}\{\text{expit}(0.1-2X_1+1.5X_2+X_3+0.5A\varepsilon_{3}-0.5A\varepsilon_{4}-0.1\varepsilon_{5})\}.
$$
It is worth noting that the true model of the log-odds function $\log \gamma$ also includes nonlinear and interaction terms implicitly through $(\varepsilon_{3},\varepsilon_{4},\varepsilon_{5})$.
The missing proportion of $X_1$ is approximately 43.7\%. 

We consider the following four methods for comparison: (a) Oracle, an ideal scenario that utilizes the underlying true data; (b) SRI (short for sieve-based regression imputation), the proposed estimator using polynomials of order 3 as the sieve basis. 
(c) MI, a multiple imputation estimator that fills in the missing values by generating plausible values according to the observed data distribution \citep{LittleRubin2002}; and (d) CCA, a complete-case analysis that restricts the analysis to only those units that have no missing values. For all these methods, we generate 1000 simulations of sample size 1000 and 2000.
For simplicity, we use $\text{NIE}_{M_k}$ to denote $\text{NIE}_{A \to M_k \rightsquigarrow Y}$ in the following.

Table \ref{tab:simulation} presents the bias, standard error, and $95\%$ coverage probability of different estimators, and \Cref{fig:bias} displays the boxplots over 1000 replications. The results indicate that the proposed estimator, which effectively uses the shadow variable to make up for missing data, is comparable to the oracle estimator, both having negligible biases.
In contrast, both MI and CCA estimators are significantly biased as expected. 
Meanwhile, the proposed estimator achieves coverage probabilities close to the nominal level, while both MI and CCA estimators demonstrate poor coverage due to their substantial bias.
This suggests that the multiple imputation method, although easy to implement and widely used in practice to address missing data problems, may fail and introduce misleading inference when the missingness is nonignorable. 

\begin{figure} 
\centering
\includegraphics[width=1\linewidth]{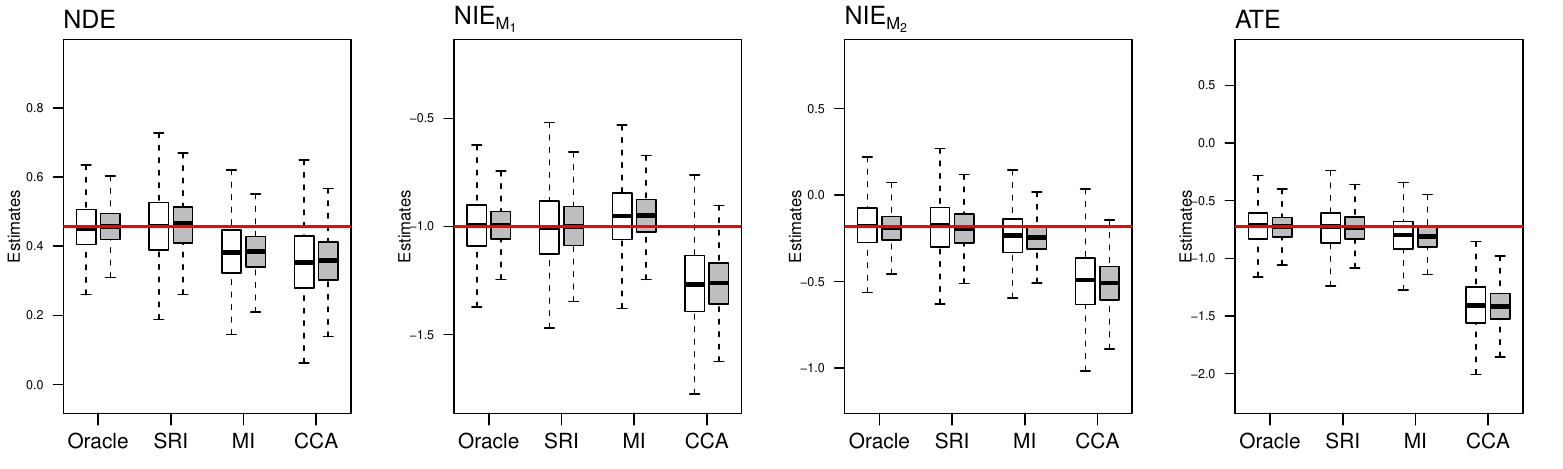}
\caption{ Boxplots of various estimators across 1000 replications. White boxes represent a sample size of 1000, while grey boxes indicate a sample size of 2000. The horizontal lines indicate the true values. }
\label{fig:bias}
\end{figure}

\begin{table}
\caption{ Boxplots of various estimators across 1000 replications. White boxes represent a sample size of 1000, while grey boxes indicate a sample size of 2000. The horizontal lines indicate the true values. }
\label{tab:simulation}
\begin{center}
\begin{tabular}{ccccccccccc}
\toprule
& & \multicolumn{4}{c}{$n=1000$}& & \multicolumn{4}{c}{$n=2000$} \\
\cline{3-6} 	\cline{8-11} 
Methods & & $\text{NDE}$ & $\text{NIE}_{M_1}$ & $\text{NIE}_{M_2}$ & $\text{TE}$ & & $\text{NDE}$ & $\text{NIE}_{M_1}$ & $\text{NIE}_{M_2}$ & $\text{TE}$\\
\midrule
\multirow{3}{*}{Oracle} 
&Bias & -0.003 & 0.001 & 0.007 & 0.005 & & 0.000 & 0.005 & -0.008 & -0.003\\
&SE & 0.075 & 0.137 & 0.149 & 0.170 & & 0.055 & 0.097 & 0.101 & 0.124\\
&CP & 0.954 & 0.962 & 0.959 & 0.947 & & 0.941 & 0.950 & 0.958 & 0.951\\
\midrule
\multirow{3}{*}{SRI}
&Bias & 0.002 & -0.008 & 0.000 & -0.006 & & 0.006 & -0.001 & -0.013 & -0.009\\
&SE & 0.105 & 0.178 & 0.177 & 0.195 & & 0.076 & 0.135 & 0.122 & 0.138\\
&CP & 0.943 & 0.970 & 0.966 & 0.954 & & 0.959 & 0.950 & 0.977 & 0.967\\
\midrule
\multirow{3}{*}{MI}
&Bias & -0.073 & 0.044 & -0.049 & -0.077 & & -0.071 & 0.049 & -0.064 & -0.086\\
&SE & 0.088 & 0.155 & 0.147 & 0.180 & & 0.066 & 0.112 & 0.100 & 0.129\\
&CP & 0.764 & 0.917 & 0.915 & 0.903 & & 0.654 & 0.901 & 0.896 & 0.877\\
\midrule
\multirow{3}{*}{CCA}
&Bias & -0.103 & -0.270 & -0.311 & -0.684 & & -0.100 & -0.266 & -0.328 & -0.695\\
&SE & 0.112 & 0.189 & 0.198 & 0.222 & & 0.079 & 0.133 & 0.142 & 0.164\\
&CP & 0.829 & 0.738 & 0.670 & 0.148 & & 0.708 & 0.502 & 0.380 & 0.011\\
\bottomrule
\end{tabular}
\end{center}
\end{table}

\section{Application to NHANES}
\label{sec:application}
Type 2 diabetes mellitus (T2DM) is a chronic condition with high global prevalence, increasing the risk of cardiovascular disease (CVD) morbidity and mortality. 
We apply the proposed methodology to the National Health and Nutrition Examination Survey (NHANES) with the aim of elucidating the mediation roles of dyslipidemia and obesity in the pathway from T2DM to CVD.
NHANES collects data through interviews, physical examinations, and laboratory testing to provide comprehensive health-related information, including chronic diseases, dietary and lifestyle habits, as well as a wide range of public health-related topics. For this study, we select 2670 individuals aged 20-60 from the NHANES 2013-2014 dataset.
According to the American Diabetes Association, the binary treatment T2DM is diagnosed by insulin or oral hypoglycemic use, a fasting blood glucose level $\geq 126$ mg/dL, or an HbA1c level $\geq 6.5\%$. The binary outcome CVD is diagnosed by conditions such as congestive heart failure, coronary heart disease, angina pectoris, myocardial infarction, or stroke.
As discussed in the introduction, we denote dyslipidaemia measured by blood lipid levels as $M_1$ and obesity measured by waist circumference as $M_2$. 
Additionally, we consider five baseline covariates: age, gender, race, hypertension and drinking status. 
Drinking status is categorized either heavy drinker (consuming $\geq 3$ drinks per day for females and $\geq 4$ drinks per day for males) or mild drinker/nondrinker, based on self-reported information from a standardized health questionnaire.
Among the 2670 participants, approximately $35.8\%$ of the drinking status data are missing. 
The missingness in the drinking variable is primarily dependent on its value \citep{bartlett2014improving}, as individuals aware of the negative health impacts of heavy drinking are more likely to withhold their responses. It is plausible that the missingness of covariate drinking is not at random.We select alcoholic hepatitis as the shadow variable $Z$, where $Z=1$ indicates a positive diagnosis based on laboratory measurements, and $Z=0$ otherwise.
As previously discussed, alcoholic hepatitis is strongly associated with heavy drinking, and its symptoms often improve significantly with abstinence. 
Additionally, alcoholic hepatitis is diagnosed based on laboratory data, which is not expected to directly influence the reporting of drinking status.

We calculate various PSEs using the methods mentioned in simulation studies, and the results are shown in \Cref{tab:application}. 
The $95\%$ confidence intervals for $\text{NIE}_{M_1}$ and $\text{NIE}_{M_2}$ do not include $0$, which means both dyslipidaemia and obesity are significant mediators in the causal pathway linking T2DM and CVD. Additionally, the mediating effect through 
dyslipidaemia amounts for $5.8\%$ of the TE, while the effect through obesity
amounts for $8.9\%$. This indicates that obesity plays a slightly more significant mediating role in the pathway from T2DM to CVD. 
The direct effect of T2DM accounts for $85.3\%$ of the TE. Although a significant proportion of residual cardiovascular risk remains due to T2DM itself, it is important to note that about $15\%$ of this risk is attributable to dyslipidemia and obesity. This suggests that about $15\%$ relative risk reduction could be achieved by managing weight and lowering blood lipid levels. The results align closely with the established conclusions on cardiovascular outcomes in T2DM.
To emphasize the potential influence of mediator-outcome confounding, we also calculate the results by treating dyslipidemia and obesity as separate mediators, as performed by \cite{sharif2019mediation}. The natural indirect effect of dyslipidemia is estimated at $0.862\times 10^{-2}$, amounting for about $5.8\%$ of TE, while the natural indirect effect of obesity is $1.407 \times 10^{-2}$, amounting for $9.5\%$ of the TE. The estimate for dyslipidaemia aligns with the PSE as expected, whereas the estimate for obesity appears to be slightly overestimated compared to the PSE. This indicates that ignoring dyslipidemia when evaluating the mediation effect of obesity may introduce slight confounding bias. This finding is consistent with the theoretical results discussed earlier.
Compared with the SRI method, both MI and CCA methods underestimate the PSEs to varying degrees. 
This indicates that assuming the missingness of drinking behavior to be completely at random or at random is, to some extent, unreasonable, thereby lending greater credibility to the proposed SRI method in this study.

\begin{table}[h] 
\caption{Point estimates, standard error and $95\%$ confidence intervals of $\text{NIE}_{M_1}$, $\text{NIE}_{M_2}$, NDE and TE. All entries are in units of $10^{-2}$.}
\label{tab:application}
\begin{center}
\begin{tabular}{ccccc}
\toprule
Estimands & Methods & Point Estimates &Standard Error & $95\%$ Confidence Intervals \\
\midrule
\multirow{3}*{$\text{NIE}_{M_1}$} & SRI & 0.860 & 0.266 & [0.338,1.382] \\
& MI & 0.639 & 0.189 & [0.269,1.009] \\
& CCA & 0.677 & 0.214 & [0.258,1.096] \\ 
\midrule
\multirow{3}*{$\text{NIE}_{M_2}$} & SRI& 1.319 & 0.631 & [0.082,2.556] \\ 
& MI & 0.818 & 0.376 & [0.081,1.554] \\ 
& CCA & 1.015 & 0.509 & [0.017,2.013] \\ 
\midrule
\multirow{3}*{$\text{NDE}$} & SRI & 12.602 & 2.835 & [7.045,18.159] \\ 
& MI & 10.312 & 1.954 & [6.483,14.141] \\ 
& CCA & 10.276 & 2.542 & [5.294,15.258] \\ 
\midrule
\multirow{3}*{$\text{TE}$} & SRI & 14.781 & 2.757 & [9.377,20.185]\\ 
& MI & 11.769 & 1.908 & [8.028,15.509] \\ 
& CCA & 11.968 & 2.536 & [6.998,16.938] \\ 
\bottomrule
\end{tabular}
\end{center}
\end{table}

\section{Discussion}
\label{s:discuss}
In this paper, we develop a methodology for estimating and making inferences about path-specific effects when multiple intermediate variables are involved and pre-treatment covariates are subject to nonignorable missingness. This issue is commonly encountered in epidemiological research and related dataset.
Our method utilizes an auxiliary shadow variable, typically obtainable from related laboratory data, to recover information on missing covariates.
Our method has several attractive features: (i) it accommodates a wide range of nonparametric models, allowing for complex nonlinear structures and thereby reducing the risk of model misspecification; (ii) our estimation procedure is robust and easy to implement using standard statistical software; (iii) we offer an alternative approach for estimating asymptotic variance that avoids the instability caused by inverse probability.

Some possible directions for future research remain open. The proposed estimation is based on the regression-imputation formula of $\psi$, which requires relatively stringent conditions on the consistency and convergence rate of each nuisance estimator. To alleviate these conditions, a potential approach is to develop semiparametric estimation that leverages the benefits of Neyman orthogonality and incorporates state-of-the-art machine learning techniques. However, it faces some challenging problems, for example, how to construct orthogonal scores, and how to employ machine learning to solving the integral equation.
Additionally, other estimands instead of nature direct and indirect effects are of interest if there is a mediator–outcome confounder affected by the treatment, such as interventional direct and indirect effects \citep{VanderweeleVansteelandtRobins2014}. These estimands have distinct identification formulas and also face the challenge of missing data in real-world studies. These problems will be explored in future work.

\section*{Supplementary Material}
The supplementary material includes additional assumptions, derivation of identifiable expressions of $\gamma$, $\mu_k$ and $\psi$ in \Cref{sec:estimation}, intermediate results including precise convergence rates of nuisance estimators, discussion about the efficiency of the proposed estimator, and all technical proofs of lemmas and theorems.
\addcontentsline{toc}{section}{References}
\putbib
\end{bibunit}

\clearpage
\begin{bibunit}
\setcounter{page}{1}
\counterwithin{thm}{section}
\counterwithin{lemma}{section}
\counterwithin{assum}{section}
\counterwithin{equation}{section}
\counterwithin{table}{section}
\titleformat{\section}[block]{\normalfont\Large\filcenter}{{\sc\appendixname}~\thesection}{1em}{}
\title{Supplementary Material for ``\papertitle''}
\pdfbookmark[1]{Supplementary Material}{title1}
\emptythanks
\maketitle
\appendix

\begin{center}
\textsc{Outline} 
\end{center}
\Cref{sec:app_assum} includes additional technical assumptions used in this paper. 
\Cref{sec:app_identification} provides details about identifiablity of $\gamma$, $\mu_k$ and $\psi$.
\Cref{sec:app_convergence_rate} presents important intermediate results regarding the convergence rates of the nuisance estimators $\wh\gamma$ and $\wh\mu_k$ for $k\in\{1,\ldots,K+1\}$ under various metrics, which are crucial for establishing asymptotic normality. 
\Cref{sec:app_pf_thm_asy_normal} provides the proof of \Cref{thm:asy_normal} regarding asymptotic normality, discusses the efficiency of $\wh\psi$, and gives the proof of \Cref{lemma:iden_coef} in the main text. 
Finally, \Cref{sec:app_Aux_lemmas} contains auxiliary lemmas used in the proofs along with their corresponding proofs.

\addtocontents{toc}{\protect\setcounter{tocdepth}{3}}

\clearpage
\section{Additional assumptions}\label{sec:app_assum}
The following assumptions are required for theoretic derivations.

\begin{assum}\label{ass:supp_dgp}
(i) $\psi\in\Psi$ that is compact, and for any $\wt\gamma\in\Gamma$ and $a\in\{0,1\}$, $\wt\gamma(x,a,\boverbar{m}_K,y)\in C_B^{\alpha}(\calX\times\boverbar\calM_K\times\calY)$ for some $B>0$ and {$\alpha>(d_x+K+1)/2$} with $d_x$ the dimension of $X$; 
(ii) $\calZ\times\calX\times\boverbar\calM_K\times\calY$ is compact and convex with nonempty interior. 
\end{assum}
\begin{assum}\label{ass:supp_sieve_delta_inter}
(i) $\Gamma_n$ is a nonempty closed convex subset of $\Gamma$; 
(ii) There is an operator $\Pi_n$ such that for any $\wt\gamma\in\Gamma$, $\Pi_n\wt\gamma\in\Gamma_n$ satisfying that $\|\Pi_n\wt\gamma- \wt\gamma\|_{\infty}=O(s_n^{-\alpha_1})$ with $\alpha_1>0$.
\end{assum}

\begin{assum}\label{ass:supp_sieve_delta_exter}
(i) The eigenvalues of 

$\P\trans \P$ is bounded above and bounded away from zero for all $l_n$; 
(ii) There is ${\bm\pi}_{\wt\gamma}\in\mR^{l_n}$ such that $\sup_{(z,a,\boverbar{m}_K,y)\in\calZ\times\{0,1\}\times\boverbar\calM_K\times\calY} |E(R\wt\gamma\mid z,a,\boverbar{m}_K,y)-\bbar{p}_{l_n}(z,a,\boverbar{m}_K,y)\trans{\bm\pi}_{\wt\gamma}| = O(l_n^{-\alpha_2})$ uniformly over $\wt\gamma\in\Gamma$ with $\alpha_2>0$; 
(iii) $\zeta_{1n}^2l_n/n=o(1)$ and $\zeta_{1n}l_n^{-\alpha_2}=o(1)$, where $\zeta_{1n}\triangleq \sup_{(z,a,\boverbar{m}_K,y)\in\calZ\times\{0,1\}\times\boverbar\calM_K\times\calY} \|\bbar{p}_{l_n}(z,a,\boverbar{m}_K,y)\|$; (iv) $l_n\ge s_n$.
\end{assum}
Recall that
\begin{align*}
\Lambda_{kn}=\left\{g:
g(x,\boverbar{m}_{k-1})=\bbar{u}_k(x,\boverbar{m}_{k-1})\trans\bm\pi = 
\textstyle\sum_{j=1}^{t_{kn}} \pi_j u_{kj}(x,\boverbar{m}_{k-1}),~\bm\pi\in\mR^{t_{kn}}
\right\},
\end{align*}
is used in estimating $\mu_k(x,\boverbar{m}_{k-1})$ for $k\in\{1,\ldots,K+1\}$.

\begin{assum}\label{ass:supp_sieve_mu}
For $k\in\{1,\ldots,K+1\}$, we assume: 
(i) The eigenvalue of the $t_{kn}$-by-$t_{kn}$ matrix $E\{\bbar{u}_k(X,\boverbar{M}_{k-1})\bbar{u}_k(X,\boverbar{M}_{k-1})\trans\}$ is bounded above and bounded away from zero for all $t_{kn}$; 
(ii) There is ${\bm\pi}_{1k}\in\mR^{t_{kn}}$ such that $\sup_{(x,\boverbar{m}_{k-1})\in\calX\times\boverbar\calM_{k-1}} \big|\mu_k(x,\boverbar{m}_{k-1})-\bbar{u}_k(x,\boverbar{m}_{k-1})\trans{\bm\pi}_{1k}\big| = O(t_{kn}^{-\alpha_{3k}})$ with $\alpha_{3k}>0$;
(iii) There is ${\bm\pi}_{2k}\in\mR^{t_{kn}}$ such that $\sup_{(x,\boverbar{m}_{k-1})\in\calX\times\boverbar\calM_{k-1}} \big|\prod_{j=1}^k\omega_j(x,\boverbar{m}_{j-1})-\bbar{u}_k(x,\boverbar{m}_{k-1})\trans{\bm\pi}_{2k}\big| = O(t_{kn}^{-\alpha_{3k}})$; 
(iv) $\zeta_{2kn}^2t_{kn}/n=o(1)$ and $\zeta_{2kn}t_{kn}^{-\alpha_{3k}}=o(1)$, where $\zeta_{2kn}\triangleq \sup_{(x,\boverbar{m}_{k-1})\in\calX\times\boverbar\calM_{k-1}} \|\bbar{u}_k(x,\boverbar{m}_{k-1})\|$.
\end{assum}
\begin{assum}\label{ass:supp_sieve_strengthened_rate}
(i) $s_n^{-2\alpha_1} + \zeta_{1n}^2(l_n/n+l_n^{-2\alpha_2})=o(n^{-1/2})$; (ii) $\zeta_{2kn}^2\sqrt{t_{kn}^2/n}=o(1)$ and $\sqrt{n}t_{kn}^{-\alpha_{3k}}=o(1)$, for $k\in\{1,\ldots,K+1\}$.
\end{assum}

Let $\varepsilon_{K+1} \triangleq I(A=a_{K+1})\{Y-\mu_{K+1}(X,\boverbar{M}_K)\}$ and $\varepsilon_k \triangleq I(A=a_k)\{\mu_{k+1,\bbar{a}_{k+1}}(X,\boverbar{M}_k)-\mu_k(X,\boverbar{M}_{k-1})\}$ for $k\in\{1,\ldots,K\}$ denote a sequence of residual variables with finite second order moments. Let $h^*_{k,j}$, $k\in\{1,\ldots,K+1\}$, $j\in\{1,\ldots,t_{kn}\}$ be solutions to 
\begin{align*}
\inf_{h_{k,j}\in\calL_2(X,A,\boverbar{M}_{K},Y)}& E\left[\{1+R\varepsilon_k u_{k,j}(X,\boverbar{M}_{k-1})h_{k,j}(X,A,\boverbar{M}_{K},Y)\}^2\right] \\ 
&\qquad + E\left[E\{Rh_{k,j}(X,A,\boverbar{M}_{K},Y)\mid Z,A,\boverbar{M}_{K},Y\}^2\right].
\end{align*}

\begin{assum}\label{ass:supp_technique_mu}
For $k\in\{1,\ldots,K+1\}$ and $j\in\{1,\ldots,t_{kn}\}$, we assume:\\
(i) $E\{1+R\varepsilon_k u_{k,j}(X,\boverbar{M}_{k-1})h^*_{k,j}(X,A,\boverbar{M}_{K},Y)\}>0$; (ii) $h^*_{k,j}\in\Gamma$, or there exists $h^*_{k,nj}\in\Gamma_n$ such that $\|h^*_{k,j}-h^*_{k,nj}\|_\infty=O(s_n^{-\alpha_1})$.
\end{assum}
We give a brief discussion of above assumptions. 
\Cref{ass:supp_dgp} imposes the compact parameter spaces. Specifically, \Cref{ass:supp_dgp}(i) requires that each function $\wt\gamma\in\mathbf{\Gamma}$ is sufficiently smooth and bounded, and is compact under $\|\cdot\|_\infty$. It is a commonly imposed condition in the nonparametric and semiparametric literature to deal with the ill-posed problem arising from discontinuity of the solution to the integral equation. 
\Cref{ass:supp_dgp}(ii) requires the support of $(Z,X,\boverbar{M},Y)$ to be bounded. In fact, this assumption can be relaxed if we restrict the tail behavior of the joint distribution of $(Z,X,\boverbar{M},Y)$, which will be illustrated in our simulation studies.
\Cref{ass:supp_sieve_delta_inter} specifies $\Gamma_n$ and its approximation error. The approximation properties of common sieves are already known in the literature to satisfy this assumption, see \cite{HiranoImbensRidder2003,Chen2007Handbook} for detailed discussions.
\Cref{ass:supp_sieve_delta_exter}(i)-(iii) and \Cref{ass:supp_sieve_mu} are typical conditions in the use of series estimators for conditional mean functions and we discuss then later together. \Cref{ass:supp_sieve_delta_exter}(iv) states the relationship between the sieve sizes $l_n$ and $s_n$, which guarantees the feasibility of the empirical optimization problem in obtaining $\wh\gamma$. 
\Cref{ass:supp_sieve_mu}(i) essentially ensures the sieve approximation estimator is nondegenerate, which is common in sieve regression literature \citep[see e.g.][]{Newey1997}. \Cref{ass:supp_sieve_mu}(ii)(iii) requires the sieve approximation error of $\mu_k$ and cumulative product of $\omega_k$ to shrink at a polynomial rate, which holds for a wide range of basis functions \citep[see e.g.][p.5573]{Chen2007Handbook}.
The polynomial rate typically depends positively on the smoothness of the estimand and negatively on the dimension of arguments. \Cref{ass:supp_sieve_mu}(iv) is a rate requirement to ensure the convergence of the series estimation. \cite{Newey1997} shows that $\zeta_{2kn}$ is $O(t_{kn})$ if $\bbar{u}_k(x,\boverbar{m}_{k-1})$ is a power series, and is $O(\sqrt{t_{kn}})$ if it is a B-spline.
\Cref{ass:supp_sieve_strengthened_rate} introduces further constraints on the smoothness parameters $(s_n,l_n,t_{kn})$, strengthening the conditions in \Cref{ass:supp_sieve_delta_inter}(ii) and \ref{ass:supp_sieve_delta_exter}(iii) and \Cref{ass:supp_sieve_mu}(iv), respectively. 
Specifically, in (i), the rate $s_n^{-\alpha_1}$ signifies the approximation error of the sieve space $\Gamma_n$ for $\Gamma$, and this restriction is satisfied under mild conditions \citep{Chen2007Handbook}. The rate $\zeta_{1n}^2(l_n/n+l_n^{-2\alpha_2})$ corresponds to the squared convergence rate of series estimation in a supremum norm \citep{Newey1997}. Similarly (ii) is the rate constraints of the variance ($\zeta_{2kn}^2t_{kn}/n$) in a supremum norm and the approximation error ($t_{kn}^{-\alpha_{3k}}$) in series estimation. 
\Cref{ass:supp_technique_mu} are regularity conditions for obtaining $\sqrt{n}$-convergence of plug-in sieve minimum distance estimators of specific functionals \citep{AiChen2007}. The $\sqrt{n}$-estimability of linear/nonlinear functionals when the nonparametric component depends on endogenous variables are also discussed in \cite{SeveriniTripathi2012} and \cite{AiChen2012}.

\section{Identification}\label{sec:app_identification}
In this section we provide details about identifiablity of $\gamma$, $\mu_k$ and $\psi$.
Recall the definition of parameters in the main text:
\begin{align*}
&\gamma(X,A,\boverbar{M}_K,Y)\triangleq\frac{f(R=0\mid X,A,\boverbar{M}_K,Y)}{f(R=1\mid X,A,\boverbar{M}_K,Y)}, \\
&\mu_{K+1}(X,\boverbar{M}_K)\triangleq E(Y \mid X, A = a_{K+1}, \boverbar{M}_K), \\
&\mu_k(X,\boverbar{M}_{k-1})\triangleq E\{\mu_{k+1}(X,\boverbar{M}_k)\mid X, A = a_k, \boverbar{M}_{k-1}\},~~ k\in\{1,\ldots,K\}.
\end{align*}

We first show that under \Crefrange{ass:shadow}{ass:completeness}, $\gamma(X,A,\boverbar{M}_K,Y)$ is the unique solution to the following integral equation with respect to the observed law only:
\beqr \label{eq:app_iden-1-1}
E\left\{R\gamma(X,A,\boverbar{M}_{K},Y) \mid Z,A,\boverbar{M}_{K},Y\right\} = E(1-R \mid Z,A,\boverbar{M}_{K},Y),
\eeqr
which is equivalent to Equation \eqref{eq:iden_gamma} in the main text because $E(1-R \mid Z,A,\boverbar{M}_{K},Y)=f(R=0\mid Z,A,\boverbar{M}_K,Y)$ and 
\begin{align*}
& E\left\{R\gamma(X,A,\boverbar{M}_{K},Y) \mid Z,A,\boverbar{M}_{K},Y\right\} \\ 
&=~f(R=1\mid Z,A,\boverbar{M}_K,Y)E\left\{\gamma(X,A,\boverbar{M}_{K},Y) \mid R=1,Z,A,\boverbar{M}_{K},Y\right\} .
\end{align*}
Note that \eqref{eq:app_iden-1-1} can be written as
\beqrs
E\left[R\left\{1+\gamma(X,A,\boverbar{M}_{K},Y)\right\} \mid Z,A,\boverbar{M}_{K},Y\right] = 1.
\eeqrs
Suppose $\widetilde{\gamma}(X,A,\boverbar{M}_{K},Y)$ is a solution of \eqref{eq:app_iden-1-1}, it satisfies that
\begin{align*}
1=&~ E\left[R\{1+\widetilde{\gamma}(X,A,\boverbar{M}_K,Y)\}\mid Z,A,\boverbar{M}_K,Y\right] \\ 
=&~ E\left[E(R\mid Z,X,A,\boverbar{M}_K,Y)\{1+\widetilde{\gamma}(X,A,\boverbar{M}_K,Y)\}\mid Z,A,\boverbar{M}_K,Y\right] \\
=&~ E\left[E(R\mid X,A,\boverbar{M}_K,Y)\{1+\widetilde{\gamma}(X,A,\boverbar{M}_K,Y)\}\mid Z,A,\boverbar{M}_K,Y\right] \\ 
=&~ E\left[\{1+\gamma(X,A,\boverbar{M}_K,Y)\}^{-1}\{1+\widetilde{\gamma}(X,A,\boverbar{M}_K,Y)\}\mid Z,A,\boverbar{M}_K,Y\right],
\end{align*}
where the third equality is by \cref{ass:shadow}(ii). Then the completeness \cref{ass:completeness} implies that $\wt\gamma(X,A,\boverbar{M}_K,Y)=\gamma(X,A,\boverbar{M}_K,Y)$ almost surely. Therefore, $\gamma$ is the unique solution to \eqref{eq:app_iden-1-1} and hence identified by the observe data.

Next, we establish identifiability of $\mu_{K+1}(X,\boverbar{M}_K)$ by demonstrating that it is the unique solution to the optimization problem
\beqr \label{eq:app_iden_uK+1-1}
\inf_{\mu\in \calL_2(X,\boverbar{M}_K)} E\left[I(A=a_{K+1})R\{1+\gamma(X,A,\boverbar{M}_K,Y)\}\{Y-\mu(X,\boverbar{M}_K)\}^2\right].
\eeqr
Since \eqref{eq:app_iden_uK+1-1} is a convex optimization problem, suppose $\wt{\mu}(X,\boverbar{M}_K)$ is the solution to \eqref{eq:app_iden_uK+1-1}. Then $\wt\mu$ satisfies the first-order condition:
\beqrs
E\left[I(A=a_{K+1})R\{1+\gamma(X,A,\boverbar{M}_K,Y)\}\{Y-\wt{\mu}(X,\boverbar{M}_K)\} \mu(X,\boverbar{M}_K)\right] = 0,~~\forall\mu\in \calL_2(X,\boverbar{M}_K).
\eeqrs
It follows that
\beqrs
E\left[I(A=a_{K+1})R\{1+\gamma(X,A,\boverbar{M}_K,Y)\}\{Y-\wt{\mu}(X,\boverbar{M}_K)\} \mid X,\boverbar{M}_K\right] = ~0.
\eeqrs
Hence, 
\begin{align*}
\wt{\mu}(X,\boverbar{M}_K)=&~ \frac{E[YI(A=a_{K+1})R\{1+\gamma(X,A,\boverbar{M}_K,Y)\}\mid X, \boverbar{M}_K]}{E[I(A=a_{K+1})R\{1+\gamma(X,A,\boverbar{M}_K,Y)\}\mid X, \boverbar{M}_K]}\\
=&~ \frac{E[Y\{1+\gamma(X,A,\boverbar{M}_K,Y)\}\mid R=1,X,A=a_{K+1},\boverbar{M}_K]}{E\{1+\gamma(X,A,\boverbar{M}_K,Y)\mid R=1,X,A=a_{K+1},\boverbar{M}_K\}}.
\end{align*}
Note that 
\begin{align}
\label{eq:app_thm1-2-1}
& E[Y\{1+\gamma(X,A,\boverbar{M}_K,Y)\}\mid R=1,X,a_{K+1},\boverbar{M}_K] \nonumber\\ 
=&~ \int y \frac{f(y\mid R=1,X,a_{K+1},\boverbar{M}_K)}{f(R=1\mid X,a_{K+1},\boverbar{M}_K,y)} dy \nonumber\\
=&~ f(R=1\mid X,a_{K+1},\boverbar{M}_K)^{-1} \int y \frac{f(R=1,y\mid X,a_{K+1},\boverbar{M}_K)}{f(R=1\mid X,a_{K+1},\boverbar{M}_K,y)} dy \nonumber\\
=&~ f(R=1\mid X,a_{K+1},\boverbar{M}_K)^{-1} \int y f(y\mid X,a_{K+1},\boverbar{M}_K) dy \nonumber\\
=&~ f(R=1\mid X,a_{K+1},\boverbar{M}_K)^{-1} E(Y\mid X,a_{K+1},\boverbar{M}_K), 
\end{align}
and similarly,
\begin{align}\label{eq:app_thm1-2-2}
E\{1+\gamma(X,A,\boverbar{M}_K,Y)\mid R=1,X,a_{K+1},\boverbar{M}_K\} = f(R=1\mid X,a_{K+1},\boverbar{M}_K)^{-1}.
\end{align}
Combining \eqref{eq:app_thm1-2-1} and \eqref{eq:app_thm1-2-2} yields that 
\begin{align*}
\wt{\mu}(X,\boverbar{M}_K)=&~ \frac{E[Y\{1+\gamma(X,A,\boverbar{M}_K,Y)\}\mid R=1,X,A=a_{K+1},\boverbar{M}_K]}{E\{1+\gamma(X,A,\boverbar{M}_K,Y)\mid R=1,X,A=a_{K+1},\boverbar{M}_K\}} \\
=&~ E(Y\mid X,a_{K+1},\boverbar{M}_K) \\
=&~ \mu_{K+1}(X,\boverbar{M}_K),
\end{align*}
where the last equality is by the definition of $\mu_{K+1}$. It follow that $\mu_{K+1}(X,\boverbar{M}_K)$ is identified by \eqref{eq:app_iden_uK+1-1}.

Similarly we shall demonstrate that $\mu_k(X,\boverbar{M}_{k-1})$ is the unique solution to 
\begin{align} \label{eq:app_iden_uk-1}
\inf_{\mu\in\calL_2(X,\boverbar{M}_{k-1})} E\left[I(A=a_k)R\{1+\gamma(X,A,\boverbar{M}_K,Y)\}\{\mu_{k+1}(X,\boverbar{M}_k)-\mu(X,\boverbar{M}_{k-1})\}^2\right],
\end{align}
for $k\in\{1,\ldots,K\}$. 
Suppose $\wt{\mu}_k(X,\boverbar{M}_{k-1})$ is a solution to \eqref{eq:app_iden_uk-1}, then the first-order condition implies that $\wt\mu_k$ satisfies
\beqrs
E\left[I(A=a_k)R\{1+\gamma(X,A,\boverbar{M}_{k-1},Y)\}\{\mu_{k+1}(X,\boverbar{M}_k)-\wt{\mu}_k(X,\boverbar{M}_{k-1})\} \mid X,\boverbar{M}_{k-1}\right]=0.
\eeqrs
It follows that 
\begin{align*}
\wt{\mu}_k(X,\boverbar{M}_{k-1})=&~ \frac{E[\mu_{k+1}(X,\boverbar{M}_k)I(A=a_k)R\{1+\gamma(X,A,\boverbar{M}_K,Y)\}\mid X, \boverbar{M}_{k-1}]}{E[I(A=a_k)R\{1+\gamma(X,A,\boverbar{M}_K,Y)\}\mid X, \boverbar{M}_{k-1}]}\\
=&~ \frac{E[\mu_{k+1}(X,\boverbar{M}_k)\{1+\gamma(X,A,\boverbar{M}_K,Y)\}\mid R=1,X,A=a_k,\boverbar{M}_{k-1}]}{E\{1+\gamma(X,A,\boverbar{M}_K,Y)\mid R=1,X,A=a_k,\boverbar{M}_{k-1}\}}.
\end{align*}
Note that
\begin{align}
\label{eq:app_thm1-3-1}
& E\left[\mu_{k+1}(X,\boverbar{M}_k)\{1+\gamma(X,A,\boverbar{M}_K,Y)\}\mid R=1,X,a_k,\boverbar{M}_{k-1}\right] \nonumber\\ 
=&~ f(R=1\mid X,a_k,\boverbar{M}_{k-1})^{-1} \iint \mu_{k+1}(X,m_k,\boverbar{M}_{k-1}) \frac{f(m_k,y, R=1\mid a_k,X)}{f(R=1\mid X,a_k,m_k,\boverbar{M}_{k-1})} dm_kdy \nonumber\\
=&~ f(R=1\mid X,a_k,\boverbar{M}_{k-1})^{-1} \iint \mu_{k+1}(X,m_k,\boverbar{M}_{k-1}) f(m_k,y\mid X,a_k,\boverbar{M}_{k-1}) dm_kdy \nonumber\\
=&~ f(R=1\mid X,a_k,\boverbar{M}_{k-1})^{-1} \int \mu_{k+1}(X,m_k,\boverbar{M}_{k-1}) f(m_k\mid X,a_k,\boverbar{M}_{k-1}) dm_k \nonumber \\
=&~ f(R=1\mid X,a_k,\boverbar{M}_{k-1})^{-1}E\{\mu_{k+1}(X,\boverbar{M}_k)\mid X, a_k, \boverbar{M}_{k-1}\},
\end{align}
and
\beqr\label{eq:app_thm1-3-2}
E\{1+\gamma(X,A,\boverbar{M}_K,Y)\mid R=1,X,a_k,\boverbar{M}_{k-1}\} =
f(R=1\mid X,a_k,\boverbar{M}_{k-1})^{-1}.
\eeqr
Combining \eqref{eq:app_thm1-3-1} and \eqref{eq:app_thm1-3-2}, we can obtain that 
\begin{align*}
\wt{\mu}_k(X,\boverbar{M}_{k-1})=&~ \frac{E[\mu_{k+1}(X,\boverbar{M}_k)\{1+\gamma(X,A,\boverbar{M}_K,Y)\}\mid R=1,X,A=a_k,\boverbar{M}_{k-1}]}{E\{1+\gamma(X,A,\boverbar{M}_K,Y)\mid R=1,X,A=a_k,\boverbar{M}_{k-1}\}} \\
=&~ E\{\mu_{k+1}(X,\boverbar{M}_k)\mid X, A = a_k, \boverbar{M}_{k-1}\} \\
=&~\mu_k(X,\boverbar{M}_{k-1}),
\end{align*}
for $k\in\{1,\dots,K\}$. It follows that $\mu_k(X,\boverbar{M}_{k-1})$ is identified by \eqref{eq:app_iden_uk-1}.

Finally, note that
\beqrs 
\psi = E\{\mu_1(X)\} 
= E\{R\mu_1(X)\} + E\{(1-R)\mu_1(X)\},
\eeqrs
and
\begin{align*}
E\{(1-R)\mu_1(X) \}
&=~ E\{E(1-R \mid X,A,\boverbar{M}_K,Y)\mu_1(X) \} \\
&=~ E\{f(R=0\mid X,A,\boverbar{M}_K,Y)\mu_1(X)\}\\
&=~ E\{f(R=1\mid X,A,\boverbar{M}_K,Y)\gamma(X,A,\boverbar{M}_K,Y)\mu_1(X) \}\\
&=~ E\{E(R \mid X,A,\boverbar{M}_K,Y)\gamma(X,A,\boverbar{M}_K,Y)\mu_1(X) \}\\
&=~E\{R\gamma(X,A,\boverbar{M}_K,Y)\mu_1(X)\}.
\end{align*}
It follows that 
\beqrs 
\psi 
= E\{R\mu_1(X) + R\gamma(X,A,\boverbar{M}_K,Y)\mu_1(X)\}.
\eeqrs
\section{Convergence rates}\label{sec:app_convergence_rate}
In this section we establish convergence rates of nuisance estimators $\wh\gamma$ and $\wh\mu_k$ for $k\in\{1,\ldots,K+1\}$ under various metrics (stated in \Cref{thm:weak_rate_delta,thm:convergence_rate_gamma_eta} later), which are necessary for establishing the asymptotic normality of $\wh\psi$.
Recall that 
\begin{align*}
\langle \gamma_1,\gamma_2 \rangle_w = 
E\big[E\{R\gamma_1(X,A,\boverbar{M}_K,Y)\mid Z,A,\boverbar{M}_K,Y\}E\{R\gamma_2(X,A,\boverbar{M}_K,Y)\mid Z,A,\boverbar{M}_K,Y\}\big],
\end{align*}
and $\|\gamma_1\|_w=\langle \gamma_1,\gamma_1 \rangle_w^{1/2}$ for any $\gamma_1,\gamma_2\in\Gamma$.
We use the calligraphic letter $\calV$ to denote the support of $V$ and $\calL_2(V)\equiv\{g:\calV\mapsto\mR:~\int_\calV \Abs{g(v)}^2f_V(v)dv<\infty\}$. 
For any function $g\in\calL_2(V)$, the $\calL_2$ metric is denoted as $\|g\|_{\calL_2}= (\int_\calV \Abs{g(v)}^2f_V(v)dv )^{1/2}$, and the $\calL_{\infty}$ (supremum) metric is denoted as $\|g\|_\infty=\sup_{v\in\calV}\Abs{g(v)}$. 
\subsection{Convergence rate of $\wh\gamma$}

The consistency of $\wh\gamma$ under $\|\cdot\|_\infty$ and the convergence rate under $\|\cdot\|_w$ are established in the following theorem.

\begin{thm}\label{thm:weak_rate_delta}
Suppose \Crefrange{ass:supp_dgp}{ass:supp_sieve_delta_exter} hold. We have 
\beqrs
\|\wh\gamma-\gamma\|_\infty=o_p(1)
~~\mbox{and}~~
\|\wh\gamma-\gamma\|_w= O_p(\sqrt{l_n/n}+l_n^{-\alpha_2}+s_n^{-\alpha_1}).
\eeqrs
Additionally, suppose \Cref{ass:supp_sieve_strengthened_rate} holds, we have $\|\wh\gamma-\gamma\|_w=o_p(n^{-1/4})$.
\end{thm}

\begin{proof}
The proof of consistency of $\wh\gamma$ under the supremum norm $\|\cdot\|_\infty$ proceeds by verifying the conditions of Theorem 1 of \cite{NeweyPowell2003}. Assumptions 1-5 of \cite{NeweyPowell2003} are satisfied by \Cref{ass:completeness}, \Cref{ass:supp_dgp}(i), \Cref{ass:supp_sieve_delta_inter}(ii) and \Cref{ass:supp_sieve_delta_exter}(ii). Then applying Theorem 4.1 in \cite{NeweyPowell2003}, we obtain 
\beqrs
\|\wh\gamma-\gamma\|_{\infty} = o_p(1).
\eeqrs

We denote $\gamma_n=\Pi_n\gamma\in\Gamma_n$. By the definition of $\wh\gamma$, we have 
\begin{align}
\label{eq:app_weak_rate_delta_pf1}
\frac{1}{n}\sum_{i=1}^n \big\{\wh{E}(R\wh\gamma-1+R\mid Z_i,A_i,\boverbar{M}_{K,i},Y_i)\big\}^2 \le \frac{1}{n}\sum_{i=1}^n \big\{\wh{E}(R\gamma_n-1+R\mid Z_i,A_i,\boverbar{M}_{K,i},Y_i)\big\}^2.
\end{align}
By \Cref{lemma:lemmaC2_of_ChenPouzo2012}, \eqref{eq:app_weak_rate_delta_pf1} implies that there exist constants $C,C\pprime>0$, such that
\begin{align}
\label{eq:app_weak_rate_delta_pf2}
&C E [\{E(R\wh\gamma-1+R\mid Z,A,\boverbar{M}_K,Y)\}^2] \nonumber\\ 
&\quad\le C\pprime E[\{E(R\gamma_n-1+R\mid Z,A,\boverbar{M}_K,Y)\}^2] + O_p(l_n/n+l_n^{-2\alpha_2}).
\end{align}
By \Cref{ass:supp_sieve_delta_inter}(ii), we have 
\begin{align}
\label{eq:app_weak_rate_delta_pf3}
&E[\{E(R\gamma_n-1+R\mid Z,A,\boverbar{M}_K,Y)\}^2] \nonumber\\ 
&\quad= E[(E[R(\gamma_n-\gamma)\mid Z,A,\boverbar{M}_K,Y])^2] \le \|\gamma_n-\gamma\|_\infty^2 = O(s_n^{-2\alpha_1}).
\end{align}
By the definition of $\|\cdot\|_w$, we have
\begin{align}
\label{eq:app_weak_rate_delta_pf4}
E [\{E(R\wh\gamma-1+R\mid Z,A,\boverbar{M}_K,Y)\}^2] = \|\wh\gamma-\gamma\|_w^2.
\end{align}
Combining \eqref{eq:app_weak_rate_delta_pf2}, \eqref{eq:app_weak_rate_delta_pf3} and \eqref{eq:app_weak_rate_delta_pf4} yields that 
\beqrs
\|\wh\gamma-\gamma\|_w^2 = O_p(l_n/n+l_n^{-2\alpha_2}+s_n^{-2\alpha_1}).
\eeqrs
\end{proof}

\subsection{Convergence rate of $\wh\mu_k$}

The convergence rates of $\wh\mu_k$ for $k\in\{1,\ldots,K+1\}$ under $\|\cdot\|_{\calL_2}$ and $\|\cdot\|_\infty$ are established in the following theorem.
\begin{thm}\label{thm:convergence_rate_gamma_eta}
Suppose \Crefrange{ass:supp_dgp}{ass:supp_sieve_mu} and \ref{ass:supp_technique_mu} in the supplementary material hold. If $\max\{s_n^{-2\alpha_1},l_n^{-2\alpha_2},l_n/n\}=o(n^{-1/2})$, $0<\lambda_k=o(n^{-1/2})$ and $\max\{t_{kn}/n,t_{kn}^{-2\alpha_{3k}}\}=o(\lambda_k)$ for $k\in\{1,\ldots,K+1\}$, then we have
\begin{align*}
\|\wh\mu_{K+1}-\mu_{K+1}\|_{\calL_2} = O_p\big(\sqrt{\lambda_{K+1}}\big),\quad
\|\wh\mu_{K+1}-\mu_{K+1}\|_\infty = O_p\big(\zeta_{2(K+1)n}\sqrt{\lambda_{K+1}}\big),
\end{align*}
and for $k\in\{1,\ldots,K\}$,
\begin{align*}
\|\wh\mu_k-\mu_k\|_{\calL_2} = O_p\big(\sqrt{\lambda_k}+\textstyle{\sum_{j=k+1}^{K+1}}t_{jn}^{-\alpha_{3j}}\big),\quad
\|\wh\mu_k-\mu_k\|_\infty = O_p\Big\{\zeta_{2kn}\big(\sqrt{\lambda_k}+\textstyle{\sum_{j=k+1}^{K+1}}t_{jn}^{-\alpha_{3j}}\big)\Big\}.
\end{align*}

\end{thm}

\begin{proof}

We first calculate the convergence rate of $\wh\mu_{K+1}$.
For notational simplicity, throughout the proof of this theorem we use the abbreviation $\alpha_3$ and $t_{n}$ to denote $\alpha_{3,K+1}$ and $t_{K+1,n}$ when no confusion arises. Let $T\triangleq I(A=a_{K+1})$ and $T_i\triangleq I(A_i=a_{K+1})$. 
By \Cref{ass:supp_sieve_mu}(ii), there exists $\mu_{K+1,n}(x,\boverbar{m}_K)=\bbar{u}_{K+1}(x,\boverbar{m}_K)\trans {\bm\pi}_{K+1}$ with ${\bm\pi}_{K+1}\in\mR^{t_n}$ such that 
\beqr\label{eq:supp_def_pi_gamma}
\sup_{(x,\boverbar{m}_K)\in\calX\times\boverbar\calM_K} \Abs{\mu_{K+1}(x,\boverbar{m}_K) - \mu_{K+1,n}(x,\boverbar{m}_K) } = O(t_n^{-\alpha_3}).
\eeqr
Recall that $\wh\mu_{K+1}(x,\boverbar{m}_K)=\bbar{u}_{K+1}(x,\boverbar{m}_K)\trans\wh{\bm\pi}_{K+1}$ with \begin{small}
\beqrs 
\wh{\bm\pi}_{K+1} = \left\{\sum_{i=1}^{n} T_iR_i(1+\wh\gamma_i)\bbar{u}_{K+1}(X_i,\boverbar{M}_{K,i})\bbar{u}_{K+1}(X_i,\boverbar{M}_{K,i})\trans\right\}^{-1} \sum_{i=1}^{n} T_iR_i(1+\wh\gamma_i)Y_i\bbar{u}_{K+1}(X_i,\boverbar{M}_{K,i}),
\eeqrs
\end{small}
Then by \Cref{ass:supp_sieve_mu}(i)(iv),
\begin{align}
\label{eq:supp_pf_rate_gamma_1}
&\|\wh\mu_{K+1}-\mu_{K+1}\|_{\calL_2} \le \|\wh\mu_{K+1}-\mu_{K+1,n} \|_{\calL_2} + \|\mu_{K+1,n} -\mu_{K+1}\|_{\calL_2} \nonumber\\ 
&~~\le \|\wh{\bm\pi}_{K+1}-{\bm\pi}_{K+1}\| \mbox{Eig}_{\max}^{1/2}(E\{\bbar{u}_{K+1}(M,X^*)\bbar{u}_{K+1}(M,X^*)\trans\}) \nonumber\\ 
&\qquad + \sup_{(x,\boverbar{m}_K)\in\calX\times\boverbar\calM_K} \Abs{\mu_{K+1}(x,\boverbar{m}_K) - \mu_{K+1,n}(x,\boverbar{m}_K)} \nonumber\\
&~~\lesssim \|\wh{\bm\pi}_{K+1}-{\bm\pi}_{K+1}\| + O(t_n^{-\alpha_3}),
\end{align}
where $\mbox{Eig}_{\max}(\cdot)$ indicates the maximal eigenvalue of a matrix, and 
\begin{align}
\label{eq:supp_pf_rate_gamma_2}
&\|\wh\mu_{K+1}-\mu_{K+1}\|_\infty 
\le \|\wh\mu_{K+1}-\mu_{K+1,n} \|_\infty + \|\mu_{K+1,n} -\mu_{K+1}\|_\infty \nonumber\\ 
\le&~ \|\wh{\bm\pi}_{K+1}-{\bm\pi}_{K+1}\| \sup_{(x,\boverbar{m}_K)\in\calX\times\boverbar\calM_K} \|\bbar{u}_{K+1}(x,\boverbar{m}_K)\| \nonumber\\
&+ \sup_{(x,\boverbar{m}_K)\in\calX\times\boverbar\calM_K} \Abs{\mu_{K+1}(x,\boverbar{m}_K) - \mu_{K+1,n}(x,\boverbar{m}_K)} \nonumber\\
\lesssim&~ \|\wh{\bm\pi}_{K+1}-{\bm\pi}_{K+1}\|\zeta_{2(K+1)n} + O(t_n^{-\alpha_3}).
\end{align}
Combining \eqref{eq:supp_pf_rate_gamma_1}, \eqref{eq:supp_pf_rate_gamma_2} and the result
\beqr\label{eq:supp_pf_coeff_gamma_0}
\|\wh{\bm\pi}_{K+1}-{\bm\pi}_{K+1}\| = O_p(\sqrt{\lambda_{K+1}}).
\eeqr
that we will prove soon later yields the theorem. Now we prove \eqref{eq:supp_pf_coeff_gamma_0} holds under \Cref{ass:supp_technique_mu,ass:supp_dgp,ass:supp_sieve_mu}.
Following \cite{Newey1997}, by \Cref{ass:positivity} and \ref{ass:supp_sieve_mu}(i), and \Cref{thm:weak_rate_delta}, with probability approaching one 
\beqrs
\left\{\frac{1}{n} \sum_{i=1}^{n} T_iR_i(1+\wh\gamma_i) \bbar{u}_{K+1}(X_i,\boverbar{M}_{K,i})\bbar{u}_{K+1}(X_i,\boverbar{M}_{K,i})\trans\right\}^{-1}
\eeqrs
is invertible and 
\begin{subequations}
\begin{align}
&\|\wh{\bm\pi}_{K+1}-{\bm\pi}_{K+1} \| \nonumber\\ 
&~\lesssim \Big\| \frac{1}{n}\sum_{i=1}^{n} T_iR_i(1+\gamma_i) \bbar{u}_{K+1}(X_i,\boverbar{M}_{K,i})\left\{Y_i-\mu_{K+1,n}(X_i,\boverbar{M}_{K,i})\right\} \Big\| \label{eq:supp_pf_coeff_gamma_T1}\\ 
&\quad+ \Big\| \frac{1}{n}\sum_{i=1}^{n} T_iR_i(\wh\gamma_i-\gamma_i) \bbar{u}_{K+1}(X_i,\boverbar{M}_{K,i})\{Y_i-\mu_{K+1}(X_i,\boverbar{M}_{K,i})\} \Big\| \label{eq:supp_pf_coeff_gamma_T2}\\
&\quad+ \Big\| \frac{1}{n}\sum_{i=1}^{n} T_iR_i(\wh\gamma_i-\gamma_i) \bbar{u}_{K+1}(X_i,\boverbar{M}_{K,i})\{\mu_{K+1}(X_i,\boverbar{M}_{K,i})-\mu_{K+1,n}(X_i,\boverbar{M}_{K,i})\} \Big\|. \label{eq:supp_pf_coeff_gamma_T3}
\end{align} 
\end{subequations}
For term \eqref{eq:supp_pf_coeff_gamma_T1}, following \cite{Newey1997}, we obtain 
\beqr\label{eq:supp_pf_coeff_gamma_1}
&& \Big\| \frac{1}{n}\sum_{i=1}^{n} T_iR_i(1+\gamma_i) \bbar{u}_{K+1}(X_i,\boverbar{M}_{K,i})\left\{Y_i-\mu_{K+1,n}(X_i,\boverbar{M}_{K,i})\right\} \Big\| \nonumber\\
&\lesssim& \Big\| \frac{1}{n}\sum_{i=1}^{n} T_iR_i(1+\gamma_i) \bbar{u}_{K+1}(X_i,\boverbar{M}_{K,i})\left\{Y_i-\mu_{K+1}(X_i,\boverbar{M}_{K,i})\right\} \Big\| \nonumber\\ 
&&+ \sup_{(x,\boverbar{m}_K)\in\calX\times\boverbar\calM_K} \Abs{\mu_{K+1}(x,\boverbar{m}_K) - \mu_{K+1,n}(x,\boverbar{m}_K) } \nonumber\\
&=& O_P(\sqrt{t_n/n}+t_n^{-\alpha_3}).
\eeqr
For term \eqref{eq:supp_pf_coeff_gamma_T2}, we define 
\beqrs
W_{ij}(\wt\gamma) = T_iR_i\wt\gamma(X_i,A_i,\boverbar{M}_i,Y_i) u_j(X_i,\boverbar{M}_{K,i})\{Y_i-\mu_{K+1}(X_i,\boverbar{M}_{K,i})\}.
\eeqrs
Note that $n^{-1}\sum_{i=1}^{n}W_{ij}(\wh\gamma)$ is a plug-in estimator of smooth functionals where the nonparametric component $\wt\gamma$ depends on endogenous variables and is estimated by the sieve minimum distance. Under \Cref{ass:supp_dgp,ass:supp_sieve_delta_inter,ass:supp_sieve_delta_exter,ass:supp_technique_mu}, \cite{AiChen2007} showed that
\begin{align*}
&\frac{1}{n}\sum_{i=1}^{n}W_{ij}(\wh\gamma)-\frac{1}{n}\sum_{i=1}^{n}W_{ij}(\gamma) = -E\{1+R\varepsilon_{K+1} u_j(M,X)h^*_{K+1,j}(X,A,\boverbar{M}_K,Y)\}^{-1}\\ 
&\qquad \times\frac{1}{n}\sum_{i=1}^{n} E\{Rh^*_{K+1,j}\mid Z_i,A_i,\boverbar{M}_{K,i},Y_i\}\{R_i\gamma(X_i,A_i,\boverbar{M}_i,Y_i)-1+R\} +o_p(n^{-1/2}).
\end{align*}
It follows that 
\begin{align*}
|\eqref{eq:supp_pf_coeff_gamma_T2}|
=&~ \sqrt{\sum_{j=1}^{t_n}\left|\frac{1}{n}\sum_{i=1}^{n}W_{ij}(\wh\gamma)-\frac{1}{n}\sum_{i=1}^{n}W_{ij}(\gamma)\right|^2}
= O_p(\sqrt{t_n/n}),
\end{align*}
since $E(R\gamma-1+R\mid Z,A,\boverbar{M}_K,Y)=0$ almost surely.
For term \eqref{eq:supp_pf_coeff_gamma_T3}, we define 
\beqrs
V_{ij}(\wt\gamma) = T_iR_i\wt\gamma(X_i,A_i,\boverbar{M}_i,Y_i) u_j(X_i,\boverbar{M}_{K,i})\{\mu_{K+1}(X_i,\boverbar{M}_{K,i})-\mu_{K+1,n}(X_i,\boverbar{M}_{K,i})\},
\eeqrs
for $j=1,\ldots,t_n$, and $V_{i\cdot}(\wt\gamma)=(V_{i1}(\wt\gamma),\ldots,V_{it_n}(\wt\gamma))\trans$. Define the class of functions 
\begin{align*}
\calG_j=\{tr&u_j(x,\boverbar{m}_K)\{\mu_{K+1}(x,\boverbar{m}_K)-\mu_{K+1,n}(x,\boverbar{m}_K)\}\wt\gamma(x,a,\boverbar{m}_K,y): \wt\gamma\in\Gamma\}.
\end{align*}
Note that $\calG_j$ consists of functions that is linear and hence Lipschitz continuous in $\wt\gamma\in\Gamma$ under $\|\cdot\|_\infty$, then Theorem 2.7.11 of \cite{Vaart_ep_1996} implies that $N_{[]}(\epsilon,\calG_j,\|\cdot\|_\infty)\le N(\epsilon/C,\Gamma,\|\cdot\|_\infty)$ for some $C>0$. Thus, by \Cref{ass:supp_dgp} and then Theorem 2.5.6 and 2.7.1 of \cite{Vaart_ep_1996}, we have $\calG_j$ is a Donsker class.
It follows that 
\begin{align*} 
\frac{1}{n}\sum_{i=1}^{n} V_{ij}(\wt\gamma) 
- \frac{1}{n}\sum_{i=1}^{n} V_{ij}(\gamma) 
- E\{V_{ij}(\wt\gamma-\gamma)\} = o_p(n^{-1/2}),
\end{align*}
uniformly over $\wt\gamma\in\{\wt\gamma\in\Gamma:\|\wt\gamma-\gamma\|_\infty=o_p(1)\}$. It follows that 
\begin{align}\label{eq:supp_pf_coeff_gamma_T3.1}
&\left\|\frac{1}{n}\sum_{i=1}^{n} V_{i\cdot}(\wt\gamma) 
- \frac{1}{n}\sum_{i=1}^{n} V_{i\cdot}(\gamma) 
- E\{V_{i\cdot}(\wt\gamma-\gamma)\}\right\|^2 \nonumber\\ 
&\quad=~ \sum_{j=1}^{t_n} 
\left|\frac{1}{n}\sum_{i=1}^{n} V_{ij}(\wt\gamma) 
- \frac{1}{n}\sum_{i=1}^{n} V_{ij}(\gamma) 
- E\{V_{ij}(\wt\gamma-\gamma)\}\right|^2 = o_p(t_n/n).
\end{align}
uniformly over $\wt\gamma\in\{\wt\gamma\in\Gamma:\|\wt\gamma-\gamma\|_\infty=o_p(1)\}$. Note that 
\begin{align}\label{eq:supp_pf_coeff_gamma_T3.2}
&\left\|E\{V_{i\cdot}(\wt\gamma-\gamma)\}\right\|^2\nonumber\\
\le&~ \mbox{Eig}_{\max}^{-1}(E\{\bbar{u}_{K+1}(X,\boverbar{M}_K)\bbar{u}_{K+1}(X,\boverbar{M}_K)\trans\})
E\left\{TR(\wt\gamma-\gamma)(\mu_{K+1}-\mu_{K+1,n})\bbar{u}_{K+1}(X,\boverbar{M}_K)\right\}\trans\nonumber\\
&\qquad\times E\left\{\bbar{u}_{K+1}(X,\boverbar{M}_K)\bbar{u}_{K+1}(X,\boverbar{M}_K)\trans\right\}^{-1}E\left\{TR(\wt\gamma-\gamma)(\mu_{K+1}-\mu_{K+1,n})\bbar{u}_{K+1}(X,\boverbar{M}_K)\right\}\nonumber\\ 
\lesssim&~ E\left[\left\{TR(\wt\gamma-\gamma)(\mu_{K+1}-\mu_{K+1,n})\right\}^2\right] \nonumber\\ 
\lesssim&~ \sup_{(x,\boverbar{m}_K)\in\calX\times\boverbar\calM_K} \Abs{\mu_{K+1}(x,\boverbar{m}_K) -\mu_{K+1,n}(x,\boverbar{m}_K) }^2 \|\wt\gamma-\gamma\|_\infty^2 = o(t_n^{-2\alpha_3}).
\end{align}
uniformly over $\wt\gamma\in\{\wt\gamma\in\Gamma:\|\wt\gamma-\gamma\|_\infty=o_p(1)\}$, where the second inequality is because 
the product is the squared norm of $\calL_2$-projection of $TR(\wt\gamma-\gamma)(\mu_{K+1}-\mu_{K+1,n})$ on the space linearly spanned by $\bbar{u}_{K+1}(x,\boverbar{m}_K)$.
We denote $\Gamma_0 = \{\wt\gamma\in\Gamma:\|\wt\gamma-\gamma\|_\infty=o_p(1)\}$. By Chebyshev's inequality, \eqref{eq:supp_pf_coeff_gamma_T3.1} and \eqref{eq:supp_pf_coeff_gamma_T3.2} imply that
\begin{align*}
|\eqref{eq:supp_pf_coeff_gamma_T3}|^2 \le&~ \sup_{\wt\gamma\in\Gamma_0} \Big\| \frac{1}{n}\sum_{i=1}^{n} T_iR_i(\wt\gamma_i-\gamma_i) \bbar{u}_{K+1}(X_i,\boverbar{M}_{K,i})\{\mu_{K+1}(X_i,\boverbar{M}_{K,i})-\mu_{K+1,n}(X_i,\boverbar{M}_{K,i})\} \Big\|^2 \\ 
\le&~ \sup_{\wt\gamma\in\Gamma_0} \left\|\frac{1}{n}\sum_{i=1}^{n} V_{i\cdot}(\wt\gamma) 
- \frac{1}{n}\sum_{i=1}^{n} V_{i\cdot}(\gamma) 
- E\{V_{i\cdot}(\wt\gamma-\gamma)\}\right\|^2 
+ \sup_{\wt\gamma\in\Gamma_0} \left\|E\{V_{i\cdot}(\wt\gamma-\gamma)\}\right\|^2
\\
= &~ o_p(t_n/n+t_n^{-2\alpha_3}).
\end{align*}
Recalling that $\max\{t_n/n,t_n^{-\alpha_3}\}=o(\lambda_{K+1})$, 
then combining the results of \eqref{eq:supp_pf_coeff_gamma_T1}-\eqref{eq:supp_pf_coeff_gamma_T3} yields \eqref{eq:supp_pf_coeff_gamma_0}.

Next, we calculate the convergence rate of $\wh\mu_K$.
Let $S\triangleq I(A=a_K)$ and $S_i\triangleq I(A_i=a_K)$. 
By \Cref{ass:supp_sieve_mu}(ii), there exists $\mu_{Kn}(x,\boverbar{m}_{K-1})=\bbar{u}_K(x,\boverbar{m}_{K-1})\trans{\bm\pi}_K$ with ${\bm\pi}_K\in\mR^{t_{Kn}}$ such that 
\beqr\label{eq:supp_def_pi_eta}
\sup_{x\in \calX} \Abs{\mu_{K}(x,\boverbar{m}_{K-1}) - \mu_{Kn}(x,\boverbar{m}_{K-1}) } = O(t_{Kn}^{-\alpha_{3K}}).
\eeqr
Recall that $\wh\mu_K(x,\boverbar{m}_{K-1})=\bbar{u}_K(x,\boverbar{m}_{K-1})\trans\wh{\bm\pi}_K$ with 
\begin{small}
\beqrs 
&&\wh{\bm\pi}_K = \Big\{\sum_{i=1}^{n} S_iR_i(1+\wh\gamma_i)\bbar{u}_K(X_i,\boverbar{M}_{K-1,i})\bbar{u}_K(X_i,\boverbar{M}_{K-1,i})\trans\Big\}^{-1} \\ 
&&\qquad\qquad\qquad \times \sum_{i=1}^{n} S_iR_i(1+\wh\gamma_i)\wh\mu_{K+1}(X_i,\boverbar{M}_{K,i})\bbar{u}_K(X_i,\boverbar{M}_{K-1,i}).
\eeqrs
\end{small}
Similar to part (i), the the convergence rate of $\wh\mu_K$ in \Cref{thm:convergence_rate_gamma_eta} is obtained following the conclusion
\beqr\label{eq:supp_pf_coeff_eta_0}
\|\wh{\bm\pi}_K-{\bm\pi}_K\| = O_p\big(\sqrt{\lambda_K}+t_{n}^{-\alpha_{3}}\big).
\eeqr
Now we prove \eqref{eq:supp_pf_coeff_eta_0} holds under \Cref{ass:supp_technique_mu,ass:supp_dgp,ass:supp_sieve_mu}. Suppose $\|\wh{\bm\pi}_{1,K+1}-{\bm\pi}_{1,K+1}\|$ satisfies \eqref{eq:supp_pf_coeff_eta_0}. Following \cite{Newey1997}, by \Cref{ass:positivity,ass:supp_sieve_mu}(i), with probability approaching one 
\beqrs
\Big\{\frac{1}{n} \sum_{i=1}^{n} S_iR_i(1+\wh\gamma_i)\bbar{u}_K(X_i,\boverbar{M}_{K-1,i})\bbar{u}_K(X_i,\boverbar{M}_{K-1,i})\trans\Big\}^{-1}
\eeqrs
is invertible and 
\begin{align}
&\|\wh{\bm\pi}_K-{\bm\pi}_K \| \nonumber\\
\lesssim&~ \Big\| \frac{1}{n}\sum_{i=1}^{n} S_iR_i(1+\gamma_i) \bbar{u}_K(X_i,\boverbar{M}_{K-1,i})\left\{\mu_{K+1}(X_i,\boverbar{M}_{K,i})-\mu_{Kn}(X_i,\boverbar{M}_{K-1,i})\right\} \Big\| \nonumber\\ 
&+ \Big\| \frac{1}{n}\sum_{i=1}^{n} S_iR_i(\wh\gamma_i-\gamma_i) \bbar{u}_K(X_i,\boverbar{M}_{K-1,i})\{\mu_{K+1}(X_i,\boverbar{M}_{K,i})-\mu_{K}(X_i,\boverbar{M}_{K-1,i})\} \Big\| \nonumber\\ 
&+ \Big\| \frac{1}{n}\sum_{i=1}^{n} S_iR_i(\wh\gamma_i-\gamma_i) \bbar{u}_K(X_i,\boverbar{M}_{K-1,i})\{\mu_{K}(X_i,\boverbar{M}_{K-1,i})-\mu_{Kn}(X_i,\boverbar{M}_{K-1,i})\}\Big\| \nonumber\\ 
&+ \Big\| \frac{1}{n}\sum_{i=1}^{n} S_iR_i(1+\wh\gamma_i) \bbar{u}_K(X_i,\boverbar{M}_{K-1,i})\left\{\wh\mu_{K+1}(X_i,\boverbar{M}_{K,i})-\mu_{K+1}(X_i,\boverbar{M}_{K,i})\right\}\Big\|.\nonumber\\ \label{eq:supp_pf_coeff_eta_1}
\end{align}
Similar to \eqref{eq:supp_pf_coeff_gamma_T1}, we obtain 
\begin{align}
\label{eq:supp_pf_coeff_eta_2}
&\Big\| \frac{1}{n}\sum_{i=1}^{n} S_iR_i(1+\gamma_i) \bbar{u}_K(X_i,\boverbar{M}_{K-1,i})\left\{\mu_{K+1}(X_i,\boverbar{M}_{K,i})-\mu_{Kn}(X_i,\boverbar{M}_{K-1,i})\right\} \Big\| \nonumber\\ 
&\quad= O_P(\sqrt{t_{Kn}/n}+t_{Kn}^{-\alpha_{3K}}).
\end{align}
Similar to \eqref{eq:supp_pf_coeff_gamma_T2}, we obtain 
\begin{align}
\label{eq:supp_pf_coeff_eta_3}
&\Big\| \frac{1}{n}\sum_{i=1}^{n} S_iR_i(\wh\gamma_i-\gamma_i) \bbar{u}_K(X_i,\boverbar{M}_{K-1,i})\{\mu_{K+1}(X_i,\boverbar{M}_{K,i})-\mu_{K}(X_i,\boverbar{M}_{K-1,i})\}\Big\| \nonumber\\ 
&\quad =O_p(\sqrt{t_{Kn}/n}),
\end{align}
Similar to \eqref{eq:supp_pf_coeff_gamma_T3}, we obtain 
\begin{align}
\label{eq:supp_pf_coeff_eta_4}
&\Big\| \frac{1}{n}\sum_{i=1}^{n} S_iR_i(\wh\gamma_i-\gamma_i) \bbar{u}_K(X_i,\boverbar{M}_{K-1,i})\{\mu_{K}(X_i,\boverbar{M}_{K-1,i})-\mu_{Kn}(X_i,\boverbar{M}_{K-1,i})\}\Big\|\qquad \nonumber\\ 
&\quad=o_p(\sqrt{t_{Kn}/n}+t_{Kn}^{-\alpha_{3K}}).
\end{align}
We consider the last term of \eqref{eq:supp_pf_coeff_eta_1}. 
For any $\wt\gamma \in \{\wt\gamma\in\Gamma:\|\wt\gamma-\gamma\|_\infty= o_p(1)\}$, we have the decomposition
\begin{align}
&\Big\| \frac{1}{n}\sum_{i=1}^{n} S_iR_i(1+\wt\gamma) \bbar{u}_K(X_i,\boverbar{M}_{K-1,i})\left\{\wh\mu_{K+1}(X_i,\boverbar{M}_{K,i})-\mu_{K+1}(X_i,\boverbar{M}_{K,i})\right\}\Big\| \nonumber\\
\lesssim&~ \sup_{(x,\boverbar{m}_K)\in\calX\times\boverbar\calM_K} \Abs{\mu_{K+1}(x,\boverbar{m}_K) - \mu_{K+1,n}(x,\boverbar{m}_{K}) } \nonumber\\
&+\Big\| \frac{1}{n}\sum_{i=1}^{n} S_iR_i(1+\wt\gamma) \bbar{u}_K(X_i,\boverbar{M}_{K-1,i})\left\{\wh\mu_{K+1}(X_i,\boverbar{M}_{K,i})-\mu_{K+1,n}(X_i,\boverbar{M}_{K,i})\right\}\Big\|. \label{eq:supp_coeff_eta_T4_pf1.2}
\end{align}
Note that
\begin{align}
\label{eq:supp_coeff_eta_T4_pf2}
& \Big\| \frac{1}{n}\sum_{i=1}^{n} S_iR_i(1+\wt\gamma) \bbar{u}_K(X_i,\boverbar{M}_{K-1,i})\left\{\wh\mu_{K+1}(X_i,\boverbar{M}_{K,i})-\mu_{K+1,n}(X_i,\boverbar{M}_{K,i})\right\} \nonumber\\ 
&\qquad - E [SR(1+\wt\gamma) \bbar{u}_K(X,\boverbar{M}_{K-1})\bbar{u}_{K+1}(X,\boverbar{M}_K)\trans] (\wh{\bm\pi}_{K+1}-{\bm\pi}_{K+1})\Big\|^2 \nonumber\\ 
\le&~ \sum_{j=1}^{t_{Kn}} \Big\|\frac{1}{n}\sum_{i=1}^{n} S_iR_i(1+\wt\gamma) u_{Kj}(X_i,\boverbar{M}_{K-1,i}) \bbar{u}_{K+1}(X_i,\boverbar{M}_{K,i}) \nonumber\\ 
&\qquad - E [SR(1+\wt\gamma) u_{Kj}(X,\boverbar{M}_{K-1})\bbar{u}_{K+1}(X,\boverbar{M}_K)\trans] \Big\|^2 
\cdot \Big\|(\wh{\bm\pi}_{K+1}-{\bm\pi}_{K+1})\Big\|^2 \nonumber\\ 
=&~ O(t_{Kn})\cdot O_p({t_n/n})\cdot O_p\big(\lambda_{K+1}\big) = o_p({t_{Kn}/n}),
\end{align}
uniformly over $\wt\gamma \in \{\wt\gamma\in\Gamma:\|\wt\gamma-\gamma\|_\infty= o_p(1)\}$ by \eqref{eq:supp_pf_coeff_gamma_0} because $\calG_{ij}\equiv \{sr(1+\wt\gamma) u_{K,i}(x,\boverbar{m}_{K-1})u_j(x,\boverbar{m}_K):\wt\gamma\in\Gamma\}$ for $i=1,\ldots, t_{Kn}$ and $j=1,\ldots, t_n$ are DonsKer by \Cref{ass:supp_dgp} and Theorem 2.5.6, 2.7.1 and 2.7.11 of \cite{Vaart_ep_1996}.
Let
\begin{align*}
\bm\Upsilon= E\{& TR(1+\gamma) \bbar{u}_{K+1}(X,\boverbar{M}_K)\bbar{u}_{K+1}(X,\boverbar{M}_K)\trans\}^{-1} \\ 
& \times E\{SR(1+\gamma) \bbar{u}_{K+1}(X,\boverbar{M}_K)\bbar{u}_K(X,\boverbar{M}_{K-1})\trans\}\in\mR^{t_n\times t_{Kn}},
\end{align*}
and $\bm\Upsilon_j$ is the $j$-th column of $\bm\Upsilon$. By \Cref{ass:positivity,ass:supp_sieve_mu}(i) and \Cref{thm:weak_rate_delta}, we have 
\begin{align}
& \Big\| E[SR(1+\wt\gamma) \bbar{u}_K(X,\boverbar{M}_{K-1})\bbar{u}_{K+1}(X,\boverbar{M}_K)\trans] (\wh{\bm\pi}_{K+1}-{\bm\pi}_{K+1}) \Big\|^2 \nonumber\\ 
\lesssim&~ \Big\| \bm\Upsilon\trans \frac{1}{n}\sum_{i=1}^{n} T_iR_i(1+\wh\gamma_i) (Y_i-\mu_{K+1,n}(X_i,\boverbar{M}_{K,i}))\bbar{u}_{K+1}(X_i,\boverbar{M}_{K,i}) \Big\|^2 \nonumber\\ 
=&~ \sum_{j=1}^{t_{Kn}} \Abs{ \frac{1}{n}\sum_{i=1}^{n} T_iR_i(\wh\gamma_i-\gamma_i) (Y_i-\mu_{K+1}(X_i,\boverbar{M}_{K,i}))\bbar{u}_{K+1}(X_i,\boverbar{M}_{K,i})\trans \bm\Upsilon_j}^2 \nonumber\\ 
&+ \sum_{j=1}^{t_{Kn}} \Abs{ \frac{1}{n}\sum_{i=1}^{n} T_iR_i(\wh\gamma_i-\gamma_i) (\mu_{K+1}(X_i,\boverbar{M}_{K,i})-\mu_{K+1,n}(X_i,\boverbar{M}_{K,i}))\bbar{u}_{K+1}(X_i,\boverbar{M}_{K,i})\trans \bm\Upsilon_j}^2 \nonumber\\
&+ \sum_{j=1}^{t_{Kn}} \Abs{ \frac{1}{n}\sum_{i=1}^{n} T_iR_i(1+\gamma_i) (Y_i-\mu_{K+1}(X_i,\boverbar{M}_{K,i}))\bbar{u}_{K+1}(X_i,\boverbar{M}_{K,i})\trans \bm\Upsilon_j}^2 \nonumber\\ 
&+ \Big\| \bm\Upsilon\trans \frac{1}{n}\sum_{i=1}^{n} T_iR_i(1+\gamma_i) (\mu_{K+1}(X_i,\boverbar{M}_{K,i})-\mu_{K+1,n}(X_i,\boverbar{M}_{K,i}))\bbar{u}_{K+1}(X_i,\boverbar{M}_{K,i}) \Big\|^2, \label{eq:supp_coeff_eta_T4_pf3.0}
\end{align}
uniformly over $\wt\gamma\in \{\wt\gamma\in\Gamma:\|\wt\gamma-\gamma\|_\infty= o_p(1)\}$. Notice that $\bbar{u}_{K+1}(X_i,\boverbar{M}_{K,i})\trans \bm\Upsilon_j$ is the weighted $\calL_2$ projection of $u_{Kj}(x,\boverbar{m}_{k-1})f_{A\mid X,\boverbar{M}_K}(a_K\mid x,\boverbar{m}_K)/f_{A\mid X,\boverbar{M}_K}(a_{K+1}\mid x,\boverbar{m}_K)$ on the linear space spanned by $\bbar{u}_{K+1}(x,\boverbar{m}_K)$ and thus has finite second moment by \Cref{ass:positivity} and \ref{ass:supp_sieve_mu}(i). Then by \Cref{ass:supp_dgp}, similar to \eqref{eq:supp_pf_coeff_gamma_T2}, we obtain 
\begin{align}\label{eq:supp_coeff_eta_T4_pf3.1}
\sum_{j=1}^{t_{Kn}} \Abs{ \frac{1}{n}\sum_{i=1}^{n} T_iR_i(\wh\gamma_i-\gamma_i) (Y_i-\mu_{K+1}(X_i,\boverbar{M}_{K,i}))\bbar{u}_{K+1}(X_i,\boverbar{M}_{K,i})\trans \bm\Upsilon_j}^2 =O_p(t_{Kn}/n).
\end{align}
Similar to \eqref{eq:supp_pf_coeff_gamma_T3}, we obtain 
\begin{align}\label{eq:supp_coeff_eta_T4_pf3.1.1}
&\sum_{j=1}^{t_{Kn}} \Abs{ \frac{1}{n}\sum_{i=1}^{n} T_iR_i(\wh\gamma_i-\gamma_i) (\mu_{K+1}(X_i,\boverbar{M}_{K,i})-\mu_{K+1,n}(X_i,\boverbar{M}_{K,i}))\bbar{u}_{K+1}(X_i,\boverbar{M}_{K,i})\trans \bm\Upsilon_j}^2\quad\nonumber\\ 
&\quad=o_p(t_{Kn}/n+t_n^{-2\alpha_3}).
\end{align}
We also have
\begin{align}\label{eq:supp_coeff_eta_T4_pf3.2}
\sum_{j=1}^{t_{Kn}} \Abs{ \frac{1}{n}\sum_{i=1}^{n} T_iR_i(1+\gamma_i) (Y_i-\mu_{K+1}(X_i,\boverbar{M}_{K,i}))\bbar{u}_{K+1}(X_i,\boverbar{M}_{K,i})\trans \bm\Upsilon_j}^2 =O_p(t_{Kn}/n),
\end{align}
because $E[I(A=a_{K+1})R(1+\gamma) (Y-\mu_{K+1}(X,\boverbar{M}_{K}))\mid X,\boverbar{M}_{K}]=0$. By \Cref{ass:positivity}, \Cref{ass:supp_sieve_mu}(i),
\begin{align}\label{eq:supp_coeff_eta_T4_pf3.3}
&\Big\| \bm\Upsilon\trans \frac{1}{n}\sum_{i=1}^{n} T_iR_i(1+\gamma_i) (\mu_{K+1}(X_i,\boverbar{M}_{K,i})-\mu_{K+1,n}(X_i,\boverbar{M}_{K,i}))\bbar{u}_{K+1}(X_i,\boverbar{M}_{K,i}) \Big\|^2 \nonumber\\ 
&\qquad\lesssim \sup_{(x,\boverbar{m}_k)\in\calX\times\boverbar\calM_K} \Abs{\mu_{K+1}(x,\boverbar{m}_K) - \mu_{K+1,n}(x,\boverbar{m}_K) }^2 =O_p(t_n^{-2\alpha_3}),
\end{align}
Combining \eqref{eq:supp_coeff_eta_T4_pf1.2}--\eqref{eq:supp_coeff_eta_T4_pf3.3} yields that
\begin{align*} 
&\Big\| \frac{1}{n}\sum_{i=1}^{n} S_iR_i(1+\wt\gamma_i) \bbar{u}_K(X_i,\boverbar{M}_{K-1,i})\left\{\wh\mu_{K+1}(X_i,\boverbar{M}_{K,i})-\mu_{K+1}(X_i,\boverbar{M}_{K,i})\right\}\Big\|\qquad\nonumber\\ 
&\quad = O_p\big(\sqrt{t_{Kn}/n}+t_n^{-\alpha_3}\big).
\end{align*}
Recalling that $\max\{\sqrt{t_{Kn}/n},t_{Kn}^{-\alpha_{3K}}\}=o(\lambda_K)$, then \eqref{eq:supp_pf_coeff_eta_0} follows, and we obtain 
\begin{align*}
\|\wh\mu_K-\mu_K\|_{\calL_2} = O_p\big(\sqrt{\lambda_K}+t_n^{-\alpha_3}\big),\quad\mbox{and}
\|\wh\mu_K-\mu_K\|_\infty = O_p\Big\{\zeta_{2Kn}\big(\sqrt{\lambda_K}+t_n^{-\alpha_3}\big)\Big\}.
\end{align*}
For $k\in\{1,\ldots,K-1\}$, we can obtain 
\beqr\label{eq:supp_pf_coeff_eta_01}
\|\wh{\bm\pi}_k-{\bm\pi}_k\| = O_p\big(\sqrt{\lambda_k}+\textstyle{\sum_{j=k+1}^{K+1}}t_{jn}^{-\alpha_{3j}}\big).
\eeqr
iteratively following very similar procedure of calculating $\|\wh{\bm\pi}_K-{\bm\pi}_K\|$. It follows that
\begin{align*}
\|\wh\mu_k-\mu_k\|_{\calL_2} = O_p\big(\sqrt{\lambda_k}+\textstyle{\sum_{j=k+1}^{K+1}}t_{jn}^{-\alpha_{3j}}\big),\quad
\|\wh\mu_k-\mu_k\|_\infty = O_p\Big\{\zeta_{2kn}\big(\sqrt{\lambda_k}+\textstyle{\sum_{j=k+1}^{K+1}}t_{jn}^{-\alpha_{3j}}\big)\Big\}.
\end{align*}
Then we complete the proof.
\end{proof}

\newpage
\section{Asymptotic normality}\label{sec:app_pf_thm_asy_normal}
\subsection{Proof of Theorem \ref{thm:asy_normal}.}

\begin{proof}
Let $0<\epsilon_n=o(n^{-1/2})$, and $\varrho_0$ satisfies \Cref{ass:representer}. Let $\varrho_n=\Pi_n\varrho_0\in\Gamma_n$ with the operator $\Pi_n$ defined in \Cref{ass:supp_sieve_delta_inter}(ii). By \Cref{ass:supp_sieve_delta_inter}(i), $\wh\gamma \pm \epsilon_n\varrho_n\in\Gamma_n$. By the definition of $\wh\gamma$, it follows that
\begin{align}
\label{eq:app_delta_1}
0 \le&~ \frac{1}{2n}\sum_{i=1}^n [\wh{E}\{R(\wh\gamma \pm \epsilon_n\varrho_n)-1+R\mid Z_i,A_i,\boverbar{M}_{K,i},Y_i\}]^2 \nonumber\\ 
& ~ - \frac{1}{2n}\sum_{i=1}^n \big\{\wh{E}(R\wh\gamma-1+R\mid Z_i,A_i,\boverbar{M}_{K,i},Y_i)\big\}^2 \nonumber\\ 
=&~ \frac{1}{2n}\sum_{i=1}^n \big\{\wh{E}(R\wh\gamma-1+R\mid Z_i,A_i,\boverbar{M}_{K,i},Y_i) \pm \epsilon_n\wh{E}(R\varrho_n\mid Z_i,A_i,\boverbar{M}_{K,i},Y_i)\big\}^2 \nonumber\\ 
&~ - \frac{1}{2n}\sum_{i=1}^n \big\{\wh{E}(R\wh\gamma-1+R\mid Z_i,A_i,\boverbar{M}_{K,i},Y_i)\big\}^2 \nonumber\\ 
=&~ \frac{\pm\epsilon_n}{n}\sum_{i=1}^n \wh{E}(R\varrho_n\mid Z_i,A_i,\boverbar{M}_{K,i},Y_i) \wh{E}(R\wh\gamma-1+R\mid Z_i,A_i,\boverbar{M}_{K,i},Y_i) \nonumber\\ 
& ~ + \frac{\epsilon_n^2}{2n}\sum_{i=1}^n\big\{\wh{E}(R\varrho_n\mid Z_i,A_i,\boverbar{M}_{K,i},Y_i)\big\}^2,
\end{align}
for any $0<\epsilon_n=o(n^{-1/2})$. By \cref{lemma:app_delta_exter_bound} and \Cref{ass:supp_dgp}(i), \ref{ass:supp_sieve_delta_inter} and \ref{ass:supp_sieve_delta_exter}(iii), we have 
\beqr\label{eq:app_delta_2}
\frac{\epsilon_n^2}{2n}\sum_{i=1}^n\big\{\wh{E}(R\varrho_n\mid Z_i,A_i,\boverbar{M}_{K,i},Y_i)\big\}^2 = O_p(\epsilon_n^2).
\eeqr
Combining \eqref{eq:app_delta_1} and \eqref{eq:app_delta_2} yields
\beqr\label{eq:app_delta_3}
\frac{1}{n}\sum_{i=1}^n \wh{E}(R\varrho_n\mid Z_i,A_i,\boverbar{M}_{K,i},Y_i) \wh{E}(R\wh\gamma-1+R\mid Z_i,A_i,\boverbar{M}_{K,i},Y_i) = o_p(n^{-1/2}).
\eeqr
Then by \Cref{lemma:app_delta,lemma:app_rho}, \Cref{thm:weak_rate_delta}, \Cref{ass:supp_sieve_strengthened_rate}(i), and the Cauchy Schwarz inequality, \eqref{eq:app_delta_3} implies that 
\beqr\label{eq:app_delta_4}
&& \frac{1}{n}\sum_{i=1}^n E(R\varrho_0\mid Z_i,A_i,\boverbar{M}_{K,i},Y_i) \wh{E}(R\wh\gamma-1+R\mid Z_i,A_i,\boverbar{M}_{K,i},Y_i) \nonumber\\ 
&=& \frac{1}{n}\sum_{i=1}^n \big\{E(R\varrho_0\mid Z_i,A_i,\boverbar{M}_{K,i},Y_i) - \wh{E}(R\varrho_n\mid Z_i,A_i,\boverbar{M}_{K,i},Y_i)\big\} \wh{E}(R\wh\gamma-1+R\mid Z_i,A_i,\boverbar{M}_{K,i},Y_i) \nonumber\\
&&+ \frac{1}{n}\sum_{i=1}^n \wh{E}(R\varrho_n\mid Z_i,A_i,\boverbar{M}_{K,i},Y_i) \wh{E}(R\wh\gamma-1+R\mid Z_i,A_i,\boverbar{M}_{K,i},Y_i) = o_p(n^{-1/2}).
\eeqr
We denote $\wh{E}^{(d)}(R\varrho_0\mid Z_i,A_i,\boverbar{M}_{K,i},Y_i)\triangleq \wh{E}\big\{E(R\varrho_0\mid Z,A,\boverbar{M}_K,Y)\mid Z_i,A_i,\boverbar{M}_{K,i},Y_i\big\}$. By exchanging summation, \eqref{eq:app_delta_4} can be written as 
\begin{align} 
\label{eq:app_delta_5}
&& \frac{1}{n}\sum_{i=1}^n \wh{E}^{(d)}(R\varrho_0\mid Z_i,A_i,\boverbar{M}_{K,i},Y_i) \big\{R_i\wh\gamma(X_i,A_i,\boverbar{M}_{K,i},Y_i)-1+R_i\big\} = o_p(n^{-1/2}).
\end{align}
Then by \Cref{lemma:app_rho_hat}, \Cref{thm:weak_rate_delta} and \Cref{ass:supp_sieve_strengthened_rate}(i), \eqref{eq:app_delta_5} implies that 
\begin{small}
\begin{align}
\label{eq:app_delta_0}
& \frac{1}{n}\sum_{i=1}^n E(R\varrho_0\mid Z_i,A_i,\boverbar{M}_{K,i},Y_i) \{R_i\wh\gamma(X_i,A_i,\boverbar{M}_{K,i},Y_i)-1+R_i\} \nonumber\\ 
=&~ \frac{1}{n}\sum_{i=1}^n \big\{E(R\varrho_0\mid Z_i,A_i,\boverbar{M}_{K,i},Y_i) - \wh{E}^{(d)}(R\varrho_0\mid Z_i,A_i,\boverbar{M}_{K,i},Y_i)\big\} \big\{R_i\wh\gamma(X_i,A_i,\boverbar{M}_{K,i},Y_i)-1+R_i\big\} \nonumber\\ 
&+ \frac{1}{n}\sum_{i=1}^n \wh{E}^{(d)}(R\varrho_0\mid Z_i,A_i,\boverbar{M}_{K,i},Y_i) \big\{R_i\wh\gamma(X_i,A_i,\boverbar{M}_{K,i},Y_i)-1+R_i\big\} 
= o_p(n^{-1/2}).
\end{align}
\end{small}

Recall the definition of $\omega_1(x)=f_{A\mid X}(a_{1}\mid x)^{-1}$.
By \Cref{ass:supp_sieve_mu}(iii), there exists ${\bm\pi}_{21}\in\mR^{t_{1n}}$ such that 
\beqr \label{eq:app_eta_1}
\sup_{x\in\calX} \Abs{\omega_1(x)-\bbar{u}_1(x)\trans{\bm\pi}_{21}} = O(t_{1n}^{-\alpha_4}).
\eeqr
Recall that 
\begin{align*}
& \wh\mu_1(x) = \sum_{i=1}^n I(A_i=a_1)R_i\big\{1+\wh\gamma(X_i,A_i,\boverbar{M}_{K,i},Y_i)\big\} \wh\mu_2(X_i,M_{1i})\bbar{u}_1(X_i)\trans \\ 
&\quad \times\Big\{\sum_{i=1}^n I(A_i=a_1)R_i\big\{1+\wh\gamma(X_i,A_i,\boverbar{M}_{K,i},Y_i)\big\} \bbar{u}_1(X_i)\bbar{u}_1(X_i)\trans \Big\}^{-1}\bbar{u}_1(x).
\end{align*}
It follows that
\begin{align*}
&\frac{1}{n}\sum_{i=1}^n I(A_i=a_1)R_i \big\{1+\wh\gamma(X_i,A_i,\boverbar{M}_{K,i},Y_i)\big\}\wh\mu_1(X_i) \bbar{u}_1(X_i) \\
=&~ \frac{1}{n}\sum_{i=1}^n I(A_i=a_1)R_i \big\{1+\wh\gamma(X_i,A_i,\boverbar{M}_{K,i},Y_i)\big\}\bbar{u}_1(X_i)\bbar{u}_1(X_i)\trans \\
&\quad \times\Big[\sum_{i=1}^n I(A_i=a_1)R_i\big\{1+\wh\gamma(X_i,A_i,\boverbar{M}_{K,i},Y_i)\big\} \bbar{u}_1(X_i)\bbar{u}_1(X_i)\trans \Big]^{-1} \\ 
&\quad \times \sum_{i=1}^n I(A_i=a_1)R_i\big\{1+\wh\gamma(X_i,A_i,\boverbar{M}_{K,i},Y_i)\big\} \wh\mu_2(X_i,M_{1i})\bbar{u}_1(X_i) \\
=&~ \frac{1}{n} \sum_{i=1}^n I(A_i=a_1)R_i\big\{1+\wh\gamma(X_i,A_i,\boverbar{M}_{K,i},Y_i)\big\} \wh\mu_2(X_i,M_{1i})\bbar{u}_1(X_i).
\end{align*}
That is 
\begin{align}\label{eq:app_eta_1.1}
\frac{1}{n}\sum_{i=1}^n I(A_i=a_1)R_i \big\{1+\wh\gamma(X_i,A_i,\boverbar{M}_{K,i},Y_i)\big\}\big\{\wh\mu_2(X_i,M_{1i})-\wh\mu_1(X_i)\big\}\bbar{u}_1(X_i)=0.
\end{align}
By \Cref{ass:supp_sieve_strengthened_rate}(ii), \Cref{thm:weak_rate_delta,thm:convergence_rate_gamma_eta}, \eqref{eq:app_eta_1} and \eqref{eq:app_eta_1.1} implies
\begin{align}
\label{eq:app_omega_1}
&\frac{1}{n}\sum_{i=1}^n I(A_i=a_1)R_i \big\{1+\wh\gamma(X_i,A_i,\boverbar{M}_{K,i},Y_i)\big\} \omega_1(X_i) \big\{\wh\mu_2(X_i,M_{1i})-\wh\mu_1(X_i)\big\} \nonumber\\
=&~ \frac{1}{n}\sum_{i=1}^n I(A_i=a_1)\big\{1+\wh\gamma(X_i,A_i,\boverbar{M}_{K,i},Y_i)\big\} \{\omega_1(X_i)-\bbar{u}_1(X_i)\trans{\bm\pi}_{21}\}\big\{\wh\mu_2(X_i,M_{1i})-\wh\mu_1(X_i)\big\} \nonumber \\ 
&+ \frac{1}{n}\sum_{i=1}^n I(A_i=a_1)\big\{1+\wh\gamma(X_i,A_i,\boverbar{M}_{K,i},Y_i)\big\}\big\{\wh\mu_2(X_i,M_{1i})-\wh\mu_1(X_i)\big\}\bbar{u}_1(X_i)\trans{\bm\pi}_{21} \nonumber \\
\lesssim&~ \sup_{x\in\calX} \Abs{\omega_1(x)-\bbar{u}_1(x)\trans{\bm\pi}_{21}} = O_p(t_{1n}^{-\alpha_4})=o_p(n^{-1/2}). 
\end{align}
Similarly we can obtain
\begin{align}
\label{eq:app_omega_k}
&\frac{1}{n}\sum_{i=1}^n I(A_i=a_k) R_i \big\{1+\wh\gamma(X_i,A_i,\boverbar{M}_{K,i},Y_i)\big\}\nonumber\\
&\qquad \times \Big\{\prod_{j=1}^k \omega_j(X_i,\boverbar{M}_{j-1,i})\Big\}
\Big\{\wh\mu_{k+1}(X_i,\boverbar{M}_k)-\wh\mu_k(X_i,\boverbar{M}_{k-1,i})\Big\} =o_p(n^{-1/2}),
\end{align}
for $k\in\{2,\ldots,K\}$, and 
\begin{align}
\label{eq:app_omega_K+1}
&\frac{1}{n}\sum_{i=1}^n I(A_i=a_{K+1}) R_i \big\{1+\wh\gamma(X_i,A_i,\boverbar{M}_{K,i},Y_i)\big\} \nonumber\\
&\qquad \times
\Big\{\prod_{j=1}^{K+1} \omega_j(X_i,\boverbar{M}_{j-1,i})\Big\}
\Big\{Y_i-\wh\mu_{K+1}(X_i,\boverbar{M}_{K,i})\Big\} =o_p(n^{-1/2}).
\end{align}

Let $\mu=(\mu_1,\ldots,\mu_{K+1})$ and define
\begin{align}
\label{eq:rho}
& \varpi(R,Z,X,A,\boverbar{M}_{K},Y;\gamma,\mu) = R\big\{1+\gamma(X,A,\boverbar{M}_{K},Y)\big\}
\nonumber \\ 
&\qquad \times\Bigg[ \sum_{k=1}^{K}I(A=a_k) \Big\{\prod_{j=1}^k \omega_j(X,\boverbar{M}_{j-1})\Big\}
\Big\{\mu_{k+1}(X,\boverbar{M}_k)-\mu_k(X,\boverbar{M}_{k-1})\Big\} \nonumber\\ 
&\qquad + I(A=a_{K+1}) \Big\{\prod_{j=1}^{K+1} \omega_j(X,\boverbar{M}_{j-1})\Big\}
\Big\{Y-\mu_{{K+1},\bunderline{a}_{K+1}}(X,\boverbar{M}_K)\Big\}\Bigg] \nonumber \\ 
&\qquad - E\{R\varrho_0(X,A,\boverbar{M}_K,Y)\mid Z,A,\boverbar{M}_K,Y\}\{R\gamma(X,A,\boverbar{M}_K,Y)-1+R\}.
\end{align}
Then combining \eqref{eq:app_delta_0}, \eqref{eq:app_omega_1}, \eqref{eq:app_omega_k} and \eqref{eq:app_omega_K+1} yields that
\beqr\label{eq:supp_theta_3}
\frac{1}{n}\sum_{i=1}^{n} \varpi(R_i,Z_i,X_i,A_i,\boverbar{M}_{K,i},Y_i;\wh\gamma,\wh\mu) = o_p(n^{-1/2}).
\eeqr

Finally, recalling that 
$$\wh\psi=\frac{1}{n}\sum_{i=1}^{n}\left\{R_i\wh\mu_1(X_i) + R_i\wh\gamma(X_i,A_i,\boverbar{M}_{K,i},Y_i)\wh\mu_1(X_i)\right\},$$
we obtain
\begin{align}\label{eq:supp_theta_0.5}
\wh\psi-\psi 
&~= \frac{1}{n}\sum_{i=1}^{n}\left\{R_i\wh\mu_1(X_i) + R_i\wh\gamma(X_i,A_i,\boverbar{M}_{K,i},Y_i)\wh\mu_1(X_i)\right\} - \psi \nonumber\\ 
&~= \frac{1}{n}\sum_{i=1}^{n}\left\{R_i\mu_1(X_i) + R_i\gamma(X_i,A_i,\boverbar{M}_{K,i},Y_i)\mu_1(X_i)\right\} -\psi \nonumber\\ 
&~\quad + \frac{1}{n}\sum_{i=1}^{n} R_i\big\{\wh\mu_1(X_i)-\mu_1(X_i)\big\} \nonumber\\ 
&~\quad + \frac{1}{n}\sum_{i=1}^{n} R_i\big\{\wh\gamma(X_i,A_i,\boverbar{M}_{K,i},Y_i)\wh\mu_1(X_i)-\gamma(X_i,A_i,\boverbar{M}_{K,i},Y_i)\mu_1(X_i)\big\}.
\end{align}
By \Cref{lemma:supp_stoc_ec}, \Cref{thm:weak_rate_delta,thm:convergence_rate_gamma_eta}, \Cref{ass:supp_dgp} and \ref{ass:supp_sieve_strengthened_rate}, we have
\begin{align}\label{eq:supp_theta_4}
&\frac{1}{n}\sum_{i=1}^{n} R_i\big\{\wh\mu_1(X_i)-\mu_1(X_i)\big\}+\frac{1}{n}\sum_{i=1}^{n} R_i\big\{\wh\gamma(X_i,A_i,\boverbar{M}_{K,i},Y_i)\wh\mu_1(X_i)-\gamma(X_i,A_i,\boverbar{M}_{K,i},Y_i)\mu_1(X_i)\big\} \nonumber\\ 
&- E\Big[ R\big\{\wh\mu_1(X)-\mu_1(X)\big\}+ R\big\{\wh\gamma(X,A,\boverbar{M}_{K},Y)\wh\mu_1(X)-\gamma(X,A,\boverbar{M}_{K},Y)\mu_1(X)\big\}\Big] \nonumber\\ 
& + \frac{1}{n}\sum_{i=1}^{n} \varpi(R_i,Z_i,X_i,A_i,\boverbar{M}_{K,i},Y_i;\wh\gamma,\wh\mu) - \frac{1}{n}\sum_{i=1}^{n} \varpi(R_i,Z_i,X_i,A_i,\boverbar{M}_{K,i},Y_i;\gamma,\mu) \nonumber\\ 
& - E\Big[\varpi(R,Z,X,A,\boverbar{M}_{K},Y;\wh\gamma,\wh\mu) - \varpi(R,Z,X,A,\boverbar{M}_{K},Y;\gamma,\mu)\Big]= o_p(n^{-1/2}).
\end{align}
Recalling the definition of $\phi$ in the main text, the pathwise derivative of 
$$E[R\mu_1(X) + R\gamma(X,A,\boverbar{M}_{K},Y)\mu_1(X)+\varpi(R,Z,X,A,\boverbar{M}_{K},Y;\gamma,\mu)]$$
with respect to $(\gamma,\mu)$ in the direction $(\wt\gamma,\wt\mu)$ evaluated at the true value, denoted by $A(\wt\gamma,\wt\mu)$, is
\begin{align}
\label{eq:A_func_def}
A(\wt\gamma,\wt\mu) 
=&~ E\big[R\{1+\gamma(X,A,\boverbar{M}_{K},Y)\}\left\{1-I(A=a_1)\omega_1(X)\right\}\wt\mu_1(X)\big] \nonumber\\ 
&~ +\sum_{k=2}^{K+1} E\bigg[R\{1+\gamma(X,A,\boverbar{M}_{K},Y)\}\Big\{\prod_{j=1}^{k-1} \omega_j(X,\boverbar{M}_{j-1})\Big\} \nonumber\\ 
&\qquad \times\left\{I(A=a_{k-1})-I(A=a_k)\omega_k(X,\boverbar{M}_{k-1})\right\}\wt\mu_k(X,\boverbar{M}_{k-1})\bigg] \nonumber\\ 
&~ +E\{R\phi(X,A,\boverbar{M}_K,Y)\wt\gamma(X,A,\boverbar{M}_K,Y)\}\nonumber\\ 
&~ - E\{ E(R\varrho_0\mid Z,A,\boverbar{M}_K,X) R\wt\gamma(X,A,\boverbar{M}_K,Y) \}.
\end{align}
By \Cref{ass:representer,lemma:iden_coef} in the main text, \eqref{eq:A_func_def} implies that 
\beqr\label{eq:A_func}
A(\wt\gamma,\wt\mu) \equiv 0,
\eeqr
for any $\wt\gamma$ and $\wt\mu$.
Then by \Cref{lemma:second_order_residual} and \Cref{ass:supp_dgp,ass:supp_sieve_mu}, and \Cref{thm:weak_rate_delta}, \eqref{eq:A_func} implies that
\begin{align}\label{eq:supp_theta_5}
& E\Big[ R\big\{\wh\mu_1(X)-\mu_1(X)\big\}+ R\big\{\wh\gamma(X,A,\boverbar{M}_{K},Y)\wh\mu_1(X)-\gamma(X,A,\boverbar{M}_{K},Y)\mu_1(X)\big\}\Big] \nonumber\\ 
& + E\Big[\varpi(R,Z,X,A,\boverbar{M}_{K},Y;\wh\gamma,\wh\mu) - \varpi(R,Z,X,A,\boverbar{M}_{K},Y;\gamma,\mu)\Big] \nonumber\\ 
&= A(\wh\gamma-\gamma,\wh\mu-\mu) + 
\mathfrak{R}(\wh\gamma,\wh\mu) 
= o_p(n^{-1/2}),
\end{align}
where 
\begin{align*}
\mathfrak{R}(\wt\gamma,\wt\mu) 
\triangleq &~ E\big[R\{\wt\gamma(X,A,\boverbar{M}_{K},Y)-\gamma(X,A,\boverbar{M}_{K},Y)\}\left\{1-I(A=a_1)\omega_1(X)\right\}\big\{\wt\mu_1(X)-\mu_1(X)\big\}\big] \\ 
&~ + \sum_{k=2}^{K+1}E\Big[R\{\wt\gamma(X,A,\boverbar{M}_{K},Y)-\gamma(X,A,\boverbar{M}_{K},Y)\}\Big\{\prod_{j=1}^{k-1} \omega_j(X,\boverbar{M}_{j-1})\Big\} \nonumber\\ 
&\quad \times\left\{I(A=a_{k-1})-I(A=a_k)\omega_k(X,\boverbar{M}_{k-1})\right\}\big\{\wt\mu_k(X,\boverbar{M}_{k-1})-\mu_k(X,\boverbar{M}_{k-1})\big\}\Big].
\end{align*}
Combining \eqref{eq:supp_theta_3}, \eqref{eq:supp_theta_0.5}, \eqref{eq:supp_theta_4} and \eqref{eq:supp_theta_5} yields that 
\begin{align*}
\wh\psi-\psi 
&~= 
\frac{1}{n}\sum_{i=1}^{n}\left\{R_i\mu_1(X_i) + R_i\gamma(X_i,A_i,\boverbar{M}_{K,i},Y_i)\mu_1(X_i)\right\} -\psi \nonumber\\ 
&\qquad+\frac{1}{n}\sum_{i=1}^{n} \varpi(R_i,Z_i,X_i,A_i,\boverbar{M}_{K,i},Y_i;\gamma,\mu) + o_p(n^{-1/2}) \nonumber\\ 
&~= \frac{1}{n}\sum_{i=1}^{n} IF(R_i,Z_i,X_i,\boverbar{M}_{K,i},Y_i) + o_p(n^{-1/2}),
\end{align*}
where 
\begin{align*}
&IF(R,Z,RX,A,\boverbar{M}_{K},Y) \\ 
&~= R\big\{1+\gamma(X,A,\boverbar{M}_K,Y)\big\}\phi(X,A,\boverbar{M}_K,Y) -\psi \\ 
&\qquad- E\{R\varrho(X,A,\boverbar{M}_K,Y)\mid Z,A,\boverbar{M}_K,Y\}\{R\gamma(X,A,\boverbar{M}_K,Y)-1+R\}.
\end{align*}
Therefore, \Cref{thm:asy_normal} is derived following the central limit theorem and the Slutsky theorem.
\end{proof}

\subsection{Efficiency of $\wh\psi$}\label{ssec:pf_EIF}

We demonstrate that the proposed estimator $\wh\psi$ is locally efficient in the sense that it attains the semiparametric efficiency bound for $\psi$ in the semiparametric model $\calM_{sp}$ under the following condition, which is a symmetric counterpart of \Cref{ass:completeness}.

\begin{assum}[Completeness]\label{ass:completeness_inverse}
For any squared-integrable function $g$ and for any $a,\boverbar{m}_K,y$, $E\{g(Z)\mid X,A\!=\!a,\boverbar{M}_K\!=\!\boverbar{m}_K,Y\!=\!y\}=0$ almost surely if and only if $g(Z)=0$ almost surely.
\end{assum}

\begin{prop}\label{thm:EIF}
The estimator $\wh\psi$ attains the semiparametric efficiency bound of $\psi$ in $\calM_{sp}$ at the submodels where \Cref{ass:completeness_inverse} holds.
\end{prop}

\begin{proof}
In order to demonstrate that the proposed estimator $\wh\psi$ is efficient, we present a proof in two steps. Firstly, we derive the tangent space, and secondly, we show the influence function derived in \Cref{thm:asy_normal} belongs to the tangent space.
The distribution of observed data $O=(R,RX,Z,A,\boverbar{M}_K,Y)$ is 
\begin{align*}
f(O)= f(Z,A,\boverbar{M}_K,Y)\bigbc{f(X\mid R=1,Z,A,\boverbar{M}_K,Y)f(R=1\mid Z,A,\boverbar{M}_K,Y)}^R \\ 
\times\bigbc{f(R=0\mid Z,A,\boverbar{M}_K,Y)}^{1-R}.
\end{align*}
Consider a regular parametric submodel $f_\tau(O)$ in $\calM_{sp}$ indexed by $\tau$ that equals $f(O)$ when $\tau=0$.
Then the corresponding score is given by 
\begin{align*}
S(O) =&~ S(Z,A,\boverbar{M}_K,Y) + R\big\{ S(X\mid R=1,Z,A,\boverbar{M}_K,Y) + S(R=1\mid Z,A,\boverbar{M}_K,Y) \big\} \nonumber\\ 
& ~ - \frac{1-R}{\beta(Z,A,\boverbar{M}_K,Y)}S(R=1\mid Z,A,\boverbar{M}_K,Y),
\end{align*}
where 
\beqrs
&& S(Z,A,\boverbar{M}_K,Y) = \partial \log f_\tau(Z,A,\boverbar{M}_K,Y)/\partial \tau \vert_{\tau=0}, \\
&& S(R=1\mid Z,A,\boverbar{M}_K,Y) = \partial \log f_\tau(R=1\mid Z,A,\boverbar{M}_K,Y) /\partial \tau \vert_{\tau=0},\\
&& S(X\mid R=1,Z,A,\boverbar{M}_K,Y) = \partial \log f_\tau(X\mid R=1,Z,A,\boverbar{M}_K,Y) /\partial \tau \vert_{\tau=0},\\ 
&& \beta(Z,A,\boverbar{M}_K,Y) = f(R=0\mid Z,A,\boverbar{M}_K,Y)/f(R=1\mid Z,A,\boverbar{M}_K,Y).
\eeqrs
Recalling that
\beqrs
E\left\{\gamma(X,A,\boverbar{M}_{K},Y) \mid R=1,Z,A,\boverbar{M}_{K},Y\right\} = \beta(Z,A,\boverbar{M}_{K},Y),\eeqrs
we have 
\beqrs
\frac{\partial}{\partial t} E_t\left\{\gamma_t(X,A,\boverbar{M}_{K},Y)-\beta_t(Z,A,\boverbar{M}_{K},Y)\mid R=1,Z,A,\boverbar{M}_{K},Y\right\}\vert_{t=0} = 0.
\eeqrs
Equivalently,
\beqrs
\int \frac{\partial \left\{\gamma_t(x,A,\boverbar{M}_{K},Y)-\beta_t(Z,A,\boverbar{M}_{K},Y)\right\}f_t(x\mid R=1,Z,A,\boverbar{M}_{K},Y)}{\partial t} \vert_{t=0}dx = 0.
\eeqrs
It follows that 
\begin{align*}
& E\{\gamma(X,A,\boverbar{M}_{K},Y) S(X \mid R=1, Z,A,\boverbar{M}_{K},Y) \mid R=1, Z,A,\boverbar{M}_{K},Y\}\nonumber\\ 
&~+\frac{S(R=1 \mid Z,A,\boverbar{M}_{K},Y)}{f(R=1 \mid Z,A,\boverbar{M}_{K},Y)}= - E\left\{\frac{\partial}{\partial t} \gamma_t(X,A,\boverbar{M}_{K},Y) \mid R=1, Z,A,\boverbar{M}_{K},Y\right\} .
\end{align*}
Equivalently,
\begin{align}
\label{eq:app_deriv_gamma}
& E\{R \gamma(X,A,\boverbar{M}_{K},Y) S(X \mid R=1, Z,A,\boverbar{M}_{K},Y) \mid Z,A,\boverbar{M}_{K},Y\}\nonumber\\ 
&~+ S(R=1 \mid Z,A,\boverbar{M}_{K},Y) = - E\left\{R\frac{\partial}{\partial t} \gamma_t(X,A,\boverbar{M}_{K},Y) \mid Z,A,\boverbar{M}_{K},Y\right\} .
\end{align}
From \eqref{eq:app_deriv_gamma}, we obtain that the tangent space is $\Lambda_1\bigoplus\Lambda_2$ with
\begin{align*}
\Lambda_1 = &~ \Big\{S(Z,A,\boverbar{M}_{K},Y) \in \calL_2(Z,A,\boverbar{M}_{K},Y): E\{S(Z,A,\boverbar{M}_{K},Y)\}=0\Big\}, \\
\Lambda_2 = &~ \bigg\{R S(X \mid R=1, Z,A,\boverbar{M}_{K},Y)\\
&\quad+\big\{R- {(1-R)}/{\beta(Z,A,\boverbar{M}_{K},Y)}\big\} S(R=1 \mid Z,A,\boverbar{M}_{K},Y)\in\Lambda_1^\bot: \\
&\qquad E\{S(X \mid R=1, Z,A,\boverbar{M}_{K},Y) \mid R=1, Z,A,\boverbar{M}_{K},Y\}=0, \text { and } \\
&\qquad E\{R \gamma(X,A,\boverbar{M}_{K},Y) S(X \mid R=1, Z,A,\boverbar{M}_{K},Y) \mid Z,A,\boverbar{M}_{K},Y\} \\ 
&\qquad\qquad + S(R=1 \mid Z,A,\boverbar{M}_{K},Y) \in \operatorname{cl}(\mathcal{R}(\mathcal{T}))\bigg\},
\end{align*}
where $\mathcal{T}$ is the operator that maps $g(X,A,\boverbar{M}_{K},Y)\in\calL_2(X,A,\boverbar{M}_{K},Y)$ to $E \{g(X,A,\boverbar{M}_{K},Y)\mid R=1,Z,A,\boverbar{M}_{K},Y \}\in\calL_2(Z,A,\boverbar{M}_{K},Y)$, $\mathcal{R}(\mathcal{T})$ is the range space of $\mathcal{T}$, $A^\bot$ denotes the orthogonal complement of $A$, and $cl(A)$ refers to the closure of $A$. 

In order to show $\wh\psi$ is efficient, we only need to show that the influence function derived in \Cref{thm:asy_normal} belongs to the tangent space above. For notational simplicity, below we use $\gamma$, $\phi$, $\beta$ and $\nu$ to denote $\gamma(X,A,\boverbar{M}_K,Y)$, $\phi(X,A,\boverbar{M}_K,Y)$, $\beta(Z,A,\boverbar{M}_K,Y)$ and $E(R\varrho_0\mid Z,A,\boverbar{M}_K,Y)$, respectively. Then we have the following decomposition of the influence function 
\begin{align*}
IF =&~ R(1+\gamma)\phi -\psi - \nu(R\gamma-1+R) \\ 
&=~ \frac{1}{\beta+1} E\{(1+\gamma)(\phi-\nu)\mid R=1,Z,A,\boverbar{M}_K,Y\}+\nu-\psi \\ 
&~ + R(1+\gamma)(\phi-\nu) - R E\{(1+\gamma)(\phi-\nu)\mid R=1,Z,A,\boverbar{M}_K,Y\} \\ 
&~ + \left(R-\frac{1-R}{\beta}\right) \frac{\beta}{\beta+1} E\{(1+\gamma)(\phi-\nu)\mid R=1,Z,A,\boverbar{M}_K,Y\}.
\end{align*}
Note that 
\beqrs
\frac{1}{\beta+1} = f(R=1\mid Z,A,\boverbar{M}_K,Y),
\eeqrs
then we have 
\beqrs
&& E\left[\frac{1}{\beta+1} E\{(1+\gamma)(\phi-\nu)\mid R=1,Z,A,\boverbar{M}_K,Y\}+\nu-\psi\right] \\ 
&=& E\Big[ E\{R(1+\gamma)(\phi-\nu)\mid Z,A,\boverbar{M}_K,Y\}+\nu-\psi\Big] \\ 
&=& E[R(1+\gamma)\phi + \nu (R\gamma-1+R) -\psi ] = E(IF)=0.
\eeqrs
It implies that 
\beqrs
\frac{1}{\beta+1} E[(1+\gamma)(\phi-\nu)\mid R=1,Z,A,\boverbar{M}_K,Y]+\nu-\psi \in \Lambda_1.
\eeqrs
According to \Cref{ass:completeness_inverse}, we have $\calN(\mathcal{T}\pprime)={0}$ and hence $cl(\calR(\mathcal{T}))=\calN(\mathcal{T}\pprime)^{\perp}=\calL_2(Z,A,\boverbar{M}_K,Y)$. Here $\mathcal{T}\pprime$ is the adjoint operator of $\mathcal{T}$ and $\calN(\mathcal{T}\pprime)$ is its null space. Take 
\begin{align*}
&S(X\mid R=1,Z,A,\boverbar{M}_K,Y) = (1+\gamma)(\phi-\nu) - E[(1+\gamma)(\phi-\nu)\mid R=1,Z,A,\boverbar{M}_K,Y],\\ 
&S(R=1\mid Z,A,\boverbar{M}_K,Y) = \frac{\beta}{\beta+1} E[(1+\gamma)(\phi-\nu)\mid R=1,Z,A,\boverbar{M}_K,Y].
\end{align*}
It is obvious that 
\beqrs
&&E[S(X\mid R=1,Z,A,\boverbar{M}_K,Y)\mid R=1,Z,A,\boverbar{M}_K,Y] = 0, \mbox{ and }\\ 
&&E\left[R(1+\gamma)(T,M,Y,X) S(X\mid R=1,Z,A,\boverbar{M}_K,Y) \mid Z,A,\boverbar{M}_K,Y\right] \\ 
&&\qquad\qquad+ S(R=1\mid Z,A,\boverbar{M}_K,Y) \in \calL_2(Z,A,\boverbar{M}_K,Y).
\eeqrs
It implies that 
\begin{align*}
& R(1+\gamma)(\phi-\nu) - R E[(1+\gamma)(\phi-\nu)\mid R=1,Z,A,\boverbar{M}_K,Y] \\ 
&\quad + \left(R-\frac{1-R}{\beta}\right) \frac{\beta}{\beta+1} E[(1+\gamma)(\phi-\nu)\mid R=1,Z,A,\boverbar{M}_K,Y] \in\Lambda_2.
\end{align*}
This completes the proof of \Cref{thm:EIF}.
\end{proof}

\subsection{Proof of Lemma \ref{lemma:iden_coef}}\label{ssec:app_pf_lemma1}
\begin{proof}
We first show that $\omega_1(x)\triangleq f_{A\mid X}(a_{1}\mid x)^{-1}$ satisfies the conditional moment restriction
\begin{align*}
E\big[\big\{1+\gamma(X,A,\boverbar{M}_K,Y)\big\}\big\{I(A=a_1)\omega_1(X)-1\big\}\mid R=1,X\big] =0.
\end{align*}
Note that it is equivalent to 
\begin{align*}
E\big[\big\{1+\gamma(X,A,\boverbar{M}_K,Y)\big\}I(A=a_1) \mid R=1,X\big]\omega_1(X) =	E\big\{1+\gamma(X,A,\boverbar{M}_K,Y)\mid R=1,X\big\}.
\end{align*}
So we only need to show that 
\beqrs
\omega_1(X) = \frac{E\big\{1+\gamma(X,A,\boverbar{M}_K,Y)\mid R=1,X\big\}}{E\big[\big\{1+\gamma(X,A,\boverbar{M}_K,Y)\big\}I(A=a_1)\mid R=1,X\big]}.
\eeqrs
By the definition of $\gamma(X,A,\boverbar{M}_K,Y)$, we obtain that 
\begin{align}
\label{app_lem1-1-2}
&E\big[\big\{1+\gamma(X,A,\boverbar{M}_K,Y)\big\}I(A=a_1) \mid R=1,X\big] \nonumber\\ 
=&~ \iint \frac{f(A = a_1, \boverbar{m}_K, y \mid R=1, X)}{f(R=1 \mid X, A=a_1, \boverbar{m}_K, y)} d\boverbar{m}_K dy \nonumber\\ 
=&~ f(R=1 \mid X)^{-1} \iint \frac{f(A = a_1, \boverbar{m}_K, y, R=1 \mid X)}{f(R=1 \mid X, A=a_1, \boverbar{m}_K, y)} d\boverbar{m}_K dy \nonumber\\ 
=&~ f(R=1 \mid X)^{-1} \iint f(A=a_1,\boverbar{m}_K, y \mid X) d\boverbar{m}_K dy \nonumber\\ 
=&~ f(R=1 \mid X)^{-1} f(A=a_1 \mid X).
\end{align} 
Similarly, we can deduce that 
\beqr\label{app_lem1-1-3}
E\big\{1+\gamma(X,A,\boverbar{M}_K,Y)\mid R=1,X\big\} = f(R=1 \mid X)^{-1}.
\eeqr
Combining \eqref{app_lem1-1-2} and \eqref{app_lem1-1-3} yields that 
\beqrs
\frac{E\big\{1+\gamma(X,A,\boverbar{M}_K,Y)\mid R=1,X\big\}}{E\big[\big\{1+\gamma(X,A,\boverbar{M}_K,Y)\big\}I(A=a_1)\mid R=1,X\big]}=f_{A\mid X}(a_{1}\mid X)^{-1}=\omega_1(X),
\eeqrs
which is the result we aim for.

Next we show that for $k\in\{2,\ldots,K+1\}$,
\beqrs
\omega_k(x,\boverbar{m}_{k-1}) \triangleq \frac{f_{A\mid X,\boverbar{M}_{k-1}}(a_{k-1}\mid x,\boverbar{m}_{k-1})}{f_{A\mid X,\boverbar{M}_{k-1}}(a_k\mid x,\boverbar{m}_{k-1})} 
\eeqrs
satisfies the conditional moment restriction
\begin{align*}
E\big[\big\{1+\gamma(X,A,\boverbar{M}_K,Y)\big\}\big\{I(A=a_k)\omega_k(X,\boverbar{M}_{k-1})-I(A=a_{k-1})\big\}\mid R=1,X,\boverbar{M}_{k-1}\big] =0.
\end{align*}
In the same way, the above moment restriction is equivalent to
\begin{align*}
E\big[\big\{1+\gamma(X,A,\boverbar{M}_K,Y)\big\}I(A=a_k)\mid R=1,X,\boverbar{M}_{k-1}\big]\omega_k(X,\boverbar{M}_{k-1}) \\
= E\big[\big\{1+\gamma(X,A,\boverbar{M}_K,Y)\big\}I(A=a_{k-1})\mid R=1,X,\boverbar{M}_{k-1}\big].
\end{align*}
So we only need to show that 
\begin{align*}
\omega_k(X,\boverbar{M}_{k-1}) = \frac{E\big[\big\{1+\gamma(X,A,\boverbar{M}_K,Y)\big\}I(A=a_{k-1})\mid R=1,X,\boverbar{M}_{k-1}\big]}{E\big[\big\{1+\gamma(X,A,\boverbar{M}_K,Y)\big\}I(A=a_k)\mid R=1,X,\boverbar{M}_{k-1}\big]}.
\end{align*}
We have
\beqrs
&&E\big[\big\{1+\gamma(X,A,\boverbar{M}_K,Y)\big\}I(A=a_k)\mid R=1,X,\boverbar{M}_{k-1}\big] \nonumber\\ 
&=& \iint \frac{f(A=a_k, \bunderline{m}_k, y \mid R=1, X, \boverbar{M}_{k-1})}{f(R=1 \mid X, A=a_k, \boverbar{M}_{k-1},\bunderline{m}_k, y)} d\bunderline{m}_k dy \nonumber\\ 
&=& f(R=1 \mid X, \boverbar{M}_{k-1})^{-1} \iint \frac{f(A=a_k, \bunderline{m}_k, y, R=1 \mid X, \boverbar{M}_{k-1})}{f(R=1 \mid X, A=a_k, \boverbar{M}_{k-1},\bunderline{m}_k, y)} d\bunderline{m}_k dy \nonumber\\ 
&=& f(R=1 \mid X, \boverbar{M}_{k-1})^{-1} \iint f(A=a_k, \bunderline{m}_k, y \mid X, \boverbar{M}_{k-1}) d\bunderline{m}_k dy \nonumber\\ 
&=& f(R=1 \mid X, \boverbar{M}_{k-1})^{-1} f(A=a_k \mid X, \boverbar{M}_{k-1}).
\eeqrs
Similarly, 
\beqrs 
&& E\big[\big\{1+\gamma(X,A,\boverbar{M}_K,Y)\big\}I(A=a_{k-1})\mid R=1,X,\boverbar{M}_{k-1}\big] \\
&=& f(R=1 \mid X, \boverbar{M}_{k-1})^{-1} f(A=a_{k-1} \mid X, \boverbar{M}_{k-1}).
\eeqrs
It follows that 
\beqrs
&& \frac{E\big[\big\{1+\gamma(X,A,\boverbar{M}_K,Y)\big\}I(A=a_{k-1})\mid R=1,X,\boverbar{M}_{k-1}\big]}{E\big[\big\{1+\gamma(X,A,\boverbar{M}_K,Y)\big\}I(A=a_k)\mid R=1,X,\boverbar{M}_{k-1}\big]}\\
&=&\frac{f_{A\mid X,\boverbar{M}_{k-1}}(a_{k-1}\mid x,\boverbar{m}_{k-1})}{f_{A\mid X,\boverbar{M}_{k-1}}(a_k\mid x,\boverbar{m}_{k-1})} 
= \omega_k(x,\boverbar{m}_{k-1}) .
\eeqrs
This completes the proof.
\end{proof}

\section{Auxillary lemmas}\label{sec:app_Aux_lemmas}

\begin{lemma}\label{lemma:lemmaC2_of_ChenPouzo2012}
Suppose \Cref{ass:supp_dgp,ass:supp_sieve_delta_inter,ass:supp_sieve_delta_exter} hold. Then there exists finite constants $C,C\pprime>0$ such that
\begin{align*}
&C E[\{E(R\wt\gamma-1+R\mid Z,A,\boverbar{M}_K,Y)\}^2] - O_p(l_n/n+l_n^{-2\alpha_2}) \\ 
&\qquad \le \frac{1}{n}\sum_{i=1}^n \{\wh{E}(R\wt\gamma-1+R\mid Z_i,A_i,\boverbar{M}_{K,i},Y_i)\}^2 \\ 
&\qquad\qquad \le C\pprime E[\{E(R\wt\gamma-1+R\mid Z,A,\boverbar{M}_K,Y)\}^2] + O_p(l_n/n+l_n^{-2\alpha_2}),
\end{align*}
uniformly over $\wt\gamma\in\Gamma_n$.
\end{lemma}
\begin{proof}
The proof proceeds by verifying the conditions of Lemma C.2(ii) of \cite{ChenPouzo2012}. Assumption C.2(i) of \cite{ChenPouzo2012} is satisfied since $|R\gamma-1+R|$ is uniformly bounded over $\gamma\in\Gamma_n\subset\Gamma$ by \Cref{ass:supp_dgp}(i) and \Cref{ass:supp_sieve_delta_inter}(i). Assumption C.2(i) of \cite{ChenPouzo2012} is satisfied by \Cref{ass:supp_sieve_delta_exter}(ii) with their $b^2_{m,J_n}$ being $l_n^{-2\alpha_2}$ here. Denote 
\beqr
\calO_j = \{f_j = p_j(z,a,\boverbar{m}_K,y)(r\gamma-1+r):\gamma\in\Gamma\},\quad j=1,\ldots,l_n.
\eeqr
Note that $f_j$ is linear and hence Lipschitz continuous in $\gamma\in\Gamma$ under $\|\cdot\|_\infty$, and $$\max_{1\le j\le l_n}E[p_j(Z,A,\boverbar{M}_K,Y)^2]\lesssim 1,$$ by \Cref{ass:supp_sieve_delta_exter}(i). Then Theorem 2.7.11 of \cite{Vaart_ep_1996} implies that $N_{[]}(\epsilon,\calO_j,\|\cdot\|_{\calL_2})\le N(\epsilon/C,\Gamma,\|\cdot\|_\infty)$ for some constant $C>0$. Thus, by \Cref{ass:supp_dgp} and then Theorem 2.5.6 and 2.7.1 of \cite{Vaart_ep_1996}, we have 
\beqr\label{eq:app_lemmaC2_of_ChenPouzo2012_pf1}
\max_{1\le j\le l_n} J_{[]}(1,\calO_j,\|\cdot\|_{\calL_2}) \lesssim 1.
\eeqr
Thus, \eqref{eq:app_lemmaC2_of_ChenPouzo2012_pf1} implies that Assumption C.2(iii) of \cite{ChenPouzo2012} is satisfied with their $C_n\lesssim 1$. In addition, Assumption C.1 of \cite{ChenPouzo2012} is satisfied by \Cref{ass:supp_dgp}(iii) and \Cref{ass:supp_sieve_delta_exter}(i)(iii). Thus, the result of the lemma follows by Lemma C.2(ii) of \cite{ChenPouzo2012}. 
\end{proof}

\begin{lemma}\label{lemma:app_delta}
Suppose \Cref{ass:supp_dgp,ass:supp_sieve_delta_exter} hold. We have 
\beqrs
\frac{1}{n}\sum_{i=1}^{n} \{\wh{E}(R\wt\gamma-1+R\mid Z_i,A_i,\boverbar{M}_{K,i},Y_i)\}^2 = O_p(l_n/n+l_n^{-2\alpha_2})+o_p(n^{-1/2}),
\eeqrs
uniformly over $\wt\gamma\in\{\wt\gamma\in\Gamma:\|\wt\gamma-\gamma\|_w=o_p(n^{-1/4})\}$. 
\end{lemma}

\begin{proof}
For any $\wt\gamma\in\Gamma$, by the triangle inequality, we have 
\begin{align} 
\label{eq:pf_lemma_app_delta_1}
&\frac{1}{n}\sum_{i=1}^{n} \{\wh{E}(R\wt\gamma-1+R\mid Z_i,A_i,\boverbar{M}_{K,i},Y_i)\}^2 \nonumber\\
\le&~ \frac{2}{n}\sum_{i=1}^{n} \big\{E(R\wt\gamma-1+R\mid Z_i,A_i,\boverbar{M}_{K,i},Y_i)\big\}^2 \nonumber\\ 
&\quad +\frac{4}{n}\sum_{i=1}^{n} \big\{\wh{E}(R\wt\gamma\mid Z_i,A_i,\boverbar{M}_{K,i},Y_i)-E(R\wt\gamma\mid Z_i,A_i,\boverbar{M}_{K,i},Y_i)\big\}^2 \nonumber\\
&\quad +\frac{4}{n}\sum_{i=1}^{n} \big\{\wh{E}(1-R\mid Z_i,A_i,\boverbar{M}_{K,i},Y_i)-E(1-R\mid Z_i,A_i,\boverbar{M}_{K,i},Y_i)\big\}^2.
\end{align}
Define the class of functions $\calG=\{E(R\wt\gamma-1+R\mid z,a,\boverbar{m}_K,y)^2: \wt\gamma\in\Gamma\}$. Note that for every $\gamma_1,\gamma_2\in\Gamma$, we have 
\beqrs
&& \Abs{E(R\gamma_1-1+R\mid z,a,\boverbar{m}_K,y)^2-E(R\gamma_2-1+R\mid z,a,\boverbar{m}_K,y)^2} \\ 
&\le& \Abs{E\{R(\gamma_1-\gamma_2)\mid z,a,\boverbar{m}_K,y\}}\times \Abs{E\{R(\gamma_1+\gamma_2)-2+2R\mid z,a,\boverbar{m}_K,y\}} \\ 
&\lesssim& \|\gamma_1-\gamma_2\|_\infty.
\eeqrs
Then by \Cref{ass:supp_dgp} and Theorem 2.5.6, 2.7.1 and 2.7.11 of \cite{Vaart_ep_1996}, we conclude $\calG$ is a Donsker class. It follows that 
\begin{align}
\label{eq:pf_lemma_app_delta_2}
\frac{1}{n}\sum_{i=1}^{n} \big\{E(R\wt\gamma-1+R\mid Z_i,A_i,\boverbar{M}_{K,i},Y_i)\big\}^2 - E[\{E(R\wt\gamma-1+R\mid Z,A,\boverbar{M}_K,Y)\}^2] = o_p(n^{-1/2}),
\end{align}
uniformly over $\wt\gamma\in\Gamma$. By the definition of $\|\cdot\|_w$ and \eqref{eq:iden_gamma}, we know 
\beqr\label{eq:pf_lemma_app_delta_3}
E[\{E(R\wt\gamma-1+R\mid Z,A,\boverbar{M}_K,Y)\}^2] = \|\wt\gamma-\gamma\|^2_w.
\eeqr
Combining \eqref{eq:pf_lemma_app_delta_2} and \eqref{eq:pf_lemma_app_delta_3} yields that 
\beqr\label{eq:pf_lemma_app_delta_4}
\frac{1}{n}\sum_{i=1}^{n} \{E(R\wt\gamma-1+R\mid Z_i,A_i,\boverbar{M}_{K,i},Y_i)\}^2 = o_p(n^{-1/2}),
\eeqr
uniformly over $\wt\gamma\in\{\wt\gamma\in\Gamma:\|\wt\gamma-\gamma\|_w=o_p(n^{-1/4})\}$. By \Cref{lemma:app_delta_exter_bound}, we have 
\beqr\label{eq:pf_lemma_app_delta_5}
\frac{1}{n}\sum_{i=1}^{n} \big\{\wh{E}(R\wt\gamma\mid Z_i,A_i,\boverbar{M}_{K,i},Y_i)-E(R\wt\gamma\mid Z_i,A_i,\boverbar{M}_{K,i},Y_i)\big\}^2 = O_p(l_n/n+l_n^{-2\alpha_2}),
\eeqr
uniformly over $\wt\gamma\in\Gamma$. Following \cite{Newey1997}, we have 
\begin{align}
\label{eq:pf_lemma_app_delta_6}
\frac{1}{n}\sum_{i=1}^{n} \big\{\wh{E}(1-R\mid Z_i,A_i,\boverbar{M}_{K,i},Y_i)-E(1-R\mid Z_i,A_i,\boverbar{M}_{K,i},Y_i)\big\}^2 = O_p(l_n/n+l_n^{-2\alpha_2}).
\end{align}
Then the proof is completed by combining \eqref{eq:pf_lemma_app_delta_1}, \eqref{eq:pf_lemma_app_delta_4}--\eqref{eq:pf_lemma_app_delta_6}.
\end{proof}

\begin{lemma}\label{lemma:app_rho}
Suppose \Cref{ass:supp_dgp,ass:supp_sieve_delta_inter,ass:supp_sieve_delta_exter} hold. We have 
\beqrs
\frac{1}{n}\sum_{i=1}^{n} \big\{\wh{E}(R\varrho_n\mid Z_i,A_i,\boverbar{M}_{K,i},Y_i) - E(R\varrho_0\mid Z_i,A_i,\boverbar{M}_{K,i},Y_i)\big\}^2 = O_p(\sqrt{l_n/n}+l_n^{-\alpha_2}+s_n^{-\alpha_1}).
\eeqrs
\end{lemma}

\begin{proof}
By Cauchy Schwarz inequality, we have 
\begin{align}
\label{eq:pf_lemma_app_rho_1}
&\frac{1}{n}\sum_{i=1}^{n} \big\{\wh{E}(R\varrho_n\mid Z_i,A_i,\boverbar{M}_{K,i},Y_i) - E(R\varrho_0\mid Z_i,A_i,\boverbar{M}_{K,i},Y_i)\big\}^2 \nonumber\\
\le&~ \frac{2}{n}\sum_{i=1}^{n} \big\{\wh{E}(R\varrho_n\mid Z_i,A_i,\boverbar{M}_{K,i},Y_i)-E(R\varrho_n\mid Z_i,A_i,\boverbar{M}_{K,i},Y_i)\big\}^2 \nonumber\\ 
&\quad + \frac{2}{n}\sum_{i=1}^{n} \big[E\{R(\varrho_n-\varrho_0)\mid Z_i,A_i,\boverbar{M}_{K,i},Y_i\}\big]^2.
\end{align}
By the definition of $\varrho_n=\Pi_n\varrho_0$ and \Cref{ass:supp_sieve_delta_inter}(ii), we have
\beqr\label{eq:pf_lemma_app_rho_2}
\frac{2}{n}\sum_{i=1}^{n} \big[E\{R(\varrho_n-\varrho_0)\mid Z_i,A_i,\boverbar{M}_{K,i},Y_i\}\big]^2 \lesssim \|\varrho_n-\varrho_0\|^2_\infty =O_p(s_n^{-2\alpha_1}).
\eeqr
By \Cref{lemma:app_delta_exter_bound}, $\varrho_n\in\Gamma_n\subset\Gamma$, and \cref{ass:supp_sieve_delta_exter}, we have 
\begin{align}
\label{eq:pf_lemma_app_rho_3}
\frac{1}{n}\sum_{i=1}^{n} \big\{\wh{E}(R\varrho_n\mid Z_i,A_i,\boverbar{M}_{K,i},Y_i)-E(R\varrho_n\mid Z_i,A_i,\boverbar{M}_{K,i},Y_i)\big\}^2 = O_p(l_n/n+l_n^{-2\alpha_2}).
\end{align}
Combining \eqref{eq:pf_lemma_app_rho_1}, \eqref{eq:pf_lemma_app_rho_2} and \eqref{eq:pf_lemma_app_rho_3} yields the result of the lemma.
\end{proof}

\begin{lemma}\label{lemma:app_rho_hat}
Suppose \Cref{ass:supp_dgp,ass:supp_sieve_delta_exter} hold. We have 
\begin{align*}
&\frac{1}{n}\sum_{i=1}^n \{\wh{E}^{(d)}(R\varrho_0\mid Z_i,A_i,\boverbar{M}_{K,i},Y_i) - E(R\varrho_0\mid Z_i,A_i,\boverbar{M}_{K,i},Y_i)\} \{R_i\wt\gamma(X_i,A_i,\boverbar{M}_{K,i},Y_i)-1+R_i\}\\
&\qquad =O_p\{l_n^{-\alpha_2}(\sqrt{l_n/n}+l_n^{-\alpha_2})\}+o_p(n^{-1/2}), 
\end{align*}
uniformly over $\wt\gamma\in\{\wt\gamma\in\Gamma:\|\wt\gamma-\gamma\|_w=o_p(n^{-1/4})\}$, where the notation $\wh{E}^{(d)}$ is defined in \eqref{eq:app_delta_5}.
\end{lemma}

\begin{proof}
For notational simplicity, we denote $g(Z,A,\boverbar{M}_K,Y)\triangleq E(R\varrho_0\mid Z,A,\boverbar{M}_K,Y)$ and $\wh{g}(Z,A,\boverbar{M}_K,Y)\triangleq \wh{E}^{(d)}(R\varrho_0\mid Z,A,\boverbar{M}_K,Y)$. 
By \Cref{ass:supp_sieve_delta_exter}(ii), there exists $g_n(z,a,\boverbar{m}_K,y)=\bbar{p}_{l_n}(z,a,\boverbar{m}_K,y)\trans{\bm\pi}_g$ such that 
\beqr\label{eq:lemma_app_rho_hat_pf1}
\sup_{(z,a,\boverbar{m}_K,y)\in\calZ\times\{0,1\}\times\boverbar\calM_K\times\calY} \Abs{g(z,a,\boverbar{m}_K,y)-g_n(z,a,\boverbar{m}_K,y)} = O(l_n^{-\alpha_2}).
\eeqr
Note that $\wh{g}(z,a,\boverbar{m}_K,y)=\bbar{p}_{l_n}(z,a,\boverbar{m}_K,y)\trans\wh{\bm\pi}_g$ with 
\begin{small}
\begin{align}
\wh{\bm\pi}_g = \Big\{\sum_{i=1}^n \bbar{p}_{l_n}(Z_i,A_i,\boverbar{M}_{K,i},Y_i) \bbar{p}_{l_n}(Z_i,A_i,\boverbar{M}_{K,i},Y_i)\trans\Big\}^{-1}\sum_{i=1}^n \bbar{p}_{l_n}(Z_i,A_i,\boverbar{M}_{K,i},Y_i)g(Z_i,A_i,\boverbar{M}_{K,i},Y_i), \label{eq:lemma_app_rho_hat_pf2}
\end{align}
\end{small}
Then by \Cref{ass:supp_sieve_delta_exter}(i), with probability approaching one 
$$\Big\{\sum_{i=1}^n \bbar{p}_{l_n}(Z_i,A_i,\boverbar{M}_{K,i},Y_i) \bbar{p}_{l_n}(Z_i,A_i,\boverbar{M}_{K,i},Y_i)\trans\Big\}$$ is invertible and
\begin{align}
\|\wh{\bm\pi}_g-{\bm\pi}_g\| \lesssim \sup_{(z,a,\boverbar{m}_K,y)\in\calZ\times\{0,1\}\times\boverbar\calM_K\times\calY} \Abs{g(z,a,\boverbar{m}_K,y)-g_n(z,a,\boverbar{m}_K,y)} = O_p(l_n^{-\alpha_2}). \label{eq:lemma_app_rho_hat_pf4}
\end{align}
It follows that 
\begin{align}
\label{eq:lemma_app_rho_hat_pf6}
&\frac{1}{n}\sum_{i=1}^{n} \bigbc{\wh{g}(Z_i,A_i,\boverbar{M}_{K,i},Y_i) - g(Z_i,A_i,\boverbar{M}_{K,i},Y_i)}^2 \nonumber\\ 
\lesssim&~ \|\wh{\bm\pi}_g-{\bm\pi}_g\|^2 + \sup_{(z,a,\boverbar{m}_K,y)\in\calZ\times\{0,1\}\times\boverbar\calM_K\times\calY} \Abs{g(z,a,\boverbar{m}_K,y)-g_n(z,a,\boverbar{m}_K,y)}^2=O_p(l_n^{-2\alpha_2}) 
\end{align}
Note that 
\begin{align}
\label{eq:lemma_app_rho_hat_pf7}
&\frac{1}{n}\sum_{i=1}^n \left\{\wh{g}(Z_i,A_i,\boverbar{M}_{K,i},Y_i) - g(Z_i,A_i,\boverbar{M}_{K,i},Y_i)\right\} (R_i\wt\gamma_i-1+R_i)\nonumber\\ 
=&~ \frac{1}{n}\sum_{i=1}^n \left\{\wh{g}(Z_i,A_i,\boverbar{M}_{K,i},Y_i) - g(Z_i,A_i,\boverbar{M}_{K,i},Y_i)\right\} E(R\wt\gamma-1+R\mid Z_i,A_i,\boverbar{M}_{K,i},Y_i)\nonumber\\ 
&+ \frac{1}{n}\sum_{i=1}^n \left\{\wh{g}(Z_i,A_i,\boverbar{M}_{K,i},Y_i) - g(Z_i,A_i,\boverbar{M}_{K,i},Y_i)\right\} \nonumber\\ 
&\qquad\qquad\quad\times\{R_i\wt\gamma_i-1+R_i-E(R\wt\gamma-1+R\mid Z_i,A_i,\boverbar{M}_{K,i},Y_i)\}.
\end{align}
Combining \eqref{eq:pf_lemma_app_delta_4} and \eqref{eq:lemma_app_rho_hat_pf6}, by Cauchy-Schwarz inequality, yields that 
\beqr\label{eq:lemma_app_rho_hat_pf8}
&& \frac{1}{n}\sum_{i=1}^n \left\{\wh{g}(Z_i,A_i,\boverbar{M}_{K,i},Y_i) - g(Z_i,A_i,\boverbar{M}_{K,i},Y_i)\right\} E(R\wt\gamma-1+R\mid Z_i,A_i,\boverbar{M}_{K,i},Y_i)\nonumber\\ 
&&\quad = o_p(n^{-1/2}) + O_p(l_n^{-2\alpha_2}),
\eeqr
uniformly over $\wt\gamma\in\{\wt\gamma\in\Gamma:\|\wt\gamma-\gamma\|_w=o_p(n^{-1/4})\}$.
Note that 
\begin{small}
\begin{align}
\label{eq:lemma_app_rho_hat_pf9}
&\Abs{\frac{1}{n}\sum_{i=1}^n \left\{\wh{g}(Z_i,A_i,\boverbar{M}_{K,i},Y_i) - g(Z_i,A_i,\boverbar{M}_{K,i},Y_i)\right\} \{R_i\wt\gamma_i-1+R_i-E(R\wt\gamma-1+R\mid Z_i,A_i,\boverbar{M}_{K,i},Y_i)\}} \nonumber\\ 
\le&~ \left\|\wh{\bm\pi}_g-{\bm\pi}_g\right\| \Big\| \frac{1}{n}\sum_{i=1}^n \bbar{p}_{l_n}(Z_i,A_i,\boverbar{M}_{K,i},Y_i) \{R_i\wt\gamma_i-1+R_i-E(R\wt\gamma-1+R\mid Z_i,A_i,\boverbar{M}_{K,i},Y_i)\} \Big\| \nonumber\\
&+ \Abs{\frac{1}{n}\sum_{i=1}^n \left\{g(Z_i,A_i,\boverbar{M}_{K,i},Y_i) - g_n(Z_i,A_i,\boverbar{M}_{K,i},Y_i)\right\} \{R_i\wt\gamma_i-1+R_i-E(R\wt\gamma-1+R\mid Z_i,A_i,\boverbar{M}_{K,i},Y_i)\}}.
\end{align}
\end{small}
By \Cref{ass:supp_dgp} and Theorem 2.5.6, 2.7.1 and 2.7.11 of \cite{Vaart_ep_1996}, we conclude $\calG_j=\{tp_j(z,a,\boverbar{m}_K,y)\{r\wt\gamma-E(R\wt\gamma\mid z,a,\boverbar{m}_K,y)\}:\wt\gamma\in\Gamma\}$, $j=1,\ldots,l_n$, are Donsker. It implies that
\begin{align*}
&&\frac{1}{n}\sum_{i=1}^n p_j(M_i,Y_i,Z_i) \{R_i\wt\gamma_i-1+R_i-E(R\wt\gamma-1+R\mid Z_i,A_i,\boverbar{M}_{K,i},Y_i)\} = O_p(n^{-1/2}),
\end{align*}
uniformly over $\wt\gamma\in\Gamma$. Then we obtain 
\begin{small}
\begin{align}
\label{eq:lemma_app_rho_hat_pf11}
\Big\| \frac{1}{n}\sum_{i=1}^n \bbar{p}_{l_n}(Z_i,A_i,\boverbar{M}_{K,i},Y_i) \{R_i\wt\gamma_i-1+R_i-E(R\wt\gamma-1+R\mid Z_i,A_i,\boverbar{M}_{K,i},Y_i)\} \Big\| = O_p(\sqrt{l_n/n}),
\end{align}
\end{small}
uniformly over $\wt\gamma\in\Gamma$. By Chebyshev's inequality, we obtain 
\begin{small}
\begin{align}
\label{eq:lemma_app_rho_hat_pf13}
&\Abs{\frac{1}{n}\sum_{i=1}^n \left\{g(Z_i,A_i,\boverbar{M}_{K,i},Y_i) - g_n(Z_i,A_i,\boverbar{M}_{K,i},Y_i)\right\} \{R_i\wt\gamma_i-1+R_1-E(R\wt\gamma-1+R\mid Z_i,A_i,\boverbar{M}_{K,i},Y_i)\}} \nonumber\\ 
&\quad = O_p(n^{-1/2}l_n^{-\alpha_2}).
\end{align}
\end{small}
Combining \eqref{eq:lemma_app_rho_hat_pf7}, \eqref{eq:lemma_app_rho_hat_pf11} and \eqref{eq:lemma_app_rho_hat_pf13} yields that
\begin{small}
\begin{align}
\label{eq:lemma_app_rho_hat_pf14}
&\Abs{\frac{1}{n}\sum_{i=1}^n \left\{\wh{g}(Z_i,A_i,\boverbar{M}_{K,i},Y_i) - g(Z_i,A_i,\boverbar{M}_{K,i},Y_i)\right\} \{R_i\wt\gamma_i-1+R_1-E(R\wt\gamma-1+R\mid Z_i,A_i,\boverbar{M}_{K,i},Y_i)\}}\nonumber\\ 
&\quad = O_p(l_n^{-\alpha_2}\sqrt{l_n/n}),
\end{align}
\end{small}
uniformly over $\wt\gamma\in\Gamma$. Combining \eqref{eq:lemma_app_rho_hat_pf7}, \eqref{eq:lemma_app_rho_hat_pf8} and \eqref{eq:lemma_app_rho_hat_pf14} yields the result of the lemma.
\end{proof}

\begin{lemma}\label{lemma:app_delta_exter_bound}
Suppose \Cref{ass:supp_dgp,ass:supp_sieve_delta_exter} hold. Then we have 
\beqrs
&&\sup_{(z,a,\boverbar{m}_K,y)\in\calZ\times\{0,1\}\times\boverbar\calM_K\times\calY} \Abs{\wh{E}(R\wt\gamma\mid z,a,\boverbar{m}_K,y)- E(R\wt\gamma\mid z,a,\boverbar{m}_K,y) } \\ 
&&\qquad\qquad\qquad\qquad\qquad\qquad\qquad\qquad= O_p( \zeta_{1n}(\sqrt{l_n/n}+l_n^{-\alpha_2})+l_n^{-\alpha_2}),
\eeqrs
and 
\beqrs
\frac{1}{n}\sum_{i=1}^{n} \bigbc{\wh{E}(R\wt\gamma\mid Z_i,A_i,\boverbar{M}_{K,i},Y_i)- E(R\wt\gamma\mid Z_i,A_i,\boverbar{M}_{K,i},Y_i) }^2 = O_p( l_n/n + l_n^{-2\alpha_2}),
\eeqrs
uniformly over $\wt\gamma\in\Gamma$.

\end{lemma}

\begin{proof}
For any $\wt\gamma\in\Gamma$, let
\beqrs 
\wh{\bm\pi}_{\wt\gamma}=(\P\trans \P)^{-1} \sum_{i=1}^{n} \bbar{p}_{l_n}(Z_i,A_i,\boverbar{M}_{K,i},Y_i)R_i\wt\gamma(X_i,A_i,\boverbar{M}_{K,i},Y_i).
\eeqrs
Then $\wh{E}(R\wt\gamma\mid z,a,\boverbar{m}_K,y) = \bbar{p}_{l_n}(z,a,\boverbar{m}_K,y)\trans \wh{\bm\pi}_{\wt\gamma}$. By \Cref{ass:supp_sieve_delta_exter}(ii), there exists ${\bm\pi}_{\wt\gamma}\in\mR^{l_n}$ such that 
\begin{align}
\label{eq:supp_pf_delta_exter_bound_1}
\sup_{(z,a,\boverbar{m}_K,y)\in\calZ\times\{0,1\}\times\boverbar\calM_K\times\calY} \Abs{E(R\wt\gamma\mid z,a,\boverbar{m}_K,y) - \bbar{p}_{l_n}(z,a,\boverbar{m}_K,y)\trans {\bm\pi}_{\wt\gamma} } = O(l_n^{-\alpha_2}),
\end{align}
uniformly over $\wt\gamma\in\Gamma$. By the triangle inequality and \Cref{ass:supp_sieve_delta_exter}(i)(iii), it follows 
\begin{align}
\label{eq:supp_pf_delta_exter_bound_2.1}
&\sup_{(z,a,\boverbar{m}_K,y)\in\calZ\times\{0,1\}\times\boverbar\calM_K\times\calY} \Abs{\wh{E}(R\wt\gamma\mid z,a,\boverbar{m}_K,y)- E(R\wt\gamma\mid z,a,\boverbar{m}_K,y) } \nonumber\\ 
\le&~ \|\wh{\bm\pi}_{\wt\gamma}-{\bm\pi}_{\wt\gamma} \|\zeta_{1n}
+ \sup_{(z,a,\boverbar{m}_K,y)\in\calZ\times\{0,1\}\times\boverbar\calM_K\times\calY} \Abs{E(R\wt\gamma\mid z,a,\boverbar{m}_K,y) - \bbar{p}_{l_n}(z,a,\boverbar{m}_K,y)\trans {\bm\pi}_{\wt\gamma} },
\end{align}
and 
\begin{align}
\label{eq:supp_pf_delta_exter_bound_2.2}
&\frac{1}{n}\sum_{i=1}^{n} \bigbc{\wh{E}(R\wt\gamma\mid Z_i,A_i,\boverbar{M}_{K,i},Y_i)- E(R\wt\gamma\mid Z_i,A_i,\boverbar{M}_{K,i},Y_i) }^2 \nonumber\\ 
\lesssim&~ \|\wh{\bm\pi}_{\wt\gamma}-{\bm\pi}_{\wt\gamma} \|^2 \mbox{Eig}_{\max}(E\{\bbar{p}_{l_n}(Z,A,\boverbar{M}_K,Y)\bbar{p}_{l_n}(Z,A,\boverbar{M}_K,Y)\trans\}) \nonumber\\
&+ \sup_{(z,a,\boverbar{m}_K,y)\in\calZ\times\{0,1\}\times\boverbar\calM_K\times\calY} \Abs{E(R\wt\gamma\mid z,a,\boverbar{m}_K,y) - \bbar{p}_{l_n}(z,a,\boverbar{m}_K,y)\trans {\bm\pi}_{\wt\gamma}}^2.
\end{align}
Following \cite{Newey1997}, by \Cref{ass:supp_sieve_delta_exter}(i), with probability approaching one $(\P\trans \P)/n$ is invertible and 
\begin{align}
\label{eq:supp_pf_delta_exter_bound_3}
&\|\wh{\bm\pi}_{\wt\gamma}-{\bm\pi}_{\wt\gamma} \| \lesssim \sup_{(z,a,\boverbar{m}_K,y)\in\calZ\times\{0,1\}\times\boverbar\calM_K\times\calY} \Abs{E(R\wt\gamma\mid z,a,\boverbar{m}_K,y) - \bbar{p}_{l_n}(z,a,\boverbar{m}_K,y)\trans {\bm\pi}_{\wt\gamma}} \nonumber\\ 
&\quad+ \Big\| \frac{1}{n}\sum_{i=1}^{n} \bbar{p}_{l_n}(Z_i,A_i,\boverbar{M}_{K,i},Y_i)\left\{R_i\wt\gamma(X_i,A_i,\boverbar{M}_{K,i},Y_i)-E(R\wt\gamma\mid Z_i,A_i,\boverbar{M}_{K,i},Y_i)\right\} \Big\|.
\end{align}
uniformly over $\wt\gamma\in\Gamma$. Note that $p_j(z,a,\boverbar{m}_K,y)r\wt\gamma$ is linear and hence Lipschitz continuous in $\wt\gamma\in\Gamma$ under $\|\cdot\|_\infty$. Then by \Cref{ass:supp_dgp} and Theorem 2.5.6, 2.7.1 and 2.7.11 of \cite{Vaart_ep_1996}, we conclude $\calG_j=\{p_j(z,a,\boverbar{m}_K,y)r\wt\gamma: \wt\gamma\in\Gamma\}$, $j=1,\ldots,l_n$ is Donsker. It follows that 
\begin{small}
\begin{align}
\label{eq:supp_pf_delta_exter_bound_4}
\Big\| \frac{1}{n}\sum_{i=1}^{n} T_i \bbar{p}_{l_n}(Z_i,A_i,\boverbar{M}_{K,i},Y_i)\left\{R_i\wt\gamma(X_i,A_i,\boverbar{M}_{K,i},Y_i)-E(R\wt\gamma\mid Z_i,A_i,\boverbar{M}_{K,i},Y_i)\right\} \Big\| = O_p( \sqrt{l_n/n}),
\end{align}
\end{small}
uniformly over $\wt\gamma\in\Gamma$. Combining \eqref{eq:supp_pf_delta_exter_bound_1}, \eqref{eq:supp_pf_delta_exter_bound_2.1}, \eqref{eq:supp_pf_delta_exter_bound_3} and \eqref{eq:supp_pf_delta_exter_bound_4} yields
\beqrs
&&\sup_{(z,a,\boverbar{m}_K,y)\in\calZ\times\{0,1\}\times\boverbar\calM_K\times\calY} \Abs{\wh{E}(R\wt\gamma\mid z,a,\boverbar{m}_K,y)- E(R\wt\gamma\mid z,a,\boverbar{m}_K,y) } \\ 
&&\qquad\qquad\qquad\qquad\qquad\qquad\qquad\qquad= O_p( \zeta_{1n}(\sqrt{l_n/n}+l_n^{-\alpha_2})+l_n^{-\alpha_2}),
\eeqrs
uniformly over $\wt\gamma\in\Gamma$. And Combining \eqref{eq:supp_pf_delta_exter_bound_1}, \eqref{eq:supp_pf_delta_exter_bound_2.2}, \eqref{eq:supp_pf_delta_exter_bound_3} and \eqref{eq:supp_pf_delta_exter_bound_4} yields
\beqrs\label{eq:supp_pf_delta_exter_bound_0.2}
\frac{1}{n}\sum_{i=1}^{n} \bigbc{\wh{E}(R\wt\gamma\mid Z_i,A_i,\boverbar{M}_{K,i},Y_i)- E(R\wt\gamma\mid Z_i,A_i,\boverbar{M}_{K,i},Y_i) }^2 = O_p( l_n/n + l_n^{-2\alpha_2}),
\eeqrs
uniformly over $\wt\gamma\in\Gamma$.
This completes the proof of the lemma.

\end{proof}

\begin{lemma}\label{lemma:supp_stoc_ec}
Suppose \Cref{ass:supp_dgp,ass:supp_sieve_strengthened_rate,ass:supp_sieve_mu} hold. We have
\begin{align*} 
&\frac{1}{n}\sum_{i=1}^{n} R_i\big\{\wh\mu_1(X_i)-\mu_1(X_i)\big\}+\frac{1}{n}\sum_{i=1}^{n} R_i\big\{\wt\gamma(X_i,A_i,\boverbar{M}_{K,i},Y_i)\wh\mu_1(X_i)-\gamma(X_i,A_i,\boverbar{M}_{K,i},Y_i)\mu_1(X_i)\big\} \nonumber\\ 
&- E\Big[ R\big\{\wh\mu_1(X)-\mu_1(X)\big\}+ R\big\{\wt\gamma(X,A,\boverbar{M}_{K},Y)\wh\mu_1(X)-\gamma(X,A,\boverbar{M}_{K},Y)\mu_1(X)\big\}\Big] \nonumber\\ 
& + \frac{1}{n}\sum_{i=1}^{n} \varpi(R_i,Z_i,X_i,A_i,\boverbar{M}_{K,i},Y_i;\wt\gamma,\wh\mu) - \frac{1}{n}\sum_{i=1}^{n} \varpi(R_i,Z_i,X_i,A_i,\boverbar{M}_{K,i},Y_i;\gamma,\mu) \nonumber\\ 
& - E\Big[\varpi(R,Z,X,A,\boverbar{M}_{K},Y;\wt\gamma,\wh\mu) - \varpi(R,Z,X,A,\boverbar{M}_{K},Y;\gamma,\mu)\Big]= o_p(n^{-1/2}).
\end{align*}
uniformly over $\wt\gamma\in\{\gamma\in\Gamma:\|\wt\gamma-\gamma\|_\infty=o_p(1)\}$.

\end{lemma}

\begin{proof}
We first notice that 
\begin{small}
\begin{align} 
& \frac{1}{n}\sum_{i=1}^{n} R_i\big\{\wh\mu_1(X_i)-\mu_1(X_i)\big\}+\frac{1}{n}\sum_{i=1}^{n} R_i\big\{\wt\gamma(X_i,A_i,\boverbar{M}_{K,i},Y_i)\wh\mu_1(X_i)-\gamma(X_i,A_i,\boverbar{M}_{K,i},Y_i)\mu_1(X_i)\big\} \nonumber\\ 
&- E\Big[ R\big\{\wh\mu_1(X)-\mu_1(X)\big\}+ R\big\{\wt\gamma(X,A,\boverbar{M}_{K},Y)\wh\mu_1(X)-\gamma(X,A,\boverbar{M}_{K},Y)\mu_1(X)\big\}\Big] \nonumber\\ 
& + \frac{1}{n}\sum_{i=1}^{n} \varpi(R_i,Z_i,X_i,A_i,\boverbar{M}_{K,i},Y_i;\wt\gamma,\wh\mu) - \frac{1}{n}\sum_{i=1}^{n} \varpi(R_i,Z_i,X_i,A_i,\boverbar{M}_{K,i},Y_i;\gamma,\mu) \nonumber\\ 
& - E\Big[\varpi(R,Z,X,A,\boverbar{M}_{K},Y;\wt\gamma,\wh\mu) - \varpi(R,Z,X,A,\boverbar{M}_{K},Y;\gamma,\mu)\Big] \nonumber\\ 
=&~ \frac{1}{n}\sum_{i=1}^{n} R_i\{1+\wt\gamma(X_i,A_i,\boverbar{M}_{K,i},Y_i)\}\left\{1-I(A_i=a_1)\omega_1(X_i)\right\}\big\{\wh\mu_1(X_i)-\mu_1(X_i)\big\} \nonumber\\ 
&~ - E\big[R\{1+\wt\gamma(X,A,\boverbar{M}_{K},Y)\}\left\{1-I(A=a_1)\omega_1(X)\right\}\big\{\wh\mu_1(X)-\mu_1(X)\big\}\big] \label{eq:supp_stoc_ec_T1}\\ 
&~ +\sum_{k=2}^{K+1} \Bigg(
\frac{1}{n}\sum_{i=1}^{n} R_i\{1+\wt\gamma(X_i,A_i,\boverbar{M}_{K,i},Y_i)\}\Big\{\prod_{j=1}^{k-1} \omega_j(X_i,\boverbar{M}_{j-1,i})\Big\} \nonumber\\ 
&\qquad\quad \times\left\{I(A_i=a_{k-1})-I(A_i=a_k)\omega_k(X_i,\boverbar{M}_{k-1,i})\right\}\big\{\wh\mu_k(X_i,\boverbar{M}_{k-1})-\mu_k(X_i,\boverbar{M}_{k-1,i})\big\} \nonumber\\
&~ - E\Big[R\{1+\wt\gamma(X,A,\boverbar{M}_{K},Y)\}\Big\{\prod_{j=1}^{k-1} \omega_j(X,\boverbar{M}_{j-1})\Big\} \nonumber\\ 
&\qquad\quad \times\left\{I(A=a_{k-1})-I(A=a_k)\omega_k(X,\boverbar{M}_{k-1})\right\}\big\{\wh\mu_k(X,\boverbar{M}_{k-1})-\mu_k(X,\boverbar{M}_{k-1})\big\}\Big]
\Bigg ) \label{eq:supp_stoc_ec_T2}\\ 
&~ + \frac{1}{n}\sum_{i=1}^{n} R_i\{\wt\gamma(X_i,A_i,\boverbar{M}_{K,i},Y_i)-\gamma(X_i,A_i,\boverbar{M}_{K,i},Y_i)\} \{\phi(X_i,A_i,\boverbar{M}_{K,i},Y_i)-E(R\varrho_0\mid Z_i,A_i,\boverbar{M}_{K,i},Y_i)\} \nonumber\\ 
&~ - E\big[R\{\wt\gamma(X,A,\boverbar{M}_{K},Y)-\gamma(X,A,\boverbar{M}_{K},Y)\} \{\phi(X,A,\boverbar{M}_{K},Y)-E(R\varrho_0\mid Z,A,\boverbar{M}_{K},Y)\}\big].\label{eq:supp_stoc_ec_T3}
\end{align}
\end{small}
In what follows, we will show that \eqref{eq:supp_stoc_ec_T1}--\eqref{eq:supp_stoc_ec_T3} are all $o_p(n^{-1/2})$.
We consider \eqref{eq:supp_stoc_ec_T1}. By \Cref{ass:supp_sieve_mu}(ii), there exist ${\bm\pi}_{11}\in\mR^{t_{1n}}$ such that $\mu_{1n}(x)=\bbar{u}_1(x)\trans{\bm\pi}_{11}$ and 
\beqrs
\sup_{x\in\calX} \Abs{\mu_{1}(x)-\mu_{1n}(x)} = O_p(t_{1n}^{-\alpha_{31}}).
\eeqrs
Let $\epsilon=1-I(A=a_1)\omega_1(X)$, then we have 
\begin{align}
&\Big|\frac{1}{n}\sum_{i=1}^{n} R_i\{1+\wt\gamma(X_i,A_i,\boverbar{M}_{K,i},Y_i)\}\left\{1-I(A_i=a_1)\omega_1(X_i)\right\}\big\{\wh\mu_1(X_i)-\mu_1(X_i)\big\} \nonumber\\ 
&~ - E\big[R\{1+\wt\gamma(X,A,\boverbar{M}_{K},Y)\}\left\{1-I(A=a_1)\omega_1(X)\right\}\big\{\wh\mu_1(X)-\mu_1(X)\big\}\big]\Big|\nonumber\\ 
\lesssim &~ \sup_{x\in\calX} \Abs{\mu_{1}(x)-\mu_{1n}(x)}\nonumber\\ 
&~ + \Big\|\frac{1}{n}\sum_{i=1}^{n} \Big[R_i(1+\wt\gamma_i)\epsilon_i\bbar{u}_1(X_i)-E\{R(1+\wt\gamma)\epsilon \bbar{u}_1(X)\}\Big]\Big\|\cdot\big\|\wh{\bm\pi}_{11}-{\bm\pi}_{11}\big\|.\label{eq:supp_stoc_ec_T1.1}
\end{align}
Because $r\epsilon u_{1j}(x)(1+\wt\gamma)$ for $j=1,\ldots,t_{1n}$ is linear and hence Lipschitz continuous in $\wt\gamma\in\Gamma$ under $\|\cdot\|_\infty$, by \Cref{ass:supp_dgp} and Theorem 2.5.6, 2.7.1 and 2.7.11 of \cite{Vaart_ep_1996}, we conclude $\calG_j=\{r\epsilon u_{1j}(x)(1+\wt\gamma): \wt\gamma\in\Gamma\}$ is Donsker. It follows that 
\begin{align}
\Big\|\frac{1}{n}\sum_{i=1}^{n} \Big[R_i(1+\wt\gamma_i)\epsilon_i\bbar{u}_1(X_i)-E\{R(1+\wt\gamma)\epsilon \bbar{u}_1(X)\}\Big]\Big\| = O_p(\sqrt{t_{1n}/n}), \label{eq:supp_stoc_ec_T1.2}
\end{align}
uniformly over $\wt\gamma\in\Gamma$. By \Cref{ass:supp_sieve_strengthened_rate}, combining \eqref{eq:supp_pf_coeff_eta_01}, \eqref{eq:supp_stoc_ec_T1.1} and \eqref{eq:supp_stoc_ec_T1.2} yield that
\begin{align*}
\eqref{eq:supp_stoc_ec_T1} = O_p(t_{1n}^{-\alpha_{31}})+O_p(\sqrt{t_{1n}/n})O_p\big(\sqrt{\lambda_1}+\textstyle{\sum_{j=2}^{K+1}}t_{jn}^{-\alpha_{3j}}\big) = o_p(n^{-1/2}).
\end{align*}
Similarly, we can obtain $\eqref{eq:supp_stoc_ec_T2}=o_p(n^{-1/2})$. As for \eqref{eq:supp_stoc_ec_T3}, note that $r\wt\gamma\{\phi(x,a,\boverbar{m}_{K},y)-E(R\varrho_0\mid z,a,\boverbar{m}_{K},y)\}$ is linear and hence Lipschitz continuous in $\wt\gamma\in\Gamma$ under $\|\cdot\|_\infty$. Then by \Cref{ass:supp_dgp} and Theorem 2.5.6, 2.7.1 and 2.7.11 of \cite{Vaart_ep_1996}, we conclude $\calG=\{r\wt\gamma\{\phi(x,a,\boverbar{m}_{K},y)-E(R\varrho_0\mid z,a,\boverbar{m}_{K},y)\}: \wt\gamma\in\Gamma\}$ is a Donsker class. Thus the stochastic equicontinuity implies that $\eqref{eq:supp_stoc_ec_T3}=o_p(n^{-1/2})$ uniformly over $\{\wt\gamma\in\Gamma:\|\wt\gamma-\gamma\|_\infty=o_p(1)\}$. Then we complete the proof.
\end{proof}

\begin{lemma}\label{lemma:second_order_residual}
Suppose \Cref{ass:supp_dgp,ass:supp_sieve_strengthened_rate,ass:supp_sieve_mu} hold. We have
\begin{align*} 
&E\big[R\{\wt\gamma(X,A,\boverbar{M}_{K},Y)-\gamma(X,A,\boverbar{M}_{K},Y)\}\left\{1-I(A=a_1)\omega_1(X)\right\}\big\{\wh\mu_1(X)-\mu_1(X)\big\}\big] \\ 
&~ + \sum_{k=2}^{K+1}E\Big[R\{\wt\gamma(X,A,\boverbar{M}_{K},Y)-\gamma(X,A,\boverbar{M}_{K},Y)\}\Big\{\prod_{j=1}^{k-1} \omega_j(X,\boverbar{M}_{j-1})\Big\} \nonumber\\ 
&\quad \times\left\{I(A=a_{k-1})-I(A=a_k)\omega_k(X,\boverbar{M}_{k-1})\right\}\big\{\wh\mu_k(X,\boverbar{M}_{k-1})-\mu_k(X,\boverbar{M}_{k-1})\big\}\Big]\\
=&~ o_p(n^{-1/2}),
\end{align*}
uniformly over $\wt\gamma\in\{\gamma\in\Gamma:\|\wt\gamma-\gamma\|_\infty=o_p(1)\}$.

\end{lemma}

\begin{proof}
By \Cref{ass:supp_sieve_mu}(ii), there exist ${\bm\pi}_{11}\in\mR^{t_{1n}}$ such that $\mu_{1n}(x)=\bbar{u}_1(x)\trans{\bm\pi}_{11}$ and 
\beqrs
\sup_{x\in\calX} \Abs{\mu_{1}(x)-\mu_{1n}(x)} = O_p(t_{1n}^{-\alpha_{31}}).
\eeqrs
Let $\epsilon=1-I(A=a_1)\omega_1(X)$, then we have 
\begin{align}
&\Big| E\big[R\{\wt\gamma(X,A,\boverbar{M}_{K},Y)-\gamma(X,A,\boverbar{M}_{K},Y)\}\left\{1-I(A=a_1)\omega_1(X)\right\}\big\{\wh\mu_1(X)-\mu_1(X)\big\}\big]\Big|\nonumber\\ 
\lesssim &~ \sup_{x\in\calX} \Abs{\mu_{1}(x)-\mu_{1n}(x)}
+ \Big|E\{R(\wt\gamma-\gamma)\epsilon \bbar{u}_1(X)\trans\}\big\{\wh{\bm\pi}_{11}-{\bm\pi}_{11}\big\} \Big|. \label{eq:supp_lemma_second_order_residual_T1.1}
\end{align}
Let $\bm\Upsilon_1\triangleq E[I(A=a_1)R(1+\gamma) \bbar{u}_1(X)\bbar{u}_1(X)\trans]^{-1} E[R(\wt\gamma-\gamma)\epsilon \bbar{u}_1(X)]$. Then 
\beqrs
&& \sqrt{f_{A\mid X}(a_1\mid X)}\bbar{u}_1(x)\trans\bm\Upsilon \\ 
&&\qquad = E\left[\frac{E\{R(\wt\gamma-\gamma)\epsilon\mid X\}}{\sqrt{f_{A\mid X}(a_1\mid X)}} \sqrt{f_{A\mid X}(a_1\mid X)}\bbar{u}_1(X)\right] \\ 
&&\qquad\qquad \times E\left[f(a_1\mid X) \bbar{u}_1(X)\bbar{u}_1(X)\trans\right]^{-1}\sqrt{f_{A\mid X}(a_1\mid X)}\bbar{u}_1(x)
\eeqrs
is the $\calL_2$-projection of $\frac{E\{R(\wt\gamma-\gamma)\epsilon\mid X\}}{\sqrt{f_{A\mid X}(a_1\mid X)}}$ on the space linearly spanned by $\{\sqrt{f_{A\mid X}(a_1\mid X)} \bbar{u}_1(X)\}$.
Thus we obtain 
\beqr\label{eq:lemma:second_order_residual_pf2}
\sup_{x\in\calX} \Abs{\sqrt{f_{A\mid X}(a_1\mid X)}\bbar{u}_1(x)\trans\bm\Upsilon-\frac{E\{R(\wt\gamma-\gamma)\epsilon\mid X\}}{\sqrt{f_{A\mid X}(a_1\mid X)}} } =o(1).
\eeqr
Note that
\begin{subequations}
\begin{align}
& \Abs{E[R(\wt\gamma-\gamma) \epsilon \bbar{u}_1(X)\trans] (\wh{\bm\pi}_{11}-{\bm\pi}_{11})} \nonumber\\ 
\lesssim&~ \Abs{ \frac{1}{n}\sum_{i=1}^{n} I(A_i=a_1)R_i(1+\wh\gamma_i) \{\wh\mu_2(X_i,M_{1i})-\mu_{1n}(X_i)\}\bbar{u}_1(X_i)\trans\bm\Upsilon} \nonumber\\ 
\le&~ \Abs{ \frac{1}{n}\sum_{i=1}^{n} I(A_i=a_1)R_i(1+\gamma_i) \{\mu_{1}(X_i)-\mu_{1n}(X_i)\}\bbar{u}_1(X_i)\trans\bm\Upsilon } \label{eq:lemma:second_order_residual_pf3.1} \\
&+ \Abs{ \frac{1}{n}\sum_{i=1}^{n} I(A_i=a_1)R_i(1+\gamma_i) \{\mu_2(X_i,M_{1i})-\mu_{1}(X_i)\}\bbar{u}_1(X_i)\trans \bm\Upsilon} \label{eq:lemma:second_order_residual_pf3.2}\\
&+ \Abs{ \frac{1}{n}\sum_{i=1}^{n} I(A_i=a_1)R_i(\wh\gamma_i-\gamma_i) \{\mu_2(X_i,M_{1i})-\mu_{1}(X_i)\}\bbar{u}_1(X_i)\trans \bm\Upsilon} \label{eq:lemma:second_order_residual_pf3.3} \\ 
&+ \Abs{ \frac{1}{n}\sum_{i=1}^{n} I(A_i=a_1)R_i(\wh\gamma_i-\gamma_i) \{\mu_{1}(X_i)-\mu_{1n}(X_i)\}\bbar{u}_1(X_i)\trans \bm\Upsilon} \label{eq:lemma:second_order_residual_pf3.4}\\ 
&+ \Abs{ \frac{1}{n}\sum_{i=1}^{n} I(A_i=a_1)R_i(1+\wh\gamma_i) \{\wh\mu_2(X_i,M_{1i})-\mu_2(X_i,M_{1i})\}\bbar{u}_1(X_i)\trans\bm\Upsilon} . \label{eq:lemma:second_order_residual_pf3.5} 
\end{align}
\end{subequations}
\Cref{ass:supp_sieve_mu}(ii), \ref{ass:supp_sieve_strengthened_rate}(ii) and \eqref{eq:lemma:second_order_residual_pf2} implies that 
\beqrs
\eqref{eq:lemma:second_order_residual_pf3.1} \lesssim \sup_{(m,x)\in\calM\times\calX} \Abs{\mu_1(x) - \mu_{1n}(x) }\cdot \|\wt\gamma-\gamma\|_\infty = O_p(t_{1n}^{-\alpha_{31}})o_p(1)=o_p(n^{-1/2}),
\eeqrs
uniformly over $\wt\gamma\in\{\gamma\in\Gamma:\|\wt\gamma-\gamma\|_\infty=o_p(1)\}$. Note that \eqref{eq:lemma:second_order_residual_pf2} implies that 
\begin{align*}
& E|\eqref{eq:lemma:second_order_residual_pf3.2}|^2 \\
&~~= n^{-1} E\left[\left\{I(A=a_1)R(1+\gamma) \{\mu_2(X,M_{1})-\mu_{1}(X)\}\bbar{u}_1(X)\trans \bm\Upsilon\right\}^2\right] \\
&~~= n^{-1} E\left[\left\{ \frac{I(A=a_1)R(1+\gamma) \{\mu_2(X,M_{1})-\mu_{1}(X)\}}{\sqrt{f_{A\mid X}(a_1\mid X)}}\cdot \frac{E\{R(\wt\gamma-\gamma)\epsilon\mid X\}}{\sqrt{f_{A\mid X}(a_1\mid X)}}\right\}^2\right]\left\{1+o_p(1)\right\} \\ 
&~~\lesssim n^{-1} E\left[ \frac{E\left[I(A=a_1)R(1+\gamma)^2 \{\mu_2(X,M_{1})-\mu_{1}(X)\}^2\mid X\right]}{f_{A\mid X}(a_1\mid X)^2} R(\wt\gamma-\gamma)^2\epsilon^2\right] \\ 
&~~\lesssim n^{-1} \|\wt\gamma-\gamma\|_\infty^2 =o(n^{-1}),
\end{align*}
uniformly over $\wt\gamma\in\{\gamma\in\Gamma:\|\wt\gamma-\gamma\|_\infty=o_p(1)\}$. Then the Chebyshev inequality implies that 
$
\eqref{eq:lemma:second_order_residual_pf3.2} = o_p(n^{-1/2}).
$
Similar to \eqref{eq:supp_pf_coeff_gamma_T2}, we can obtain
$
\eqref{eq:lemma:second_order_residual_pf3.3} = o_p(n^{-1/2})
$ under \Cref{ass:supp_dgp,ass:supp_sieve_delta_inter,ass:supp_sieve_delta_exter,ass:supp_technique_mu}.
Similarly to \eqref{eq:supp_pf_coeff_gamma_T3}, we obtain $
\eqref{eq:lemma:second_order_residual_pf3.4} = o_p(t_{1n}^{-\alpha_{31}}) = o_p(n^{-1/2})
$ by \Cref{ass:supp_sieve_strengthened_rate}.
As for \eqref{eq:lemma:second_order_residual_pf3.5}, it can be calculated similar to the last term of \eqref{eq:supp_pf_coeff_eta_1}. Then we obtain 
\begin{align*}
\eqref{eq:lemma:second_order_residual_pf3.5}\lesssim
\|\wh\gamma-\gamma\|_{\infty}\cdot O_p(n^{-1/2}+\textstyle{\sum_{j=2}^{K+1}}t_{jn}^{-\alpha_{3j}})=o_p(n^{-1/2}),
\end{align*}
uniformly over $\wt\gamma \in \{\wt\gamma\in\Gamma:\|\wt\gamma-\gamma\|_\infty= o_p(1)\}$.
Combing above arguments implies that
\beqr
\Big|E\{R(\wt\gamma-\gamma)\epsilon \bbar{u}_1(X)\trans\}\big\{\wh{\bm\pi}_{11}-{\bm\pi}_{11}\big\} \Big|=o_p(n^{-1/2}). \label{eq:supp_lemma_second_order_residual_T1.2}
\eeqr
uniformly over $\wt\gamma\in\{\gamma\in\Gamma:\|\wt\gamma-\gamma\|_\infty=o_p(1)\}$. Combining \eqref{eq:supp_lemma_second_order_residual_T1.1} and \eqref{eq:supp_lemma_second_order_residual_T1.2} yields that 
\begin{align*} 
E\big[R\{\wt\gamma(X,A,\boverbar{M}_{K},Y)-\gamma(X,A,\boverbar{M}_{K},Y)\}\left\{1-I(A=a_1)\omega_1(X)\right\}\big\{\wh\mu_1(X)-\mu_1(X)\big\}\big] = o_p(n^{-1/2}).
\end{align*}
uniformly over $\wt\gamma\in\{\gamma\in\Gamma:\|\wt\gamma-\gamma\|_\infty=o_p(1)\}$. Similarly we can obtain
\begin{align*} 
& E\Big[R\{\wt\gamma(X,A,\boverbar{M}_{K},Y)-\gamma(X,A,\boverbar{M}_{K},Y)\}\Big\{\prod_{j=1}^{k-1} \omega_j(X,\boverbar{M}_{j-1})\Big\} \nonumber\\ 
&\quad \times\left\{I(A=a_{k-1})-I(A=a_k)\omega_k(X,\boverbar{M}_{k-1})\right\}\big\{\wh\mu_k(X,\boverbar{M}_{k-1})-\mu_k(X,\boverbar{M}_{k-1})\big\}\Big]= o_p(n^{-1/2}).
\end{align*}
for $k\in\{2,\ldots,K+1\}$.
Then we complete the proof of the lemma.
\end{proof}

\titleformat{\section}{\normalfont\Large\bfseries}{\thesection}{1em}{}
\addcontentsline{toc}{section}{References}
\putbib
\end{bibunit}

\begin{thebibliography}{}

\bibitem[Ai and Chen, 2003]{AiChen2003}
Ai, C. and Chen, X. (2003).
\newblock Efficient estimation of models with conditional moment restrictions
  containing unknown functions.
\newblock {\em Econometrica}, 71(6):1795--1843.

\bibitem[Avin et~al., 2005]{AvinShpitserPearl2005}
Avin, C., Shpitser, I., and Pearl, J. (2005).
\newblock Identifiability of path-specific effects.
\newblock In {\em Proceedings of the 19th International Joint Conference on
  {{Artificial}} Intelligence}, {{IJCAI}}'05, pages 357--363, San Francisco,
  CA, USA. Morgan Kaufmann Publishers Inc.

\bibitem[Baron and Kenny, 1986]{BaronKenny1986}
Baron, R.~M. and Kenny, D.~A. (1986).
\newblock The moderator--mediator variable distinction in social psychological
  research: {{Conceptual}}, strategic, and statistical considerations.
\newblock {\em Journal of Personality and Social Psychology}, 51(6):1173--1182.

\bibitem[Bartlett et~al., 2014]{bartlett2014improving}
Bartlett, J.~W., Carpenter, J.~R., Tilling, K., and Vansteelandt, S. (2014).
\newblock Improving upon the efficiency of complete case analysis when
  covariates are mnar.
\newblock {\em Biostatistics}, 15(4):719--730.

\bibitem[Chen, 2007]{Chen2007Handbook}
Chen, X. (2007).
\newblock Large sample sieve estimation of semi-nonparametric models.
\newblock In Heckman, J.~J. and Leamer, E.~E., editors, {\em Handbook of
  Econometrics}, volume~6, pages 5549--5632. {Elsevier}.

\bibitem[Daniel et~al., 2015]{daniel2015causal}
Daniel, R.~M., De~Stavola, B.~L., Cousens, S.~N., and Vansteelandt, S. (2015).
\newblock Causal mediation analysis with multiple mediators.
\newblock {\em Biometrics}, 71(1):1--14.

\bibitem[D'Haultfoeuille, 2010]{dHaultfoeuille2010}
D'Haultfoeuille, X. (2010).
\newblock A new instrumental method for dealing with endogenous selection.
\newblock {\em Journal of Econometrics}, 154(1):1--15.

\bibitem[D'Haultfoeuille, 2011]{DHaultfoeuille2011}
D'Haultfoeuille, X. (2011).
\newblock On the completeness condition in nonparametric instrumental problems.
\newblock {\em Econometric Theory}, 27(3):460--471.

\bibitem[Ding and Geng, 2014]{DingGeng2014}
Ding, P. and Geng, Z. (2014).
\newblock Identifiability of subgroup causal effects in randomized experiments
  with nonignorable missing covariates.
\newblock {\em Statistics in Medicine}, 33(7):1121--1133.

\bibitem[Grenander, 1981]{Grenander1981}
Grenander, U. (1981).
\newblock {\em Abstract Inference}.
\newblock Wiley Series in Probability and Mathematical Statistics. {Wiley},
  {New York, NY}.

\bibitem[Imai et~al., 2010a]{ImaiKeeleTingley2010}
Imai, K., Keele, L., and Tingley, D. (2010a).
\newblock A general approach to causal mediation analysis.
\newblock {\em Psychological Methods}, 15(4):309--334.

\bibitem[Imai et~al., 2010b]{ImaiKeeleYamamoto2010}
Imai, K., Keele, L., and Yamamoto, T. (2010b).
\newblock Identification, inference and sensitivity analysis for causal
  mediation effects.
\newblock {\em Statistical Science}, 25(1):51--71.

\bibitem[Kress, 1989]{Kress_LinearIntegralEquations}
Kress, R. (1989).
\newblock {\em Linear Integral Equations}, volume~82 of {\em Applied
  {{Mathematical Sciences}}}.
\newblock {Springer New York}, {New York, NY}.

\bibitem[Li et~al., 2023]{LiMiaoTchetgenTchetgen2023}
Li, W., Miao, W., and Tchetgen~Tchetgen, E. (2023).
\newblock Non-parametric inference about mean functionals of non-ignorable
  non-response data without identifying the joint distribution.
\newblock {\em Journal of the Royal Statistical Society Series B: Statistical
  Methodology}, 85(3):913--935.

\bibitem[Li and Zhou, 2017]{LiZhou2017}
Li, W. and Zhou, X.-H. (2017).
\newblock Identifiability and estimation of causal mediation effects with
  missing data.
\newblock {\em Statistics in Medicine}, 36(25):3948--3965.

\bibitem[Little and Rubin, 2002]{LittleRubin2002}
Little, R. J.~A. and Rubin, D.~B. (2002).
\newblock {\em Statistical Analysis with Missing Data}.
\newblock Wiley {{Series}} in {{Probability}} and {{Statistics}}. {Wiley}, 1st
  edition.

\bibitem[Miao et~al., 2016]{MiaoDingGeng2016}
Miao, W., Ding, P., and Geng, Z. (2016).
\newblock Identifiability of normal and normal mixture models with nonignorable
  missing data.
\newblock {\em Journal of the American Statistical Association},
  111(516):1673--1683.

\bibitem[Miao et~al., 2024]{MiaoLiuLiTchetgenTchetgenGeng2024}
Miao, W., Liu, L., Li, Y., Tchetgen~Tchetgen, E.~J., and Geng, Z. (2024).
\newblock Identification and semiparametric efficiency theory of nonignorable
  missing data with a shadow variable.
\newblock {\em ACM/IMS Journal of Data Science}, 1(2):1--23.

\bibitem[Miao and Tchetgen~Tchetgen, 2018]{MiaoTchetgenTchetgen2018}
Miao, W. and Tchetgen~Tchetgen, E. (2018).
\newblock Identification and inference with nonignorable missing covariate
  data.
\newblock {\em Statistica Sinica}, 28(4):2049--2067.

\bibitem[Miao and Tchetgen~Tchetgen, 2016]{MiaoTchetgenTchetgen2016}
Miao, W. and Tchetgen~Tchetgen, E.~J. (2016).
\newblock On varieties of doubly robust estimators under missingness not at
  random with a shadow variable.
\newblock {\em Biometrika}, 103(2):475--482.

\bibitem[Miles et~al., 2017]{miles2017quantifying}
Miles, C.~H., Shpitser, I., Kanki, P., Meloni, S., and Tchetgen~Tchetgen, E.~J.
  (2017).
\newblock Quantifying an adherence path-specific effect of antiretroviral
  therapy in the nigeria pepfar program.
\newblock {\em Journal of the American Statistical Association},
  112(520):1443--1452.

\bibitem[Miles et~al., 2020]{miles2020semiparametric}
Miles, C.~H., Shpitser, I., Kanki, P., Meloni, S., and Tchetgen~Tchetgen, E.~J.
  (2020).
\newblock On semiparametric estimation of a path-specific effect in the
  presence of mediator-outcome confounding.
\newblock {\em Biometrika}, 107(1):159--172.

\bibitem[Newey, 1997]{Newey1997}
Newey, W.~K. (1997).
\newblock Convergence rates and asymptotic normality for series estimators.
\newblock {\em Journal of Econometrics}, 79(1):147--168.

\bibitem[Newey and Powell, 2003]{NeweyPowell2003}
Newey, W.~K. and Powell, J.~L. (2003).
\newblock Instrumental variable estimation of nonparametric models.
\newblock {\em Econometrica}, 71(5):1565--1578.

\bibitem[Pearl, 2001]{Pearl2001}
Pearl, J. (2001).
\newblock Direct and indirect effects.
\newblock In {\em Proceedings of the {{Seventeenth}} Conference on
  {{Uncertainty}} in Artificial Intelligence}, {{UAI}}'01, pages 411--420, {San
  Francisco, CA, USA}. {Morgan Kaufmann Publishers Inc.}

\bibitem[Robins and Greenland, 1992]{RobinsGreenland1992}
Robins, J.~M. and Greenland, S. (1992).
\newblock Identifiability and exchangeability for direct and indirect effects.
\newblock {\em Epidemiology}, 3(2):143--155.

\bibitem[Robins and Richardson, 2011]{RobinsRichardson2011}
Robins, J.~M. and Richardson, T.~S. (2011).
\newblock {\em Alternative Graphical Causal Models and the Identification of
  Direct Effects}.
\newblock Oxford University Press, Oxford.

\bibitem[Rubin, 1976]{Rubin1976}
Rubin, D.~B. (1976).
\newblock Inference and missing data.
\newblock {\em Biometrika}, 63(3):581--592.

\bibitem[Santos, 2011]{Santos2011}
Santos, A. (2011).
\newblock Instrumental variable methods for recovering continuous linear
  functionals.
\newblock {\em Journal of Econometrics}, 161(2):129--146.

\bibitem[Severini and Tripathi, 2012]{SeveriniTripathi2012}
Severini, T.~A. and Tripathi, G. (2012).
\newblock Efficiency bounds for estimating linear functionals of nonparametric
  regression models with endogenous regressors.
\newblock {\em Journal of Econometrics}, 170(2):491--498.

\bibitem[Shan et~al., 2024]{ShanLiAi2024}
Shan, J., Li, W., and Ai, C. (2024).
\newblock Efficient nonparametric inference of causal mediation effects with
  nonignorable missing confounders.
\newblock {\em arXiv preprint arXiv:2402.05384}.

\bibitem[Sharif et~al., 2019]{sharif2019mediation}
Sharif, S., Groenwold, R.~H., van~der Graaf, Y., Berkelmans, G.~F., Cramer,
  M.~J., Visseren, F.~L., Westerink, J., study group, S., van Petersen, R.,
  Dinther, B., et~al. (2019).
\newblock Mediation analysis of the relationship between type 2 diabetes and
  cardiovascular events and all-cause mortality: Findings from the smart
  cohort.
\newblock {\em Diabetes, Obesity and Metabolism}, 21(8):1935--1943.

\bibitem[Shpitser, 2013]{Shpitser2013}
Shpitser, I. (2013).
\newblock Counterfactual graphical models for longitudinal mediation analysis
  with unobserved confounding.
\newblock {\em Cognitive Science}, 37(6):1011--1035.

\bibitem[Tchetgen~Tchetgen and Shpitser, 2012]{TchetgenTchetgenShpitser2012}
Tchetgen~Tchetgen, E. and Shpitser, I. (2012).
\newblock Semiparametric theory for causal mediation analysis: {{Efficiency}}
  bounds, multiple robustness and sensitivity analysis.
\newblock {\em The Annals of Statistics}, 40(3):1816--1845.

\bibitem[Vanderweele et~al., 2014]{VanderweeleVansteelandtRobins2014}
Vanderweele, T.~J., Vansteelandt, S., and Robins, J.~M. (2014).
\newblock Effect decomposition in the presence of an exposure-induced
  mediator-outcome confounder.
\newblock {\em Epidemiology}, 25(2):300--306.

\bibitem[Vekic et~al., 2019]{vekic2019obesity}
Vekic, J., Zeljkovic, A., Stefanovic, A., Jelic-Ivanovic, Z., and
  Spasojevic-Kalimanovska, V. (2019).
\newblock Obesity and dyslipidemia.
\newblock {\em Metabolism}, 92:71--81.

\bibitem[Wang et~al., 2022]{wang2022role}
Wang, L., Sun, M., Guo, Y., Yan, S., Li, X., Wang, X., Hu, W., Yang, Y., Li,
  J., and Li, B. (2022).
\newblock The role of dietary inflammatory index on the association between
  sleep quality and long-term cardiovascular risk: a mediation analysis based
  on {{NHANES}} (2005--2008).
\newblock {\em Nature and Science of Sleep}, 14:483--492.

\bibitem[Wang et~al., 2014]{WangShaoKim2014}
Wang, S., Shao, J., and Kim, J.~K. (2014).
\newblock An instrumental variable approach for identification and estimation
  with nonignorable nonresponse.
\newblock {\em Statistica Sinica}, 24(3):1097--1116.

\bibitem[Yang et~al., 2019]{YangWangDing2019}
Yang, S., Wang, L., and Ding, P. (2019).
\newblock Causal inference with confounders missing not at random.
\newblock {\em Biometrika}, 106(4):875--888.

\bibitem[Zhang et~al., 2023]{zhang2023triglyceride}
Zhang, Q., Xiao, S., Jiao, X., and Shen, Y. (2023).
\newblock The triglyceride-glucose index is a predictor for cardiovascular and
  all-cause mortality in cvd patients with diabetes or pre-diabetes: evidence
  from {{NHANES}} 2001--2018.
\newblock {\em Cardiovascular Diabetology}, 22(1):279.

\bibitem[Zhao and Ma, 2022]{ZhaoMa2022}
Zhao, J. and Ma, Y. (2022).
\newblock A versatile estimation procedure without estimating the nonignorable
  missingness mechanism.
\newblock {\em Journal of the American Statistical Association},
  117(540):1916--1930.

\bibitem[Zhou, 2022]{zhou2022semiparametric}
Zhou, X. (2022).
\newblock Semiparametric estimation for causal mediation analysis with multiple
  causally ordered mediators.
\newblock {\em Journal of the Royal Statistical Society Series B: Statistical
  Methodology}, 84(3):794--821.

\bibitem[Zuo et~al., 2024]{ZuoGhoshDingYang2024}
Zuo, S., Ghosh, D., Ding, P., and Yang, F. (2024).
\newblock Mediation analysis with the mediator and outcome missing not at
  random.
\newblock {\em Journal of the American Statistical Association}, 0(0):1--21.

\end{thebibliography}


\begin{thebibliography}{}

\bibitem[Ai and Chen, 2007]{AiChen2007}
Ai, C. and Chen, X. (2007).
\newblock Estimation of possibly misspecified semiparametric conditional moment
  restriction models with different conditioning variables.
\newblock {\em Journal of Econometrics}, 141(1):5--43.

\bibitem[Ai and Chen, 2012]{AiChen2012}
Ai, C. and Chen, X. (2012).
\newblock The semiparametric efficiency bound for models of sequential moment
  restrictions containing unknown functions.
\newblock {\em Journal of Econometrics}, 170(2):442--457.

\bibitem[Chen, 2007]{Chen2007Handbook}
Chen, X. (2007).
\newblock Large sample sieve estimation of semi-nonparametric models.
\newblock In Heckman, J.~J. and Leamer, E.~E., editors, {\em Handbook of
  Econometrics}, volume~6, pages 5549--5632. {Elsevier}.

\bibitem[Chen and Pouzo, 2012]{ChenPouzo2012}
Chen, X. and Pouzo, D. (2012).
\newblock Estimation of nonparametric conditional moment models with possibly
  nonsmooth generalized residuals.
\newblock {\em Econometrica}, 80(1):277--321.

\bibitem[Hirano et~al., 2003]{HiranoImbensRidder2003}
Hirano, K., Imbens, G.~W., and Ridder, G. (2003).
\newblock Efficient estimation of average treatment effects using the estimated
  propensity score.
\newblock {\em Econometrica}, 71(4):1161--1189.

\bibitem[Newey, 1997]{Newey1997}
Newey, W.~K. (1997).
\newblock Convergence rates and asymptotic normality for series estimators.
\newblock {\em Journal of Econometrics}, 79(1):147--168.

\bibitem[Newey and Powell, 2003]{NeweyPowell2003}
Newey, W.~K. and Powell, J.~L. (2003).
\newblock Instrumental variable estimation of nonparametric models.
\newblock {\em Econometrica}, 71(5):1565--1578.

\bibitem[Severini and Tripathi, 2012]{SeveriniTripathi2012}
Severini, T.~A. and Tripathi, G. (2012).
\newblock Efficiency bounds for estimating linear functionals of nonparametric
  regression models with endogenous regressors.
\newblock {\em Journal of Econometrics}, 170(2):491--498.

\bibitem[van~der Vaart and Wellner, 1996]{Vaart_ep_1996}
van~der Vaart, {\relax AW}. and Wellner, J. (1996).
\newblock {\em Weak Convergence and Empirical Processes: With Applications to
  Statistics}.
\newblock {Springer Science \& Business Media}.

\end{thebibliography}
\end{document}